\documentclass[11pt]{article}
\usepackage[margin=.75in]{geometry}

\usepackage{constants}
\newconstantfamily{c}{symbol=c}
\usepackage{preamble}
\usepackage{braket}
\usepackage{xcolor}
\usepackage[export]{adjustbox}

\makeatletter
\DeclareRobustCommand{\cev}[1]{%
  {\mathpalette\do@cev{#1}}%
}
\newcommand{\do@cev}[2]{%
  \vbox{\offinterlineskip
    \sbox\z@{$\m@th#1 x$}%
    \ialign{##\cr
      \hidewidth\reflectbox{$\m@th#1\vec{}\mkern4mu$}\hidewidth\cr
      \noalign{\kern-\ht\z@}
      $\m@th#1#2$\cr
    }%
  }%
}
\makeatother

\title{Computational Dynamical Systems}
\author{Jordan Cotler\thanks{Email: \texttt{jcotler@fas.harvard.edu}%. This work is supported by
} \\
Harvard University
\and
Semon Rezchikov\thanks{Email: \texttt{semonr@princeton.edu}%. This work is supported by...
}\\ Princeton University}

\DeclareMathAlphabet\mathbfcal{OMS}{cmsy}{b}{n}

\renewcommand{\epsilon}{\varepsilon}

\newcommand{\machine}[1]{\textnormal{\textsf{#1}}}

\usepackage{multirow}

\begin{document}

\maketitle

\begin{abstract}
We study the computational complexity theory of smooth, finite-dimensional dynamical systems. Building off of previous work, we give definitions for what it means for a smooth dynamical system to simulate a Turing machine. We then show that `chaotic' dynamical systems (more precisely, Axiom A systems) and `integrable' dynamical systems (more generally, measure-preserving systems) cannot robustly simulate universal Turing machines, although such machines can be robustly simulated by other kinds of dynamical systems.  Subsequently, we show that any Turing machine that can be encoded into a structurally stable one-dimensional dynamical system must have a decidable halting problem, and moreover an explicit time complexity bound in instances where it does halt.  More broadly, our work elucidates what it means for one `machine' to simulate another, and emphasizes the necessity of defining low-complexity `encoders' and `decoders' to translate between the dynamics of the simulation and the system being simulated. We highlight how the notion of a computational dynamical system leads to questions at the intersection of computational complexity theory, dynamical systems theory, and real algebraic geometry.
%and combinatorics
\end{abstract}

\thispagestyle{empty}
\clearpage
\tableofcontents
\addtocontents{toc}{\protect\thispagestyle{empty}}
\pagenumbering{gobble}

\clearpage
% \clearpage
\pagenumbering{arabic} 

\section{Introduction}

\subsection{Motivation and overview}

Models of digital computation, which lie at the foundation of computer science, are typically discrete, while most of our fundamental models of the physical world are essentially continuous. Nonetheless,  the Church-Turing
thesis~\cite{turing1939systems} and its physical counterparts \cite{gandy_physical_church_turing, copeland2007physical} state that this difference is  illusory: the discrete computations we can perform reliably in the physical world should be the same as those which can be performed by a Turing machine, possibly by one having access to random bits. The validity of the physical Church-Turing thesis is a subject of debate, and a number of variants of the thesis have been proposed~\cite{copeland1997church}.
%there are arguments \cite{nielsen1997computable, kieu2003quantum} that it is consistent with our understanding of quantum mechanics (which is based fundamentally on the complex numbers, a continuous object) that nature may implement physical systems which compute Turing-uncomputable functions, although no experimental examples are known. 
Furthermore, from the perspective of complexity theory rather than computatibility theory,  the possibility for quantum computers to solve with high probability, in polynomial time, decision problems which are not in \textbf{P}, is a basic motivation for research on quantum computation \cite{nielsen2010quantum, Aharonov:2021das}. 

In a different (non-quantum) direction, there have been multiple models proposed for a definition of a computable real function \cite{grzegorczyk1955computable, lacombe1959classes, blum1998complexity, smale1997complexity, braverman2005complexity},  and using this language, it has been found that simple finite-dimensional continuous dynamical systems defined by polynomial equations with integral coefficients can exhibit non-computable dynamical properties \cite{moore1990unpredictability, braverman2006non}. In general it is known that the existence of natural problems with no computable solution (such as the problem of recognizing presentations of the trivial group \cite{psalgorithmic}) forces complex behaviour of various continuous mathematical objects related to geometry and dynamics \cite{weinberger2020computers, seidel2008biased}. In yet a different direction, there has been a sequence of papers asking whether universal computation can be realized by various ordinary \cite{branicky1994analog} and partial differential equations, including in single-particle potential energy systems \cite{tao2017universality} and in solutions to fluid dynamics equations \cite{cardona2021constructing}; this was in part motivated by the hope of showing the existence of blow-up solutions to the Navier-Stokes equations by finding fluid flows which `replicate themselves' at smaller and smaller scales \cite{tao2016finite}. Such works on realizing universal computation in natural continuous physical models can be seen as a continuation of Moore's earlier work \cite{moore1998finite, moore1990unpredictability}, which realized universal computation in a simple $2$-dimensional piecewise-linear map, as well as in a Lipschitz map on the interval and an analytic map on $\R$. The relation between the computational capacity and the analytic or dynamical properties of a continuous dynamical system, such as its topological entropy or its regularity, are known to be subtle: for example, depending on the formalization, the topological entropy of a Turing-universal system can be zero \cite{cardona2023computability} or can be forced to be nonnegative \cite{bruera2024topological}. 

Researchers in both machine learning and computational neuroscience are often forced to posit that various continuous systems (recurrent neural networks, transformers, models of brains) implement certain computations \cite{sussillo2014neural, chariker2016orientation, maass2019brain, kato2015global, khona2022attractor, cotler2023analyzing}, and indeed part of the problem of neuroscience is to extract, from neuronal measurements, the `computations' implemented by the brain. Such scientific applications were part of the motivation for the founders of the mathematical field of differentiable dynamical systems theory \cite{smale1967differentiable, THOM1969313}, who originally tried to extract simple `discrete' descriptions of the dynamics of systems like Axiom A systems. The development of differentiable dynamical systems theory led to a collection of powerful mathematical methods for understanding dynamical systems. Nonetheless, these methods are rarely used by non-mathematicians, perhaps because they do not connect straightforwardly with the kinds of questions the researchers in neural network interpretability and computational neuroscience typically ask.

To `do' complexity theory with continuous dynamical systems, one must know what it means for a continuous system to `implement' a discrete computation. Unfortunately, this notion is not completely clear and many of the works cited above do not give precise definitions for this notion. Here we propose an answer to this question via a perspective which is natural to computer scientists. We will show that without some definition like the one we propose, computability questions about continuous dynamical systems become trivialized. Our proposed definition differs from other proposals connected to the Space-Bounded Church-Turing Thesis \cite{braverman2015space, braverman2012noise} by not requiring for `noise' to be introduced in the continuous dynamics. Moreover, our proposed framework naturally leads to interesting questions that should feel familiar to those interested in the mathematical study of differentiable dynamical systems. In particular, we prove several results regarding complexity and computability theory in the context of our framework by utilizing results from differentiable dynamical systems theory.

To explain our proposal, let us recall how computer scientists ask computational complexity questions about \emph{discrete} systems. Given some machine $\machine{T}: S \to S$ with discrete configuration space $S$ (where $x \in S$ describes the configuration of a machine including all of its tapes at a given moment), we would say that $\machine{T}$ is \emph{Turing universal} if it can do the same computations as any fixed universal Turing machine $\machine{T}_{\!\text{univ}}$. To establish this property of $\machine{T}$, we always need to find some \emph{encoder} $\mathcal{E}$ which lets us encode configurations of $\machine{T}_{\!\text{univ}}$ into configurations of $\machine{T}$, and some \emph{decoder} $\mathcal{D}$ which lets us decode a configuration of $\machine{T}_{\!\text{univ}}$ from a configuration of $\machine{T}$. In fact, since there are many definitions of a Turing machine (e.g.~with one-sided or two-sided tapes, with multiple tapes, as well as more exotic variants), such constructions are needed when setting up the theory of computation; many implicit examples of such encoders and decoders can be found in basic textbooks like~\cite{sipser2012introduction, arora2009computational}. For such constructions to make sense, one must require that \emph{the encoder and decoder are themselves computationally simple}, e.g.~that they can be implemented by a low-time complexity Turing machine or a uniform family of low-depth circuits. Otherwise, one can package all the computation into the encoder and decoder themselves (see Section \ref{subsubsec:counterexample1}); thus, the notion of a Turing-universal system already presupposes the existence of a basic theory of computational complexity to constrain the encoder and decoder.

Luckily, in the continuous domain there is already a well-developed theory of computable real functions and uniform real circuits \cite{blum1998complexity, braverman2005complexity}. Thus, to ask if a continuous system $f: M \to M$ (where $M$ is continuous, e.g.~$M = [0,1]^n$) is Turing-universal, we can require for there to be low-time complexity $\mathbb{R}$--Turing machines $\mathcal{E}, \mathcal{D}$ which respectively encode bit strings from the configuration space of the Turing machine into the domain of $f$, and decode regions in the domain of $f$ to e.g.~bit strings from the configuration space of the Turing machine.

\begin{definition}[informal; see Definition \ref{def:CDS1}]
    A \textbf{computational dynamical system} (or \textbf{CDS}) is a tuple $(f, \mathcal{E}, \mathcal{D}, \tau, \machine{T})$ where:
    \begin{itemize}
        \item $f: M \to M$ is a dynamical system, with $M \subset \R^k$;
        \item $\mathcal{E}: S \to M$ is a function which can be implemented by a \text{\rm BSS$_\textsf{C}$} machine (a model of real computation; see Definition \ref{def:decoders} and Appendix \ref{App:BSSC}) that runs in time $O(t(n))$ for some function $t(n)$ of the length of the \emph{input}; 
        \item $\mathcal{D}: M \rightharpoonup S$ is a partially defined function (i.e.~a function $M \to S \sqcup \{\textnormal{\textsf{Error}}\})$ which can be implemented by a \text{\rm BSS$_\textsf{C}$} machine that runs in time $O(t(n))$, for the same function $t(n)$ of the length of the \emph{output} of $\mathcal{D}$;
        \item $\tau: M \to \mathbb{Z}_{\geq 0}$ is a function which is constant on connected components of $\mathcal{D}^{-1}(s)$; and
        \item $\machine{T}: S \to S$  is a discrete computational system, e.g.~a Turing machine (although variants can be defined e.g.~for pushdown automata).
    \end{itemize}
    This tuple is required to satisfy the condition that $\mathcal{D} \circ \mathcal{E} = \text{\rm Id}$, as well as the condition that for $s \in S$, we have
    \[ \mathcal{D} \circ f^{\tau} \circ \mathcal{D}^{-1}(s) = \machine{T}(s)\,. \]
    Here $f^\tau : M \to M$ where $f^\tau : x \mapsto f^{\tau(x)}(x)$. Thus the encoder-decoder pair along with $f$ can simulate $\machine{T}$ with a slowdown determined by $\tau$ and $t(n)$. 
\end{definition}

\begin{remark}
    An equivalent condition to the above is that $f^\tau$ takes all of $\mathcal{D}^{-1}(s)$ into $\mathcal{D}^{-1}(s)$. If $\mathcal{D}^{-1}(s) = \{\mathcal{E}(s)\}$, then this is equivalent to the condition that $\mathcal{D}\circ f^\tau \circ \mathcal{E} = \machine{T}$.  However, for general $\mathcal{D}$ for which $\mathcal{D}^{-1}(s)$ may have non-empty interior, this notion corresponds to requiring that the computation be \emph{robust}: any perturbation of the ideal input $\mathcal{E}(s)$ corresponding to $s$, which still lies in a `validity region' $\mathcal{D}^{-1}(s)$, will continue to compute the correct answer. Thus, the above definition encapsulates a model of computation which is \emph{robust to non-uniform errors}, i.e.~the amount of error allowed may depend on the input and may go to zero in certain regions of the domain of $f$. It is the non-uniformity of the allowed amount of error that enables universal computation on compact domains (see the discussion regarding robustness in Section \ref{sec:related-work}). 
\end{remark}

Essentially all previous work on computational properties of continuous dynamical systems (e.g.~\cite{moore1998finite}) can be put into this framework; our definition gives a precise notion of a `reasonable' encoding of states of $\machine{T}$ into $M$, \emph{which has thus far been without a definition in the literature}. Without such a condition on $\mathcal{D}$ and $\mathcal{E}$, the notion of simulation becomes essentially trivial (Example \ref{example:uncon1}), just as in the case of Turing machines. 

However, with bounds on $\tau(x)$ and $t(n)$ (e.g.~$\tau(x) = \mathcal{O}(|\mathcal{D}^{-1}(x)|)$ and $t(n) = O(n)$) the notion of simulation is nontrivial, and we can say that a continuous dynamical system $f$ is \emph{Turing-universal} when there \emph{exists} a CDS $(f, \mathcal{E}, \mathcal{D}, \tau,  \machine{T}_{\!\text{univ}})$ for some universal Turing machine $\machine{T}_{\!\text{univ}}$. Thus universality is an \emph{intrinsic} property of $f$; we will see that dynamical systems $f$ which in our sense are both `robust' and universal exist, even though many natural dynamical conditions on $f$ will be shown to preclude universality. 

It is natural to generalize CDSs to the setting of forced dynamical systems, as would be appropriate for studying e.g.~finite state machines, RNNs, transformers, etc.  We develop the theory of forced CDSs in~\cite{cotlerrezchikovWIP}.

\subsection{Our results}
\label{sec:our-results}

One of the benefits of our notion of a CDS is that it is straightforward to define various conditions for the decoder $\mathcal{D}$ in order to model different ways of encoding.  For instance:

\begin{definition}
    Let $(f, \mathcal{E}, \mathcal{D}, \tau, \machine{T})$ be a CDS. We say that the decoder $\mathcal{D}$, as well as the CDS, is \textbf{robust} if for every $s \in S$ (where $S$ is the configuration space of $\machine{T}$) we have that $\mathcal{D}^{-1}(s)$ is the closure of its interior, and $\mathcal{E}(s)$ lies in the interior of $\mathcal{D}^{-1}(s)$. 
\end{definition}

\noindent Clearly, the notion of a robust CDS models the idea that if a state is encoded via the encoder $\mathcal{E}$, then some amount of $C^0$-bounded noise $\eta$ can be allowed in the encoding $\mathcal{E}(s)$ of $s$ such that the states $\mathcal{D}\circ (f^\tau)^k(\mathcal{E}(s) + \eta)$ correctly simulate the dynamics of $\machine{T}$ starting from $s$. (Here $f^\tau = f$ if $\tau = 1$; see Definition \ref{def:CDS1}.) This is different from other notions of robustness in the literature, which we will discuss in Section \ref{sec:related-work} below. 

\subsubsection{Universality: existence and obstructions}

The first result of the paper, beyond setting up the definitions, shows that the resulting theory is non-vacuous:

\begin{theorem}
    There exists a robustly Turing-universal CDS $(f, \mathcal{E}, \mathcal{D}, \tau, \machine{T}_{\!\text{\rm univ}})$ with $\tau(x) = 1$, $t(n) = O(n)$, and $f$ a smooth diffeomorphism of the closed $2$-disk. 
\end{theorem}
\noindent This construction is modeled off that of Moore \cite{moore1991generalized}, using an idea similar to that of \cite{cardona2021constructing}. The construction of \cite{moore1991generalized} provides a map of a square that is piecewise linear, as well as an associated smooth map that, while having correct dynamics for a decoder $\mathcal{D}$ with $\mathcal{D}^{-1}(s) = \{\mathcal{E}(s)\}$ for all $s \in S$, cannot in any evident way be upgraded to a robust decoder; the construction of \cite{cardona2021constructing} shows how to make $f$ a smooth area-preserving map, but suffers from the same problem as the construction of \cite{moore1991generalized}. In fact, area-preserving maps can never furnish a robustly Turing-universal CDS:

\begin{theorem}[see Corollary \ref{corr:measurepreserve}] \label{thm:measure-preserving-not-universal}
    Let $f:M \to M$ be such that $M$ is a codimension $0$ submanifold of $\R^n$, and suppose that there is an an $f$-invariant Borel measure $\mu$ on $M$ which is nonzero on all nonempty open sets and such that $\mu(M) < \infty$. Then $f$ cannot be extended to a robustly Turing-universal CDS $(f, \mathcal{E}, \mathcal{D}, \tau, \machine{T}_{\!\text{\rm univ}})$ for any $\tau$ or $t$.

\end{theorem}

\noindent In particular, the examples of \cite{cardona2021constructing} and subsequent papers cannot be made into robustly Turing-universal CDSs in our sense. As a consequence of a general result about measure-preserving dynamics (Theorem \ref{thm:integrable-systems}), we prove an analog of Theorem~\ref{thm:measure-preserving-not-universal} for ``integrable systems'', e.g.~$f$ is a linear translation on a torus, or a family of these depending on an additional parameter, such as if $f$ is the Hamiltonian dynamics for a Hamiltonian $H: M \to \R$ where $M$ is a ``computable manifold'' (see Remark \ref{rk:computable-manifold}) with a symplectic form $\omega$ such that $H$ is part of a system of pairwise Poisson-commuting Hamiltonians $(H_1, \ldots, H_n)$ with $H_1 = H$. 

Integrable systems are among the simplest possible model systems in physics, as their behavior is completely and efficiently predictable. On the other extreme, we have systems which are `completely chaotic': their behavior is sensitive to their initial conditions in a strong sense. In differentiable dynamics, these are axiomatized under the guise of \emph{uniformly hyperbolic}, or \emph{Anosov}, diffeomorphisms, and form fundamental examples in differentiable dynamical systems theory. Another natural class of examples in dynamical systems theory are time-$1$-gradient flows of generic Morse functions. Their common generalization is the class of \emph{Axiom A} systems, which we review in Section \ref{sec:axiom-a}; a fundamental theorem is that the \emph{structurally stable} systems, i.e. those diffeomorphisms $f$ such that any nearby diffeomorphism $\tilde{f}$ has the \emph{same} dynamics (from a topological perspective) as $f$, are exactly the Axiom A systems satisfying a transversality assumption (see Theorem \ref{thm:axiom-a-is-structurally-stable}). 

It turns out that just as the least chaotic, i.e.~integrable, systems cannot be robustly Turing universal, the `most chaotic' systems, and more generally the structurally stable systems, likewise cannot be robustly Turing universal:

\begin{theorem}
\label{thm:axiom-a-not-universal}
    Let $M$ be a manifold, and let $f: M \to M$ be an Axiom A diffeomorphism. Then $f$ cannot be extended to a robustly Turing-universal CDS $(f, \mathcal{E}, \mathcal{D}, \tau, \machine{T}_{\!\text{\rm univ}})$ for any $\tau$ or $t$.
\end{theorem}
\begin{remark}
    The manifold $M$ does not need to be compact; however, by definition (Definition \ref{def:axiom-A-diffeo}), the nonwandering set of $f$ is compact. 
\end{remark}

\noindent The argument used to prove Theorem \ref{thm:axiom-a-not-universal} uses deep structural results about the dynamics of $f$ due to Smale \cite{smale1967differentiable} and others \cite{hirsch-pugh}. Moreover, the arguments proving Theorems \ref{thm:measure-preserving-not-universal} and \ref{thm:axiom-a-not-universal} do not use all of the structure of the universal Turing machine $\machine{T}_{\!\text{\rm univ}}$. Indeed, we define a notion of a \emph{sub-machine} of $\machine{T}_{\!\text{\rm univ}}$ such that if $f$ simulates $\machine{T}_{\!\text{\rm univ}}$ and $\machine{N}$ is a sub-machine of $\machine{T}_{\!\text{\rm univ}}$, then $f$ also simulates $\machine{N}$ (possibly with modified $t$ and $\tau$, i.e.~with a \emph{slowdown}). 

\begin{theorem}[see Theorem~\ref{thm:measurepreserve}]
    In the setting of Theorem \ref{thm:measure-preserving-not-universal}, in fact $f$ cannot be extended to a robust CDS $(f, \mathcal{E}, \mathcal{D}, \tau, \machine{T})$ if $\machine{T}$ contains as a sub-machine the machine $\machine{Plus}: \{1\}^* \to \{1\}^*$ defined by $\machine{Plus}([n]_1) = [n+1]_1$, where $[n]_1$ denotes $n$ expressed in unary.  Similarly, in the setting of Theorem \ref{thm:axiom-a-not-universal}, $f$ cannot be extended to a robust CDS $(f, \mathcal{E}, \mathcal{D}, \tau, \machine{T})$ if $\machine{T}$ contains as sub-machines $\machine{T}_n$ for all $n > 0$, where $\machine{T}_n:S \to S$ is simply the identity map and $|S|=n$.
\end{theorem}

\noindent The method of proof of Theorem \ref{thm:measure-preserving-not-universal} from earlier is to note that if $f$ simulates the the sub-machine $\machine{Plus}$, then there must be an infinite collection of disjoint regions $C_i$ such that $f^{r_i}(C_i) \subset C_{i+1}$ and such that the sum of the measures of the $C_i$ is finite; thus the measures of the $C_i$ must decrease to zero, which contradicts the requirement that $f$ is measure-preserving. 

In contrast, the proof of Theorem~\ref{thm:axiom-a-not-universal}, namely a proof by contradiction, uses the sub-machines $\machine{T}_n$ to produce an arbitrarily large finite disjoint collection of closed subsets $C_i$ with nonempty interiors such that $f^n(C_i) \subset C_i$ for each $i$. Then, one invokes the \emph{spectral decomposition} of $f$, decomposing its nonwandering set into finitely many basic sets, and then uses stable manifold theory to force each $C_i$ to contain at least one of the basic sets. Essentially, points in $C_i$ must be attracted to one of the basic sets, so they lie on the stable manifold of one of the points on the basic sets; perturbing this point a small amount and following the stable manifold back to $C_i$, one finds a new point asymptotic to a point on a basic set with a dense periodic orbit. Thus since $C_i$ is closed, it must contain that entire basic set; the argument concludes with a contradiction because there are more sets $C_i$ than there are basic sets, and yet the $C_i$ are pairwise disjoint.  

These results suggest future research directions for exactly characterizing the `computational complexity classes' that can be assigned to various types of dynamical systems. There is an expectation in the differentiable dynamics literature that as one allows for \emph{non-uniform hyperbolicity}, and more generally for mixtures of integrable and chaotic behavior (such as that arising in Henon system \cite{pesin2010open}), then more complex types of dynamical behavior can occur \cite{pomeau1980intermittent}.  It would be desirable to identify precise differences between the types of computations that can be implemented by such systems. We venture the following conjecture:

\begin{conjecture}
\label{conj:generic-maps-are-not-universal}
A $C^\infty$ generic $f: M \to M$ cannot be extended to a robustly Turing-universal CDS. 
\end{conjecture}

\noindent We are able to prove this conjecture under a strong additional assumption on the decoder $\mathcal{D}$ as well as a constant slowdown function (see Theorem \ref{thm:generic-diffeo-not-universal-under-hierachical-shrinking}) via a periodic-point argument and the Kupka-Smale Theorem (see \cite[Chapter 7]{katok1995introduction}). It may be possible to resolve this conjecture under the assumption of the well known Palis conjectures on the dynamics of generic smooth dynamical systems \cite{palis2000global}; any argument for this conjecture, like the Axiom A argument, is likely to require sophisticated input from differentiable dynamical systems theory.

\subsubsection{Time complexity bounds}

Going beyond statements about decidability or universality, it is natural to ask more refined questions about the computational capacity of smooth dynamical systems. In particular, it is desirable to show that various natural dynamical conditions on $f$, e.g.~Axiom A, exponential mixing, genericity, etc., bound the computational capacity of $f$ such that one can identify a comparatively small complexity class of problems solvable by such $f$. For example, it is natural to ask whether generic one-dimensional systems can recognize languages that are not in $\textbf{P}$. We formalize such questions in Appendix~\ref{App:compclass}, and proceed by proving some initial time complexity results in this setting:

\begin{theorem}
\label{thm:time-complexity-1-d}
    Let $f: [0,1] \to [0,1]$ be an Axiom A diffeomorphism (this property is generic and equivalent to structural stability). For any CDS $(f, \mathcal{E}, \mathcal{D}, \tau, \machine{T})$ where $t(n) = n$ and $\tau$ is constant, if $\machine{T}$ halts on a configuration then it will halt in time $O(F(n))$, where 
    \[ F(n) =D^{2^n} \]
    for some constant $D$.
    In other words, in the sense of Appendix \ref{App:compclass}, $f$ can recognize languages in at best $\textnormal{\textbf{DTIME}}(O(F(n)))$. 
\end{theorem}

\begin{remark} The result actually holds for a much wider class of slowdown functions $\tau$; %which we do not precisely specify in this paper, 
see Remark
    \ref{rk:annoying-time-complexity-1d-remark}.
\end{remark}

\noindent The proof relies on the structure theory of one-dimensional Axiom A systems. Standard results \cite{katok1995introduction} imply that such systems have only finitely many hyperbolic attracting periodic points, and the complement of their basins of attraction is a (possibly fractal) repelling set. Now, each configuration of $\machine{T}$  is associated to some disjoint union of subintervals of $[0,1]$. One then uses (i) lower bounds from real algebraic geometry about the minimum separation between roots of real polynomials as a function of their coefficients \cite{rump1979polynomial}, (ii) the construction of symbolic dynamics controlling the dynamics of the chaotic hyperbolic repeller of $f$ \cite{katok1995introduction}, and (iii) the complexity condition on the decoder, in order to show that one of these $[0,1]$ subintervals contains a point that must get very close to the hyperbolic periodic point attractor in time $O(F(n))$. Subsequently, either one sees that the rest of the interval must still be stuck near the hyperbolic repelling set for several iterations, which is a contradiction if we do not halt since regions corresponding to configurations cannot overlap; or the entire interval must have escaped a neighborhood of the hyperbolic repelling set, in which case waiting another $O(F(n))$ time will allow us to decide if the computation will halt or not.

In higher dimensions, analogs of the real algebraic geometry bound \cite{rump1979polynomial} seem not to have been established, while the dynamics of Axiom A systems are comparatively more complicated. Thus, we restrict ourselves to the simplest Axiom A systems, the Anosov ones, and prove an inexplicit complexity bound:

\begin{theorem}
\label{thm:Anosov1}
    Let $f: M \to M$ be Anosov and volume-preserving. Consider a CDS $(f, \mathcal{E}, \mathcal{D}, \tau, \textnormal{\textsf{T}})$ where $t(n) = n$, the decoder is Cantor-like (Definition \ref{def:cantor-like-decoder}), and the encoder and decoder are implemented by \textnormal{BSS$_{\textsf{C}}$} machines with the finite set of computable constants being $\{a_1, \ldots, a_\ell\}$.  Then the CDS halts on all configurations  in time $O(\mathcal{C}(n))$, where $\mathcal{C}(n)$ is a computable function depending on $a_1, \ldots, a_\ell$.
\end{theorem}

\noindent Note that unlike in the Axiom A case, such systems must always halt. This feature arises from the topological mixing of Anosov systems, since robust CDSs with topologically mixing dynamics must always halt (Lemma \ref{lemma:topologically-mixing-implies-halting}). To establish a computable halting time bound, one must know \emph{how fast} these systems mix. We do this via the Sinai-Bowen-Ruelle theory of exponential mixing of Anosov systems \cite{bowen1978markov, Bowen1975}, together with some elementary real algebraic geometry counting argument to get a bound on how fast the sets corresponding to configurations of $f$ can shrink as $n$ increases. In short, exponential mixing means that  the time evolution of the space correlation between $f_*^nF_1$ and $F_2$, where $F_1$ is a nonnegative function supported in a ball inside $C_s := \mathcal{D}^{-1}(s)$, and $F_2$ is a nonnegative function supported inside the halting set, must converge to a positive number $\int F_1 \int F_2$ at an exponential rate.  This gives us an upper bound on how long the correlation can be zero, i.e.~on how long the images of $C_s$ (corresponding to input configuration $s$ of the machine $\machine{T}$ simulated by $f$) do not intersect the halting set. Since the sets $C_s$ do not shrink too fast as $n$ increases, we can argue that the bound for the halting time of this CDS is computable.

We expect that the dependence on the computable constants $(a_1, \ldots, a_\ell)$ can be removed and a bound on $\mathcal{C}(n)$ can be made explicit, but this depends on a natural problem in real algebraic geometry which we explain in the text. We expect that in fact one can give polynomial time halting bounds. Moreover, the condition that $f$ is volume-preserving is likely to be inessential, but removing this condition involves more subtle arguments about the invariant SRB measure for $f$, which in general is not absolutely continuous with respect to the Lebesgue measure, even for Anosov systems. We leave these improvements to future work.

\section{Related work}
\label{sec:related-work}

This paper contributes a precise definition that allows one to probe the intrinsic computational capacity of a continuous dynamical system $f$. Moreover, the results of Section \ref{sec:our-results} show that the resulting theory is nontrivial. As mentioned in the introduction, there have been many previous approaches to this problem. Below, we compare our work with the different perspectives taken previously. 

\paragraph{Real computation.} One might ask for the dynamical system $f$ to be computable in one of the many senses of the term \cite{blum1998complexity, braverman_bssc_arxiv, ko2012complexity, weihrauch2012computable}, and apply one of the different theories of computable real functions. Crucially, in our theory, \emph{the continuous dynamical system $f$ is not required to be a computable map}, and even for uncomputable $f$, the theory remains nontrivial. Moreover, using the notion of a \emph{robust} decoder means that we are not studying the dynamics of $f$ on a countable set of computable points, as is done in many previous works \cite{cardona2021constructing, cardona2022turing, cardona2023computability, moore1990unpredictability, moore1991generalized, moore1998finite}, but instead studying its behavior on certain open domains of controlled complexity. These two conditions together force one to use results about $f$ involving its \emph{continuous} dynamical properties (preserving measures, Axiom A-ness, etc.) rather than its properties in terms of some language used to describe $f$.  As such, we elucidate the \emph{computational significance} of various notions in differentiable dynamics. In particular, probing $f$ only along the points of $\mathcal{E}(S)$ rather than over the sets $\mathcal{D}^{-1}(S)$ means that much of the global behavior of $f$ is neglected in the analysis.  Since we think of $f$ as a `natural system' (i.e. a brain or a neural network), from this perspective it is unnatural to expect that we can prepare continuous parameters of states in any experiment without introducing some `error tolerance'. 

\paragraph{Robustness.} There are two flavors of previous perspectives on `robust computation' in continuous dynamics. The first is exemplified by \cite{braverman2012noise, braverman2015space}, where one asks that the dynamics is modified by the addition of some uniform noise. In this case, the system becomes essentially indistinguishable from a finite state Markov process, and because of this most properties of the system become computable in a suitable sense \cite{braverman2012noise}. This makes it difficult to ask \emph{computational complexity} questions about the intrinsic dynamics of $f$. 

The other notion of `robust computation' in continuous dynamics is exemplified by \cite{asarin2001perturbed,bournez2013computation, gracca2008computability}, which essentially fix a robust decoder for $f$ and require that there is an $\epsilon$ such that the \emph{same} robust decoder suffices to simulate the desired machine for \emph{all} $\tilde{f}$ which are $\epsilon$-close to $f$ in the $C^0$ norm. In this formalization, one also finds that universal computation cannot be achieved robustly by finite-dimensional dynamical systems \cite{bournez2013computation}, essentially by approximating the behavior of $f$ by its behavior viewed through a `pixelated lens' of the domain, which is again simply the behavior of a finite state machine. By defining robust computation as an asymptotic condition as the amount of uniform noise added goes to zero,~\cite{bournez2013computation, gracca2008computability} show that systems $f$ can be than finite state machines: in fact, they can robustly recognize precisely the recursive languages.  Moreover, approaches like \cite{bournez2013computation}  focus on asking dynamical systems to \emph{recognize languages} rather than to \emph{simulate machines}, as we do.  Our approach leads to a different class of questions, which are entwined with the theory of differentiable dynamical systems and the notion of simulation, and are less focused on the question of deciding reachability of the halting set. 

In our approach, we prove computational restrictions on an \emph{individual} $f$, \emph{without} any perturbations made to $f$, and for \emph{arbitrary} decoders coupled to $f$. This is a very different setup from results which \emph{fix} a decoder; our approach aims to probe the `intrinsic' computational capacity of $f$. Moreover, our notion of robustness does \emph{not} allow for arguments which boil down the dynamics of $f$ to that of a finite state machine due to the lack of uniformity in the `noise' allowed when simulating $f$.  (Analogously, although every physical computer is in effect a finite state machine, we do not model them as such in order to understand the computations they are performing, and so it is appropriate not to turn every continuous dynamical system into a finite state machine.) Our notion of robustness models the idea of preparing an initial state of an experiment to `sufficiently high precision', with the precision required described by the decoder; as such, it is analogous to a variant of the noise models of \cite{braverman2012noise} with highly non-uniform noise, or to asking for robustness to perturbations of $f$ where the norm of the perturbation is allowed to depend on $x \in M$ (with $M$ the domain of $f$) in a non-uniform way.

\paragraph{`Reasonable' state encodings.} There have been many clever constructions of various mechanisms implementing computation in continuous dynamical systems \cite{moore1990unpredictability, moore1991generalized, moore1998finite,  cardona2021constructing, cardona2023hydrodynamic}. Most often, these encode states of some discrete computational system via \emph{points} of the continuous state space $M$, i.e.~the decoder satisfies $\mathcal{D}^{-1}(s)= \{\mathcal{E}(s)\}$. As mentioned before, this ignores robustness, and only probes the behavior of $f$ on some Cantor set of points, making it difficult to prove \emph{upper} bounds on the computational capacity of $f$.

Moreover, previous works usually ask for the state encoding to be `reasonable' \cite{moore1998finite}, usually without making precise what this means \cite{cardona2021constructing}. In this paper we give a precise definition of a `reasonable' state encoding, and also advocate for focusing on dynamics of open sets (or sets with non-empty interior) rather than of individual points in the continuous state space. Although our definition of a `reasonable' state encoding may seem natural after reading it, our use of the theory of real computation as a bootstrap to make sense of the computational power of more general continuous dynamical systems is new, and is carefully designed to avoid many potential issues having to do with the theory of real computation. For example, BSS$_\R$ machines (which were the ones originally defined by Blum-Shub-Smale \cite{blum1989theory}, rather than the later BSS$_{\textsf{C}}$-machines) are capable of strong super-Turing computation via access to uncomputable constants \cite{braverman_bssc_arxiv}. Additionally, decoders defined using non-uniform circuit families also lead to pathologies like Example \ref{example:uncon1}. Finally, formulating the decoding of states in terms of bit complexity either leads to pathologies involving states encoded in regions which are essentially `pixelated' (if the decoder can be bit-computed in finite time) or only being able to define the encoder and decoder via an infinite computation, making it more challenging to state the required complexity constraints on the encoder and decoder. We do not explicate these alternative problematic formalizations in this work, but it is a central contribution of this paper to define CDSs such that they formalize previously existing intuitions while avoiding many possible technical pitfalls.

\paragraph{Symbolic dynamics.} A popular approach in the mathematical literature on continuous dynamical systems is to study the dynamics of $f$ via \emph{symbolic dynamics}, namely by discretizing the continuous state space $M$ by removing a measure zero set $M_0$, defining another map $\sigma: M \setminus M_0 \to \Sigma$ for some discrete alphabet $\Sigma$, and studying $f$ via the induced language $L_{f, \sigma} \subset \Sigma^\infty$ given by the set 
\[L_{f, \sigma} = \{\tau^\sigma_f(x)\, : \, x \in M \setminus M_0 \},\quad \tau^\sigma_f(x) := (\sigma(x), \sigma(f(x)), \sigma(f^2(x)), \ldots )\,.\] 
One hopes to find a $\sigma$ for which the map $x \mapsto \tau^\sigma_f(x)$ is injective on $x \in M \setminus M_0$, but for which the number of symbols is small, e.g.~finite. This approach is reviewed in Appendix \ref{app:symbolic-dynamics}. It turns out that it is often possible to achieve this, and the method and its variations are very helpful for studying dynamical properties of $f$. Thus, one might be tempted to say that one should identify $f$ with the corresponding discrete language $L_{f, \sigma}$.

However, this approach leads to several challenges. The first is that one might have different discretization maps $\sigma_1$ and  $\sigma_2$ such that the resulting languages $L_{f, \sigma_1}$ and $L_{f, \sigma_2}$ are of vastly different computational complexity. As such, there is no obvious mechanism for \emph{upper bounding} the computational capacity of a differentiable dynamical system if one measures computational capacity in terms of the corresponding language $L_{f, \sigma}$.  The second is that while the symbolic dynamics allows one to associate notions of language complexity to dynamical systems, it does not allow one to discuss the \emph{way} that $f$ is performing the desired pattern-recognition problem -- one cannot argue that $f$ is simulating any \emph{given} algorithm. The third is that for many systems, preferred classes of discretizations $\sigma$ called \emph{Markov partitions} have been proven to exist; however, the resulting maps $\sigma$ have the property that the boundary of $\sigma^{-1}(s)$ for $s \in \Sigma$ is a fractal \cite{bowen1978markov}, and the computability of these sets is not  clear. To arrange for computability, one must certainly require that $f$ is computable; however, upgrading existing constructions of Markov partitions to a computable analysis setting is laborious (see the thesis \cite{kenny-roberts-thesis}, which proves some results for two-dimensional diffeomorphisms).  

In this paper, while we use symbolic dynamics in the proof of Theorem \ref{thm:time-complexity-1-d} to study the dynamics of $f$, our discretization of the dynamics is instead given by $\mathcal{D}$, which is always (efficiently) computable. Tools from symbolic dynamics (see Appendix~\ref{app:symbolic-dynamics} for an overview) are naturally adapted to study \emph{dynamical} properties of $f$, e.g.~mixing and periodic point properties. These are helpful for understanding how $f$ might `compute', but do not straightforwardly answer computational questions about $f$ in general. Part of the purpose of this paper is to highlight the gap between a dynamical and a computational understanding of differentiable dynamical systems.

\section{Preliminaries}
\label{sec:preliminaries}

In this section we set up the definition of a computational dynamical system, motivated by considering a general theory of simulability of one machine by another.  We will explicate how simulability has a subtle but fundamental relationship with computational complexity, which is neglected in standard textbook treatments.  Next we provide an overview of (differentiable) dynamical systems theory, as appropriate for our analyses in this paper.  Since dynamical systems theory, including its tools and perspectives, are mostly unfamiliar to computer scientists, we delve into more detail on these vis-\`{a}-vis other preliminaries.

\subsection{Definitions for simulability and CDSs}
\label{subsec:defofCDS}

\subsubsection{Simulability}

The foundations of computer science are predicated on notions of \textit{simulability}, as manifested by the Church-Turing thesis and its variants.  Strangely, textbook treatments do not provide a general definition of `simulability' as an answer to the question: ``What does it mean for one system to simulate another?''  While it is standard in textbook treatments (see e.g.~\cite{sipser2012introduction, arora2009computational}) to give an example of universal Turing machine that can reproduce the outputs of all other Turing machines with some slowdown, this is merely an example of a Turing machine that we might take as being able to `simulate' all others, leaving unanswered the question of what `simulation' precisely means.

Let us motivate a suitable definition of simulation by exploring two initial desiderata.  To do so, we need some notation.  Abstractly, let us denote a machine by a map $\textsf{T} : S \to S$ where $S$ is the configuration space.  For example, $\textsf{T}$ could be a Turing machine and some $s \in S$ would be the joint description of the state of the head, location of the head, and the symbols on the tape.  (We will provide precise definitions in the Turing machine setting shortly.)  Then $\textsf{T}^n(s)$ describes the configuration of the system after $n$ steps.  This is a type of autonomous system since it merely has an initial condition and no external inputs at intermediate time steps; as such it does not specialize a standard formulation of e.g.~finite state machines which instead must be viewed as forced systems $\machine: S \times C \to S$ where $C$ are some `control' variables.  We will study forced systems in a separate work~\cite{cotlerrezchikovWIP}.  For now, we consider two machines $\machine{T}_1 : S_1 \to S_1$ and $\machine{T}_2 : S_2 \to S_2$, and we want to understand what it would mean for $\machine{T}_2$ to simulate $\machine{T}_1$.

The first desiderata is that all realizable configurations of $\machine{T}_1$, namely $S_1$, should correspond to some set of configurations in $S_2$.  That is, there is a translation protocol between configurations of the system being simulated and the simulation.  This can be expressed as a partial, surjective function $\mathcal{D} : S_2 \rightharpoonup S_1$ (here `$\rightharpoonup$' denotes a partial function).  We call this function $\mathcal{D}$ the \textit{decoder}.  We also use the notation $C_s := \mathcal{D}^{-1}(s)$ to denote the configurations in $S_2$ corresponding to the configuration $s$ in $S_1$.

The second desiderata is that the dynamics of the simulation $\machine{T}_2$ should `track' the dynamics of $\machine{T}_1$.  That is, for some function $\tau : S_2 \to \mathbb{N}$ which is constant on connected components of $\mathcal{D}^{-1}(s)$, we would want
\begin{align}
\label{E:firstDM2M1}
\mathcal{D} \circ \machine{T}_2^\tau(C_s) = \machine{T}_1(s)
\end{align}
for all $s \in S_1$.  Our above notation $\machine{T}_2^\tau : S_2 \to S_2$ means $\machine{T}_2^\tau : s_2 \mapsto \machine{T}_2^{\tau(s_2)}(s_2)$.  The `slowdown function' $\tau$ allows for the possibility that $\machine{T}_2$ tracks the dynamics of $\machine{T}_1$ with a slowdown for each computational step, contingent on the details of the encoded representation at that step.  We can equivalently write
\begin{align}
\label{E:firstDM2M12}
\mathcal{D} \circ (\machine{T}_2^\tau)^n(C_s) = \machine{T}_1^n(s)
\end{align}
for all $s \in S_1$ and all $n \geq 0$.  As a convenience, we might want a means to select a particular configuration in $C_s$, so we can consider a map $\mathcal{E} : S_1 \to S_2$ such that $\mathcal{E}(s) \in C_s$ for all $s \in S_1$.  We call $\mathcal{E}$ an \textit{encoder}, and it evidently satisfies $\mathcal{D}\circ \mathcal{E}(s) = s$ for all $s \in S_1$.  Then~\eqref{E:firstDM2M12} implies
\begin{align}
\mathcal{D} \circ (\machine{T}_2^{\tau})^n \circ \mathcal{E}(s) = \machine{T}_1^n(s)\,,
\end{align}
which as before holds for all $s \in S_1$ and all $n \geq 0$.

We are now prepared to provide a formal definition of simulation.  For ease of exposition, we will take the simulation and system being simulated to both be Turing machines.  Our definition will readily generalize beyond the setting of Turing machines, and we will employ such a generalization when we consider dynamical systems.  To this end, we begin with a suitable definition of a $k$-tape Turing machine:
\begin{definition}[Turing machine, adapted from~\cite{arora2009computational}]\label{def:Turingmachine1} A \textbf{Turing machine} is given by a triple $(Q, \Gamma, \delta)$ where $Q$ and $\Gamma$ are finite sets and:
\begin{enumerate}
    \item $Q$ is the set of states, containing a start state $q_0$ and at least one halt state $q_{\text{\rm halt}}$;
    \item $\Gamma$ is the tape alphabet, not containing a blank symbol $\sqcup$; and
    \item $\delta : Q \times \Gamma \cup \{\sqcup\} \to Q \times \Gamma
    \times \{\text{\rm L}, \text{\rm R}, \text{\rm S}\}$.
\end{enumerate}
The configuration space $S$ of a Turing machine is given by $S := \Gamma^* \times Q \times \Gamma^*$, where a configuration is denoted by $s = x_1 \cdots x_{m-1} \,q\,x_m \cdots x_n$ for $x_i \in \Gamma$ and $q \in Q$, with the understanding that all symbols to the left of $x_1$ are blank and all symbols to the right of $x_n$ are blank.  Our notation expresses that the head is above $x_m$, and that all of the symbols on the tape to the left of $x_1$ and to the right of $x_n$ are blank.  We can define a map $\textnormal{\textsf{T}} : S \to S$ as follows.  If $\delta(q, x_m) = (q', x_m', a)$ for $a \in \{\text{\rm L},\text{\rm R}, \text{\rm S}\}$, then
\begin{align}
\textnormal{\textsf{T}}(s)
 = \begin{cases}
x_1 \cdots x_{m-2} \,q'\, x_{m-1} x_{m}' \cdots x_n &\text{if }a = \text{\rm L} \\
x_1 \cdots x_{m-1} \,q'\, x_{m}' \cdots x_n &\text{if }a = \text{\rm S} \\
x_1 \cdots x_m' \,q'\, x_{m+1} \cdots x_n &\text{if }a = \text{\rm R}
\end{cases},
\end{align}
which describes one time step of the Turing machine.
\end{definition}
\noindent The above readily generalizes to the $k$-tape setting, wherein $\delta : Q \times (\Gamma \cup \{\sqcup\})^k \to Q \times \Gamma^k \times \{\text{\rm L},\text{\rm R},\text{\rm S}\}^k$.  We note that there are many related definitions of Turing machines, but our later analyses will not really be sensitive to the choice of definition.

We also require an initial definition of encoders, decoders, and slowdown functions, given below.

\begin{definition}[Encoders, decoders, slowdown functions, and their complexities]
\label{def:encdec1}
Let $S_1$ and $S_2$ be the configuration spaces of two Turing machines $(Q_1, \Gamma_1, \delta_1)$ and $(Q_2, \Gamma_2, \delta_2)$.
\begin{itemize} 
    \item An \textbf{encoder} is a function $\mathcal{E} : S_1 \to S_2$ that can be instantiated by a $2$-tape Turing machine, with one tape having symbols in $\Gamma_1 \cup \{\sqcup\} \cup Q_1$ and the other having symbols in $\Gamma_2 \cup \{\sqcup\} \cup Q_2$. (Here $\sqcup$ is the blank symbol.)  We will additionally allow for the $2$-tape Turing machine to write the blank symbol.  The initial state of the first tape is $q_0\,\sqcup$ and the initial state of the second tape is $q_0'\,s^{(2)}$ for $s^{(2)} \in S_2$.  Then the machine halts with the first tape in the configuration $q_{\textnormal{halt}}\,\mathcal{E}(s^{(2)})$, where $\mathcal{E}(s^{(2)}) \in S_2$, and the second tape is all blank symbols. If the machine halts in time $O(t(|s^{(2)}|))$, then we say that $\mathcal{E}$ has \textbf{input time complexity} $O(t(n))$.
    \item A \textbf{decoder} is a partial function $\mathcal{D} : S_2 \to S_1$ that can be instantiated by $2$-tape Turing machine as follows, with one tape having symbols in $\Gamma_1 \cup \{\sqcup\} \cup Q_1$ and the other having symbols in $\Gamma_2 \cup \{\sqcup\} \cup Q_2$.  As before, we will additionally allow for the $2$-tape Turing machine to write the blank symbol.  The initial state of the first tape is $q_0\,s^{(1)}$ for $s^{(1)} \in S_1$, and the initial state of the second tape is $q_0'\,\sqcup$.  If $s^{(1)}$ is in the domain of definition of $\mathcal{D}$, then the machine halts with the second tape in the configuration $q_{\textnormal{halt}}'\,\mathcal{D}(s^{(1)})$ for $\mathcal{D}(s^{(1)}) \in S_2$, and the first tape having all blank symbols.  If the machine halts in time $O(t(|\mathcal{D}(s^{(1)})|))$, then we say that $\mathcal{D}$ has \textbf{output time complexity} $O(t(n))$.
    \item A \textbf{slowdown function} is a partial function $\tau : S_2 \rightharpoonup \mathbb{Z}_{\geq 0}$ which is defined exactly on the domain of definition of $\mathcal{D}$, is constant on the connected components of $\mathcal{D}^{-1}$, and can be instantiated by a 2-tape Turing machine (which we allow to write blank symbols) as follows.  The initial state of the first tape of the Turing machine is $q_0\,s^{(2)}$ for $s^{(2)} \in S_2$, and the initial state of the second tape is $q_0'\,\sqcup$.  If $s^{(2)}$ is in the domain of definition of $\tau$, then the machine halts with the second tape in the configuration $q_{\textnormal{halt}}'\,\tau(s^{(2)})$ for $\tau(s^{(2)})$ a binary representation of an element of $\mathbb{Z}_{\geq 0}$, and the first tape blank.  If the machine halts in time $O(u(|\mathcal{D}(x)|))$, then we say that $\tau$ has \textbf{slowdown function time complexity} $O(u(n))$.  We will always assume that $u(n)$ is a computable function. 
\end{itemize}
\end{definition}

With the above definitions at hand, we are prepared to define the simulability of one Turing machine by another.  
\begin{definition}[Simulation of Turing machines] 
\label{def:sim1}
Let $\textnormal{\textsf{T}}_1$ and $\textnormal{\textsf{T}}_2$ be Turing machines with configuration spaces $S_1$ and $S_2$, respectively.  We say that $\textnormal{\textsf{T}}_2$ $(\tau, t(n))$--simulates $\textnormal{\textsf{T}}_1$ if the exists an encoder $\mathcal{E} : S_1 \to S_2$ with input complexity $O(t(n))$, a decoder $\mathcal{D} : S_2 \rightharpoonup S_1$ with output complexity $O(t(n))$, and a slowdown function $\tau : S_2 \to \mathbb{Z}_{\geq 0}$ such that:
\begin{enumerate}
    \item $\mathcal{E}(s) \in C_s$ for all $s \in S_1$, where $C_s := \mathcal{D}^{-1}(s)$; and
    \item $\mathcal{D} \circ \textnormal{\textsf{T}}_2^{\tau}(C_s) = \textnormal{\textsf{T}}_1(s)$ for all $s \in S_1$.
\end{enumerate}
We recall that $\textnormal{\textsf{T}}_2^\tau : S_2 \to S_2$ means $\textnormal{\textsf{T}}_2^\tau : s_2 \mapsto \textnormal{\textsf{T}}_2^{\tau(s_2)}(s_2)$.
\end{definition}
\noindent Before explaining the significance of this definition and its relation to prior work, several initial remarks are in order.
\begin{remark}
    We may say that $\textsf{T}_2$ is a \emph{universal} Turing machine if it $(\tau, t(n))$-simulates $\textsf{T}_1$ for every Turing machine $\textsf{T}_1$ and some $(\tau, t(n))$ which possibly depend on $\textsf{T}_1$. While this does not agree with all definitions of universal Turing Machines, which often only ask for $\textsf{T}_2$ to correctly compute recursive functions rather than to \emph{simulate} other Turing machines, it nonetheless usually holds in \emph{constructions} of universal Turing machines. For a definition in the vein of Definition \ref{def:sim1}, see e.g.~\cite{rogozhin1996small}. 
\end{remark}
\begin{remark}
\label{rem:haltstate1}
Above we are tacitly assuming that if some $s$ in $S_1$ contains a halt state, then $\textsf{T}_1(s) = s$.  This need not be the case; that is, we could instead say that time evolution of $\textsf{T}_1$ simply ends when we reach the halt state.
\end{remark}

\begin{remark}
We can readily generalize Definition~\ref{def:sim1} to the setting where $\textsf{T}_1$ and $\textsf{T}_2$ have $k$ and $k'$ tapes, respectively, possibly with $k \not = k'$.  Then it would be sensible for $\mathcal{E}$ and $\mathcal{D}$ to correspond to Turing machines with (at least) $k + k'$ tapes.  Related generalizations, where the machines $\textsf{T}_1$ and $\textsf{T}_2$ in question need not even be Turing machines, proceed by analogy with Definition~\ref{def:sim1}.
\end{remark}

\begin{remark}[Oblivious simulation]
\label{rem:obvlivious1}
An \textbf{oblivious simulation} is one for which $\tau(s_2) = f(|\mathcal{D}(s_2)|)$; that is, $\tau$ only depends on the length of the decoded string.  Any Turing machine can be simulated via an oblivious simulation (see e.g.~\cite{arora2009computational}).
\end{remark}

A key feature of Definition~\ref{def:sim1} is that it accounts for the complexity $O(t(n))$ of the encoder and decoder.  We will show that this accounting for complexity is completely essential for the definition to be sensible and meaningful.  Moreover, it was missed in previous work.  In particular, there have been a number of previous works (see~\cite{rosen2011anticipatory, rosen1991life, branicky1995universal, koiran1999closed} and specifically~\cite{gracca2005robust, gracca2008computability,  cardona2023hydrodynamic}, as well as examples in~\cite{tao2017universality, cardona2021constructing, cardona2022turing, cardona2023computability}) which define simulability akin to Definition~\ref{def:sim1}, but without any specification of the complexity of the encoders and decoders.  Without such a specification, we can run into several problems.  Specifically, suppose that $\textsf{T}_2$ simulates $\textsf{T}_1$ with respect to the encoder and decoder $\mathcal{E}$, $\mathcal{D}$.  We might be inclined to say that this relation implies that $\textsf{T}_2$ in a sense `contains' $\textsf{T}_1$, or that $\textsf{T}_2$ is at least `as powerful' as $\textsf{T}_1$.  But these statements may not be true.  For suppose that the encoder or decoder are more computationally powerful than the simulator $\textsf{T}_2$, but as powerful as the simulated system $\textsf{T}_1$.  As such, even if $\textsf{T}_2$ has very weak computational capabilities, we could still satisfy conditions akin to 1 and 2 in Definition~\ref{def:sim1} since the encoder and decoder could do the `work' of simulating the computation of $\textsf{T}_1$.  We will give an example of this phenomenon in Section~\ref{subsubsec:counterexample1} below.  In summary, accounting for the complexity of the encoder and decoder is essentially that we correctly attribute computational power to $\textsf{T}_2$ vis-\'{a}-vis the encoder and decoder.

There is a useful heuristic that further clarifies the above point.  When we say that $\textsf{T}_2$ simulates $\textsf{T}_1$, what we might want is an implication like $\textsf{Complexity}(\textsf{T}_2) \geq \textsf{Complexity}(\textsf{T}_1)$, for some notion of complexity.  However, $\textsf{T}_2$ being able to simulate $\textsf{T}_1$ using $\mathcal{E}, \mathcal{D}$ actually implies a schematic inequality along the lines of
\begin{align}
\label{E:schematic1}
\max\{\textsf{Complexity}(\textsf{T}_2), \textsf{Complexity}(\mathcal{D})\} \geq \textsf{Complexity}(\textsf{T}_1)\,.
\end{align}
As such, if $\textsf{Complexity}(\mathcal{D}) \geq \textsf{Complexity}(\textsf{T}_1)$, then~\eqref{E:schematic1} is automatically satisfied and so we do not learn about the computational power of $\textsf{T}_2$.  On the other hand, if $\textsf{Complexity}(\mathcal{D}) < \textsf{Complexity}(\textsf{T}_1)$, then~\eqref{E:schematic1} implies $\textsf{Complexity}(\textsf{T}_2) \geq \textsf{Complexity}(\textsf{T}_1)$, and so we do learn about the computational power of $\textsf{T}_2$.  As such, the conceptual lesson is: \textit{the encoder and decoder need to be less computationally complex than the system being simulated if we wish to learn about the computational powers of the simulator.}  Moreover, we emphasize that we are always learning about the computational power of the simulator \textit{vis-\'{a}-vis} the computational power of $\textsf{T}_1$.

Before proceeding, let us demonstrate that Definition~\ref{def:sim1} interfaces nicely with standard constructions of universal Turing machines, which we often say can `simulate' any other Turing machine.

\begin{theorem}[Efficient universal Turing machine, adapted from~\cite{hennie1966two, arora2009computational}]\label{thm:efficient1}
There is a 3-tape universal Turing machine that can $(\tau, t(n))$--simulate any $k$-tape Turing machine for $\tau(s) \leq O(\log |\mathcal{D}(s)|)$ and $t(n) = n$.
\end{theorem}

\noindent The statement and proof of Theorem~\ref{thm:efficient1} in~\cite{hennie1966two, arora2009computational} discuss the $\tau(s) \leq O(\log |\mathcal{D}(s)|)$ growth explicitly.  If the Turing machine being simulated terminates in $T$ steps, then the simulator will terminate in order $T \log T$ steps.  The $t(n) = n$ behavior is implicit in the proofs in~\cite{hennie1966two, arora2009computational}.  We observe that the $t(n) = n$ dependence is optimal, since the encoder and decoder need to look at every symbol in the input and output at least once, giving us $t(n) = \Omega(n)$.  This motivates the following definition, to be used later:
\begin{definition}[Optimal encoder and decoder]\label{def:optencdec1} An encoder and decoder pair $\mathcal{E}, \mathcal{D}$ are optimal if $\mathcal{E}$ has input time complexity $\Theta(n)$ and $\mathcal{D}$ has output time complexity $\Theta(n)$.
\end{definition}
\noindent As we go along, we point out when an encoder and decoder are optimal. We now state a useful and elementary lemma.
\begin{lemma}
\label{lemma:UTM-has-periodic-points}
   If $\textnormal{\textsf{T}}_2:S_2 \to S_2$ is a universal Turing machine then there is an $s \in S_2$ which is non-halting and which is a periodic point of $\textnormal{\textsf{T}}_2$. In fact, for any $N$, there is an $L \gg 0$ such that $\textnormal{\textsf{T}}_2$ has at least $N$ distinct periodic orbits of period $L$. 
\end{lemma}

\begin{proof}
    Let $\textsf{T}_1:S_1 \to S_1$ be the Turing machine which `does nothing'; in particular, whenever the head is not in the halting state we have $\textsf{T}_1(s)=s$. We must have that $\textsf{T}_2$ $(\tau, t(n))$-simulates $\textsf{T}_1$ for some $\tau$ and $t(n)$ and some encoder-decoder pair $\mathcal{E}, \mathcal{D}$. Now we note that $\mathcal{D}^{-1}(s)$ is finite for every $s \in S_1$; indeed, the bound on output time complexity controls the length of the input, as the algorithm computing $\mathcal{D}$ must \emph{erase the input} on the first tape (and in particular read all of its symbols). We then claim that for every $s \in S_1$ with the head not in the halting state, the set $\mathcal{D}^{-1}(s)$ must contain a periodic point for $\textsf{T}_2$. Indeed we must have by condition 2 of Definition \ref{def:sim1} that $\textsf{T}^\tau_2(\mathcal{D}^{-1}(s)) \subset \mathcal{D}^{-1}(s)$; so this set must contain a fixed point of $\textsf{T}^\tau_2$, i.e.~a periodic point of $\textsf{T}_2$. The second statement follows from the fact that the sets $\mathcal{D}^{-1}(s)$, each of which contains a $\textsf{T}_2$-periodic point, are disjoint as $s$ runs over the infinitely many configurations $s \in S_1$ which are period-1 non-halting configurations for $\textsf{T}_1$. 
\end{proof}

% attribute

% translation example

%At an intuitive level, a potential

% restrict the complexity

%account for

% [[Explain how the definition is the right one, and satisfied for textbook examples; maybe give an example]]

% \noindent For concreteness, we will take $\Sigma = \{0,1\}$ in our discussion of Turing machines henceforth.

% [[ How our definition relates to others ]]

\subsubsection{A useful example}
\label{subsubsec:counterexample1}

Here we show that a `trivial' Turing machine can simulate a universal Turing machine if we do not put any complexity restrictions on the encoder and decoder.  This example emphasizes that such restrictions are essential for simulation to be a useful or meaningful notion.  Let $\textsf{T}_{\!\text{univ}}$ with triple $(Q, \Gamma, \delta)$ be some single-tape universal Turing machine with configuration space $S = \Gamma^* \times Q \times \Gamma^*$. Let us suppose that $\{0,1\} \subseteq \Gamma$.  Now let us define the `trivial' Turing machine that will simulate $\textsf{T}_{\!\text{univ}}$ using a complicated encoder and decoder pair.

Let the `trivial' Turing machine $\textsf{T}_{\!\text{trivial}}$ be described by the triple $(Q, Q \cup \Gamma, \delta_{\text{simple}})$, where
\begin{align}
\label{E:simpledeltarule1}
\delta_{\text{trivial}}(q_0, \sqcup) = (q_0, 1, \text{R})
\end{align}
and otherwise $\delta_{\text{trivial}}(q', b) = (q_{\text{halt}}, 0,\,\text{S})$ for $q' \not = q_0$ or $b \not = \sqcup$.  On account of~\eqref{E:simpledeltarule1}, for any $s \in S$ we have
\begin{align}
\textsf{T}_{\!\text{trivial}}(s \,0\underbrace{1 1 \cdots 1}_{k} q_0 \, \sqcup) = s\,0\underbrace{1 1 \cdots 1}_{k+1} q_0 \, \sqcup\,.
\end{align}
Then defining the encoder $\mathcal{E}$ and decoder $\mathcal{D}$ by
\begin{align}
\label{E:encoddecodeex1}
\mathcal{E}(s) := s\,0\,q_0\,\sqcup\,,\qquad \mathcal{D}(s\,0\underbrace{1 1 \cdots 1}_{k} q_0 \, \sqcup) := \textsf{T}_{\!\text{univ}}^k(s)\,,
\end{align}
we see that $\textsf{T}_{\!\text{trivial}}$ simulates $\textsf{T}_{\!\text{univ}}$ with respect to $\mathcal{E}, \mathcal{D}$.

In this example, we see that while $\textsf{T}_{\!\text{trivial}}$ cannot furnish Turing-universal computation itself, the decoder $\mathcal{D}$ does furnish Turing-universal computation.  Then the only role of $\textsf{T}_{\!\text{trivial}}$ is to serve as a memory and counter.  As such, this example demonstrates that without constraining the complexity of the encoder and decoder, they can `do the work' of the computation that we would otherwise wish to associate with the simulator.  To connect back with our schematic~\eqref{E:schematic1}, the decoder in~\eqref{E:encoddecodeex1} indeed has the same `complexity' as $\textsf{T}_{\!\text{univ}}$ itself, which renders the simulation property uninteresting.

We emphasize that decoder in~\eqref{E:encoddecodeex1} does not have any bounded output complexity bound of the form $O(t(n))$.  To see this, consider a configuration $s$ which is a fixed point of $\textsf{T}_{\!\text{univ}}$ so that $\textsf{T}_{\!\text{univ}}(s) = s$.  Then $\mathcal{D}(s\,0\,11\cdots1\,q_0 \sqcup) = s$ regardless of the number of $1$'s, but its complexity grows with the number of $1$'s.  Then the output complexity is not a function of $|s|$.

As such, we see that merely having an input and output complexity bound for the encoder and decoder, respectively, provides important constraints on what kinds of machines can `simulate' a universal Turing machine.  These same considerations will carry over to our discussions of computational dynamical systems below.

\begin{remark}
    It is natural to require additionally that if $\textsf{T}_2$ simulates $\textsf{T}_1$, then $\textsf{T}_2$ halts exactly when $\textsf{T}_1$ halts. This is not done in Definition \ref{def:sim1}, and the astute reader may notice that including this condition invalidates the example described above. However, the purpose of Definition \ref{def:sim1} will ultimately be to clarify the meaning of what it means for a \emph{dynamical system} $f$ (which may ultimately be some differentiable map) to simulate a Turing machine, as per Definition \ref{def:sim2} below. There is no natural meaning to the `halting' of $f$, so it makes sense to drop the halting condition here; moreover, the example above forms the basis for a more interesting example in the next section below. 
\end{remark}

%inadequate

% [[ Give simple counterexample (maybe with a figure) ]]

% [[ articulate the general principle(s) (and italicize) ]]

\subsubsection{Defining CDSs}
\label{sec:defining-CDS}

We now turn to defining a computational dynamical system, or CDS.  At a high level, our definition will parallel Definition~\ref{def:sim1}; that is, a computational dynamical system is essentially a dynamical system that simulates a particular machine.  In slightly more detail, a CDS will be specified by a tuple $(f, \mathcal{E}, \mathcal{D}, \tau, \textsf{T})$ where $f : M \to M$ for some space $M$ (possibly a manifold) and $\textsf{T} : S \to S$ is a Turing machine, such that $f$ $(\tau, t(n))$--simulates $\textsf{T}$ with respect to $\mathcal{E}, \mathcal{D}$.  If $M$ is a discrete space, then we reduce to our previous discussions.  However, if $M$ is e.g.~a smooth manifold, then the story becomes more interesting. For instance, notice that $\mathcal{E} : S \to M$, $\mathcal{D} : M \rightharpoonup S$, and $\tau : M \to \mathbb{Z}_{\geq 0}$.  As such, if $S$ is a space of finite strings and $M$ is a smooth manifold, then $\mathcal{E}$, $\mathcal{D}$, and $\tau$ cannot be implemented by ordinary Turing machines.

For our purposes, we will always consider manifolds which are explicitly subsets of $\mathbb{R}^k$, although our analyses can be readily adapted to more general manifolds.  As such, it is natural that $\mathcal{E}, \mathcal{D}$ be implemented by Turing machines with $\Gamma = \mathbb{R}$, i.e.~having $\mathbb{R}$-valued symbols on the tapes.  An $\R$-Turing machine is essentially a BSS machine~\cite{blum1989theory}.  The BSS machines have been used to develop a rich theory of computability and complexity over the reals, including decidability problems involving fractals,  and complexity bounds for computing features of polynomial inequalities~\cite{blum1989theory, smale1997complexity}.  We note that BSS machines can be defined over any field $F$, with $F = \mathbb{Z}_2$ corresponding to ordinary Turing machines (with $\Gamma = \{0,1\}$) \cite{braverman_bssc_arxiv}. Since we are interested in the $F = \mathbb{R}$ setting, we will take `BSS machine' to mean a `BSS machine with $F = \mathbb{R}$' unless otherwise specified.

BSS machines will serve as our model of encoders and decoders $\mathcal{E}, \mathcal{D}$ for CDSs, as well as our slowdown functions $\tau$.  Since the definition for BSS machines is somewhat elaborate, we have put it in Appendix~\ref{App:BSSC} so as not to interrupt the flow of the paper.  We note that we make a slight modification to the definition of BSS machines suggested in~\cite{braverman_bssc_arxiv}.  Typically, the head of a BSS machine has access to a finite number of real-valued constants, which are not required to be computable.  In~\cite{braverman_bssc_arxiv}, Braverman suggests requiring the constants in question to be computable so as to avoid certain pathologies, and he calls the corresponding machines BSS$_{\textsf{C}}$ machines, where the `\textsf{C}' stands for `computable constants'.  We will use BSS$_{\textsf{C}}$ machines henceforth.  In some cases, our results will depend on the details of the particular finite set of computable constants to which a BSS$_{\textsf{C}}$ machine has access.

It will be notationally convenient to define a hybrid of a BSS$_{\textsf{C}}$ machine and ordinary Turing machine.  We defer a precise definition to Appendix~\ref{App:BSSC}, and so give an informal definition here.
\begin{definition}[Hybrid BSS$_{\textnormal{\textsf{C}}}$ machine, informal] A \textbf{hybrid BSS$_{\textnormal{\textsf{C}}}$ machine} is a \text{\rm BSS$_{\textnormal{\textsf{C}}}$} machine over $\mathbb{R} \times \mathbb{Z}_d$ with two tapes.  The first tape has symbols valued in $\mathbb{R}$, and the second tape has symbols valued in $\mathbb{Z}_d$.  We will let $0$ (in either $\mathbb{R}$ or $\mathbb{Z}_d$) be a proxy for the blank symbol, and will use $0$ and $\sqcup$ interchangeably in this context. 
\end{definition}

\noindent While a BSS$_{\textsf{C}}$ machine over $\mathbb{R}$ is sufficient to capture the power of a hybrid BSS$_{\textsf{C}}$ machine, we will nonetheless find the definition of the latter to be natural and conceptually useful.  To this end, we are now prepared to define $\mathbb{R}$-valued encoders and decoders by direction analog with Definition~\ref{def:encdec1}.

\begin{definition}[$\mathbb{R}$-valued encoders, decoders, slowdown functions, and their complexities] \label{def:decoders}
Let $S$ be a language over a finite alphabet  $\Sigma$ of size $d$, and $M$ be a submanifold of $\mathbb{R}^k$.
\begin{itemize}
    \item An \textbf{encoder} is a function $\mathcal{E} : S \to M$ that can be instantiated by a hybrid BSS$_{\textnormal{\textsf{C}}}$ machine over $\mathbb{R} \times \mathbb{Z}_d$ as follows.  The initial state of the first tape of the hybrid machine is $q_0\,\sqcup$ and the initial state of the second tape is $q_0'\,s$ for $s \in S$.  Then the machine halts with the first tape in the configuration $q_{\textnormal{halt}}\,\mathcal{E}(s)$, where $\mathcal{E}(s)\in M \subset \R^k$, and the second tape having all blank symbols.  If the machine halts in time $O(t(|s|))$, then we say that $\mathcal{E}$ has \textbf{input time complexity} $O(t(n))$.
    \item A \textbf{decoder} is a partial function $\mathcal{D} : M \rightharpoonup S$ that can be instantiated by a hybrid BSS$_{\textnormal{\textsf{C}}}$ machine over $\mathbb{R} \times \mathbb{Z}_d$ as follows.  The initial state of the first tape of the hybrid machine is $q_0\,x$ for $x \in M \subset \R^k$, and the initial state of the second tape is $q_0'\,\sqcup$.  If $x$ is in the domain of definition of $\mathcal{D}$, then the machine halts with the second tape in the configuration $q_{\textnormal{halt}}'\,\mathcal{D}(x)$ for $\mathcal{D}(x) \in S$, and the first tape having all blank symbols.  If the machine halts in time $O(t(|\mathcal{D}(x)|))$, then we say that $\mathcal{D}$ has \textbf{output time complexity} $O(t(n))$. 
    \item A \textbf{slowdown function} is a partial function $\tau : M \rightharpoonup \mathbb{Z}_{\geq 0}$ which is defined exactly on the domain of definition of $\mathcal{D}$, is constant on the connected components of $\mathcal{D}^{-1}$, and can be instantiated by a hybrid BSS$_{\textnormal{\textsf{C}}}$ machine over $\mathbb{R} \times \mathbb{Z}_d$ as follows.  The initial state of the first tape of the hybrid machine is $q_0\,x$ for $x \in M \subset \R^k$, and the initial state of the second tape is $q_0'\,\sqcup$.  If $x$ is in the domain of definition of $\tau$, then the machine halts with the second tape in the configuration $q_{\textnormal{halt}}'\,\tau(x)$ for $\tau(x) \in \mathbb{Z}_{\geq 0}$, and the first tape having all blank symbols.  If the machine halts in time $O(u(|\mathcal{D}(x)|))$, then we say that $\tau$ has \textbf{slowdown function time complexity} $O(u(n))$.  We will always assume that $u(n)$ is a computable function.  
\end{itemize}

\end{definition}
\noindent In the definition above, the hybrid BSS$_{\textsf{C}}$ machines serve `translators' between $S$ and $M$.  With the above definition at hand, we can now formulate simulability of a machine by a finite state machine, akin to Definition~\ref{def:sim1} above.
\begin{definition}[Simulation of a Turing machine by a dynamical system] 
\label{def:sim2}
Let $f : M \to M$ for $M \subseteq \mathbb{R}^k$ be a dynamical system, and let $\textnormal{\textsf{T}} : S \to S$ be a Turing machine.  We say that $f$  $(\tau, t(n))$-simulates $\textnormal{\textsf{T}}$ if the exists an encoder $\mathcal{E} : S \to M$ with input complexity $O(t(n))$, a decoder $\mathcal{D} : M \rightharpoonup S$ with input complexity $O(t(n))$, and a slowdown function $\tau: M \rightharpoonup \mathbb{Z}_{\geq 0}$   such that
\begin{enumerate}
    \item \emph{(Encoding/Decoding)} $\mathcal{E}(s) \in C_s$ for all $s \in S$, where $C_s := \mathcal{D}^{-1}(s)$; and
    \item \emph{(Simulation Condition)} $\mathcal{D} \circ f^{\tau}(C_s) = \textnormal{\textsf{T}}(s)$ for all $s \in S$.
\end{enumerate}
\end{definition}
\noindent We note that Remark~\ref{rem:haltstate1} for Definition~\ref{def:sim1} generalizes appropriately to Definition~\ref{def:sim2}. Moreover Definition~\ref{def:sim2} readily generalizes to machines beyond Turing machines $\textsf{T}$, and we will consider such generalizations later on.  We also make several additional remarks.

\begin{remark}
\label{rk:explanation-of-simulation-condition}
Notice that in the definition of simulation (Definition \ref{def:CDS1}), the encoder plays no role except to ensure that there is an efficiently computable point inside each set $C_s$. However, we wish to \emph{think} of the encoder $\mathcal{E}$ as serving the role of \emph{encoding} the state $s$ of $\machine{T}$. In this definition of `simulation', we are imagining $f$ as an `unknown' or `natural' system, and $\mathcal{E}$ and $\mathcal{D}$ as experimental apparatuses, where $\mathcal{E}$ encodes a computational state into a state of the continuous system underlying $f$, and $\mathcal{D}$ reads out the state. Thus, another natural version of the simulation condition of Definition  \ref{def:CDS1} would be the condition that 
\begin{equation}
\label{eq:simulation-condition-with-encoder}
\mathcal{D} \circ f^\tau \circ \mathcal{E}(s) = \textsf{T}(s) \,\text{ for all }\, s \in S\,.
\end{equation}
We instead use the `simulation condition' in Definition~\ref{def:sim2} for technical convenience later; see the discussion in the Section \ref{sec:robustness} on robust computation.
\end{remark}

\begin{remark}
The condition that $\mathcal{E}$ and $\mathcal{D}$ are defined by low-complexity circuits, besides avoiding certain pathologies (see Example \ref{example:uncon1} below), is meant to model the idea that an \emph{experimental apparatus} must be implemented in a `known' way, and so its behavior must be efficiently computable.  (See also~\cite{Aharonov:2021das} for a related discussion.)
\end{remark}

\begin{remark}
    In many settings the configuration space $S$ can be thought of as a tuple $(S_{\text{fin}}, S_{\text{inf}})$, where $S_{\text{fin}}$ lies in a finite set and $S_{\text{inf}}$ lies in an infinite set, e.g.~$S_{\text{fin}}$ is the state of the finite state machine composing the Turing machine head and $S_{\text{inf}}$ is the state of the tapes. We could make the stronger requirement that the decoder $\mathcal{D}$ be written as $(\mathcal{D}_{\text{fin}}, \mathcal{D}_{\text{inf}})$ where $\mathcal{D}_{\text{fin}}$ computes the corresponding element of $S_{\text{fin}}$ in $O(1)$ time. This condition often holds in examples, and in the case of the simulation of a Turing machine by a dynamical system lets us define a semialgebraic halting set $\mathcal{D}^{-1}(q_{\text{halt}})$. 
\end{remark}

\begin{remark}[Modification for continuous-time dynamical systems]
\label{rem:modification1}
 If $f$ is a continuous-time dynamical system, e.g.~a map $f : \mathbb{R}_{\geq 0} \times M \to M$ written as $f_t(x)$, then we can modify Definition~\ref{def:sim2} by letting $\tau : M \rightharpoonup \mathbb{R}_{\geq 0}$ (possibly implemented by a bounded-complexity BSS$_{\textsf{C}}$ machine) and modify the second condition in the definition to be $\mathcal{D} \circ f_{\tau}(C_s) = \textnormal{\textsf{T}}(s)$ for all $s \in S$.
\end{remark}

\noindent Finally, we now arrive at our definition of a computational dynamical system.
\begin{definition}[Computational dynamical system]
\label{def:CDS1}
The tuple $(f, \mathcal{E}, \mathcal{D}, \tau, \textnormal{\textsf{T}})$ forms a \textbf{computational dynamical system} or \textbf{CDS} if $f$ $(\tau, t(n))$--simulates $\textnormal{\textsf{T}}$ with respect to the encoder and decoder $\mathcal{E}, \mathcal{D}$.
\end{definition}

% [[Give definitions]]

\noindent  We observe that Definitions~\ref{def:sim2} and~\ref{def:CDS1} provide a means to discuss computation in autonomous dynamical systems.

\begin{remark}[Computable manifolds]
    \label{rk:computable-manifold} It is natural to generalize Definition~\ref{def:CDS1} to a setting where the dynamics of $f$ take place on a manifold, i.e.~$f: M \to M$ is a diffeomorphism of some manifold $M$. In that case, we require that $M$ is a \emph{computable manifold}, namely we fix a finite collection of charts $\phi_i: \R^n \supset U_i \to M$ for $U_i$, $i=1, \ldots, r$ (such that the $\phi_i$'s agree on the non-trivial overlaps of the $U_i's$) and our functions $\mathcal{E}$ and $\mathcal{D}$ are now defined as compositions of BSS$_{\textsf{C}}$-computable functions to or from $M = \bigcup_i U_i$ with the charts $\phi_i$. 

    In several cases below, we will consider dynamics on a torus $f: (S^1)^n \to (S^1)^n$ or on a product of a torus with a Euclidean space; in those cases we will take our charts $\phi_i$ to be the usual charts $[0,1)^n \to (S^1)^n$.
\end{remark}

What we have achieved so far is a precise way of articulating what it means for a dynamical system $f$ to instantiate computation.  As in our previous discussions, it is essential that we keep track of the complexity of the encoder and decoder.  To this end, we note that Definition~\ref{def:optencdec1} generalizes to the CDS setting immediately:

\begin{definition}[Optimal $\mathbb{R}$-valued encoder and decoder]\label{def:optencdec2} An $\mathbb{R}$-valued encoder and decoder pair $\mathcal{E}, \mathcal{D}$ are optimal if $\mathcal{E}$ has input time complexity $\Theta(n)$ and $\mathcal{D}$ has output time complexity $\Theta(n)$.
\end{definition}

\noindent Moreover, the definition of a CDS allows us to articulate what it means for a dynamical system to be Turing-universal:
\begin{definition}[Turing-universal CDS]
\label{def:TuringCDS}
We say that a CDS $(f, \mathcal{E}, \mathcal{D}, \tau, \textnormal{\textsf{T}}_{\!\text{\rm univ}})$ is a \textbf{Turing-universal CDS} if $\textnormal{\textsf{T}}_{\!\text{\rm univ}}$ is a universal Turing machine.  Moreover, we also say that $f$ is Turing universal if there exists a Turing-universal CDS to which $f$ belongs.
\end{definition}
\noindent Various refinements of the above can be used to define a relationship between complexity classes and dynamical systems, which we discuss in Appendix~\ref{App:compclass}.  We emphasize that the above definition importantly requires that $\mathcal{E}, \mathcal{D}$ have bounded input and output complexity, respectively.  This precludes examples akin to the one given in Section~\ref{subsubsec:counterexample1}.

To be explicit, let us given a CDS generalization of the example in Section~\ref{subsubsec:counterexample1} which demonstrates that trivial dynamical systems can `simulate' universal Turing machines if we put no restrictions on the encoder and decoder.

\begin{figure}[t!]
    \centering
    \includegraphics[scale = .75]{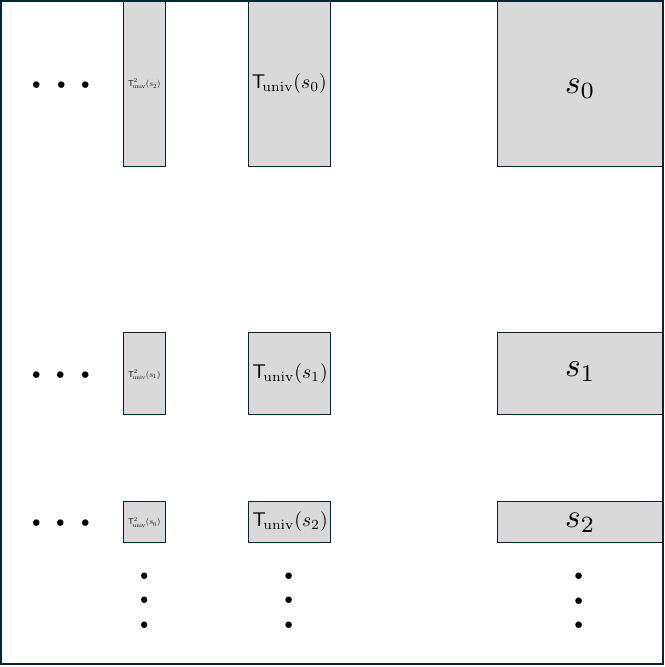}
    \caption{For Example~\ref{example:uncon1} we depict $M = [0,1]^2$, and shade in regions that encode configurations of the Turing machine.  Each such region is labelled by the configuration into which it is decoded by the $\mathcal{D}$ given in the example.}
    \label{fig:M0101v1}
\end{figure}

\begin{example}[Unconstrained encoder and decoder]\label{example:uncon1}
Let $M = [0,1]^2$, and define $f : [0,1]^2 \to [0,1]^2$ by $f(x,y) = (x/2,y)$.  Now consider a universal Turing machine $\textsf{T}_{\!\text{univ}}$ with configuration space $S$.  The $S$ is enumerable, and thus we write it as $S = \{s_0, s_1,...\}$.  Then define $\mathcal{E}(s_n) = (1,\,\frac{1}{2^n})$ and note that $f^\ell(s_n) = (\frac{1}{2^\ell},\,\frac{1}{2^n})$.  We define the decoder by
\begin{align}
\mathcal{D}^{-1}(s_n) = C_{s_n} := \bigcup_{n = 0}^\infty \,\bigcup_{\ell \,: \, \textsf{T}_{\textsf{univ}}^\ell(s_n) \,=\, s_m}\,[\frac{1}{2^\ell} - \frac{1}{2^{\ell+2}}\,,\, \frac{1}{2^\ell}] \times [\frac{1}{2^n} - \frac{1}{2^{n+2}}\,,\, \frac{1}{2^n}]\,.
\end{align}
In Figure~\ref{fig:M0101v1}, we show $M = [0,1]^2$ with regions each labelled by the configuration into which it will be decoded.  It can be checked that $\mathcal{D}$ can be implemented by a BSS$_{\textsf{C}}$ machine that does \textit{not} have bounded output complexity.  Moreover, computing $\mathcal{D}$ involves running the universal Turing machine $\textsf{T}_{\!\text{univ}}$.  The construction guarantees that $\mathcal{D} \circ f^m(C_s) = \textsf{T}_{\!\text{univ}}^m(s)$ for all $s \in S$ and all $m \in \mathbb{Z}_{\geq 0}$.  At a high level, the decoder takes in a point in its domain in $[0,1]^2$ and by examining the $y$ coordinate determines the corresponding initial configuration $s_n$ of the Turing machine.  Then the decoder runs the Turing machine for a number of steps $\ell$ determined by the $x$-coordinate, and outputs $\textsf{T}_{\!\text{univ}}^\ell(s_n)$.  As such, we see that a decoder which is `too powerful' makes trivial the notion of Turing universality for a dynamical system.  Our definitions prevent a `too powerful' encoder or decoder from arising.
\end{example}

We conclude by discussing how Definition~\ref{def:sim2} allows us to reason about what it means for one dynamical system to simulate another dynamical system.  We begin with an initial remark.

\begin{remark}[Simulating one dynamical system by another]
Suppose we have two dynamical systems $f_1 : M_1 \to M_1$ for $M_1 \subseteq \mathbb{R}^{k_1}$ and $f_2 : M_2 \to M_2$ for $M_2 \subseteq \mathbb{R}^{k_2}$.  Then Definition~\ref{def:sim2} readily generalizes to this setting if we let the encoders and decoders be e.g.~2-tape BSS$_{\textsf{C}}$ machines over $\mathbb{R}$.
\end{remark}

\noindent Since our Definition~\ref{def:sim1} of the simulation of Turing machines so closely parallels Definition~\ref{def:sim2} of the simulation of a Turing machine by an $\mathbb{R}$-valued dynamical system, it is sensible to generalize our definition of a CDS to account for both the discrete and $\mathbb{R}$-valued settings.  To this end, we have the following definition.

\begin{definition}[Discrete CDSs and sub-machines]\label{def:sub-machines1}
Suppose that $\textnormal{\textsf{T}}_2$  $(\tau, t(n))$-simulates $\textnormal{\textsf{T}}_1$ with respect to $\mathcal{E}, \mathcal{D}$, in the sense of Definition~\ref{def:sim1}.  Then we say that the tuple $(\textnormal{\textsf{T}}_2, \mathcal{E}, \mathcal{D}, \tau, \textnormal{\textsf{T}}_1)$ is a \textbf{(discrete) CDS}, and we further say that $\textnormal{\textsf{T}}_1$ is a \textbf{sub-machine} of $\textnormal{\textsf{T}}_2$.
\end{definition}

\noindent We will often omit the word `discrete' in `discrete CDS' when the context is clear.

There is a very useful relationship between CDSs and sub-machines.  We will state one version of this relationship with the following definition and lemma.
\begin{definition}
Let $f : M \to M$ by a dynamical system, and let $\tau_1, \tau_2 : M \rightharpoonup \mathbb{Z}_{\geq 0}$ be slowdown functions then we define the $\tau_1 \!\times_{\! f}\! \tau_2 := \tau_1 \circ f^{\tau_2}$, where $f^\tau_2 : x \mapsto f^{\tau_2(x)}(x)$.  Then we have $f^{\tau_1 \times_{\! f} \tau_2} = f^{\tau_1} \circ f^{\tau_2}$.
\end{definition}
\noindent We note that if $\tau_1$ and $\tau_2$ are computable by hybrid BSS$_{\textsf{C}}$ machines, then $\tau_1 \!\times_{\! f}\! \tau_2$ may not be computable contingent on $f$.  This fact will not affect any of our results.
\begin{lemma}[Fundamental lemma of sub-machines]\label{lemm:fund1}
Let $(f, \mathcal{E}, \mathcal{D}, \tau, \textnormal{\textsf{T}}_{2})$ be a CDS with encoder and decoder complexity $O(t(n))$, and let $(\textnormal{\textsf{T}}_{2}, \widetilde{\mathcal{E}}, \widetilde{\mathcal{D}}, \tilde{\tau}, \textnormal{\textsf{T}}_{1})$ be a discrete CDS with encoder and decoder complexity $O(\tilde{t}(n))$.  Then $(f, \mathcal{E} \circ \widetilde{\mathcal{E}}, \widetilde{\mathcal{D}} \circ \mathcal{D}, \tau \times_{\! f}  \tilde{\tau}, \,\textnormal{\textsf{T}}_{1})$ is a CDS with encoder and decoder complexity $O(\max\{t(n),\,\tilde{t}(n)\})$.

On the other hand, suppose that a given $f$ cannot be extended to a CDS for a Turing machine $\textnormal{\textsf{T}}_1$.  Then $f$ likewise cannot be extended to a CDS for any $\textnormal{\textsf{T}}_2$ for which $\textnormal{\textsf{T}}_1$ is a sub-machine.
\end{lemma}

\begin{proof}
For the first part, suppose that $\textsf{T}_1 : S_1 \to S_1$, $\textsf{T}_2 : S_2 \to S_2$, and $\widetilde{C}_s := \widetilde{\mathcal{D}}^{-1}(s)$ for $s \in S_1$.  Let us show that $(f,\,\mathcal{E} \circ \widetilde{\mathcal{E}}, \widetilde{\mathcal{D}} \circ \mathcal{D},\, \tau \times_{\! f} \tilde{\tau},\, \textnormal{\textsf{T}}_{1})$ is a CDS with encoder and decoder complexity $O(\max\{t(n),\,\tilde{t}(n)\})$.  By composition, $\mathcal{E} \circ \widetilde{\mathcal{E}} \in C_s$ for all $s \in S_1$, where $C_s := (\widetilde{\mathcal{D}} \circ \mathcal{D})^{-1}(s)$, and
\begin{align}
(\widetilde{\mathcal{D}} \circ \mathcal{D}) \circ f^{\tau \times_{\! f}  \widetilde{\tau}}(C_s) &= \widetilde{\mathcal{D}} \circ \textsf{T}_2^{\tilde{\tau}}(\widetilde{C}_s) = \textsf{T}_1(s) 
\end{align}
holds for all $s \in S_1$.  Moreover, the complexities of $\mathcal{E} \circ \widetilde{\mathcal{E}}$ and $\widetilde{\mathcal{D}} \circ \mathcal{D}$ are $O(\max\{t(n),\,\tilde{t}(n)\})$ since composition of functions corresponds to additivity of complexities.  This completes the first part of the proof.

For the second part, suppose by contradiction that $f$ cannot furnish a CDS for a $\textsf{T}_1$, but it can furnish a CDS for a $\textsf{T}_2$ such that $\textsf{T}_1$ is a sub-machine.  These statements imply the existence of CDSs of the form $(f,\,\mathcal{E} \circ \widetilde{\mathcal{E}}, \widetilde{\mathcal{D}} \circ \mathcal{D},\,\tau \times_{\! f}  \tilde{\tau},\, \textnormal{\textsf{T}}_{1})$ and $(\textnormal{\textsf{T}}_{2}, \widetilde{\mathcal{E}}, \widetilde{\mathcal{D}},\,\tilde{\tau}, \textnormal{\textsf{T}}_{1})$, but then by the first part of the proof we can obtain a new CDS for which $f$ simulates $\textsf{T}_1$, which is a contradiction.
\end{proof}

The above Lemma~\ref{lemm:fund1} will be used throughout the paper in the following way.  We will often prove that some dynamical system $f$ cannot be extended to a CDS of a machine $\textsf{T}_1$.  But then the second part of Lemma~\ref{lemm:fund1} implies that $f$ likewise cannot be extended to a CDS for any machines $\textsf{T}_2$ for which $\textsf{T}_1$ is a sub-machine.  So, for instance, if an $f$ cannot be extended to a CDS for some Turing machine $\textsf{T}_1$, then it certainly cannot be extended to a CDS for a universal Turing machine $\textsf{T}_{\!\text{univ}}$ for which $\textsf{T}_1$ is a sub-machine by virtue of universality.

\begin{remark}[A refinement of Lemma~\ref{lemm:fund1}] Suppose in the second part of Lemma~\ref{lemm:fund1} that we merely stipulated that $f$ cannot be extended to a CDS for $\textsf{T}_1$ with encoder and decoder complexity $O(\bar{t}(n))$ for some $\bar{t}(n)$.  Then a similar argument as in the proof of Lemma~\ref{lemm:fund1} establishes that if $f$ can furnish a CDS of a $\textsf{T}_2$ such that $\textsf{T}_1$ is a sub-machine, then the encoder and decoder complexity $t(n)$ of any $(f,\,\mathcal{E} \circ \widetilde{\mathcal{E}}, \widetilde{\mathcal{D}} \circ \mathcal{D},\,\tau \times_f \tilde{\tau},\, \textnormal{\textsf{T}}_{1})$ and $\bar{t}(n)$ of any $(\textnormal{\textsf{T}}_{2}, \widetilde{\mathcal{E}}, \widetilde{\mathcal{D}}, \tilde{\tau},\,\textnormal{\textsf{T}}_{1})$ satisfies $\max\{t(n),\,\tilde{t}(n)\} \geq \Omega(\bar{t}(n))$.
\end{remark}

\subsubsection{Robustness and other conditions for CDS encoders and decoders.}
\label{sec:robustness}

We now turn to a formalization of \emph{robust computation}. This model of robust computation is different from other previous works \cite{braverman2012noise, gracca2005robust}, as described in Section \ref{sec:related-work}. Recall that, as in Remark \ref{rk:explanation-of-simulation-condition}, given a CDS, we are simulating $\textsf{T}$ by encoding configurations $s \in S$ using the encoder $\mathcal{E}$, running $f^\tau$ for each computational step, and then decoding the configurations via $\mathcal{D}$. In many previous works on implementing computation in continuous dynamical systems (see Section \ref{sec:related-work}), one only requires that the dynamics is `correct' on the image of $\mathcal{E}$, in other words, one has that 
\[ \mathcal{D}^{-1}(s) = \{\mathcal{E}(s)\} \,\text{ for all }\, s \in S.\]
We might call such a decoder a `point decoder'; if this is the decoder underlying a CDS, then clearly the simulation of the underlying machine $\machine{T}$ is not robust in any sense.

One sense of robustness that one might wish for is that a small margin of error in the encoding still leads to a valid simulation.  That is, for every $s \in S$, there is some error $\epsilon_s$ such that for any $x_s \in M$ of distance at most $\epsilon_s$ from $\mathcal{E}(s)$, the simulation equation is satisfied:
\begin{equation}
    \label{eq:simulation-equation-with-error}
    \mathcal{D} \circ f^{\tau}(x_s) = \machine{T}(s) \,\text{ for all }\, s \in S\,.
\end{equation}
For technical convenience, it is easier to formulate an analogous property entirely in terms of the behavior of $f$ on the sets $C_s$ defined by the decoder $\mathcal{D}$, as follows:

\begin{definition}
\label{def:robust-decoder}
    We say that a decoder $\mathcal{D}: M \rightharpoonup S$ is \textbf{robust} when $\mathcal{D}^{-1}(s)$ is the closure of its own interior, for every $s \in S$. We say that a CDS is robust when the underlying decoder $\mathcal{D}$ is robust, and when $\mathcal{E}(s)$ lies in the interior of $\mathcal{D}^{-1}(s)$ for every $s$.  
\end{definition}

It is clear that that given a robust CDS, if we set $\epsilon_s = \sup_{x \in C_s} d(\mathcal{E}(s), x) > 0$ then the simulation equation \eqref{eq:simulation-equation-with-error} is satisfied; and conversely one can show that there exists a metric $d$ on $M$ such that the closed ball of radius $\epsilon_s$ around $\mathcal{E}(s)$ is exactly $\mathcal{D}^{-1}(s)$. Thus, our model of robustness is analogous to requiring robustness of the simulation after perturbing $f$ by a noise term with magnitude dependent on $x \in M$.

Below, we give two more conditions on decoders that are helpful when formulating certain results:

\begin{definition}
\label{def:shrinking-decoders-and-so-forth}
    We say that a decoder $\mathcal{D}: M\rightharpoonup S$ is \textbf{proper} when $\mathcal{D}^{-1}(s)$ is a compact set for every $s \in S$. We say that the decoder $\mathcal{D}$ is \textbf{shrinking} when it is proper and when the diameter 
    \[ \text{\rm{diam}}(\mathcal{D}^{-1}(s)) := \sup_{(x,y) \in \mathcal{D}^{-1}(s)^2} d(x,y)\]
    shrinks to zero as the length of $s$ goes to infinity, i.e.~for every sequence of strings $s_i \in S$ such that the lengths $|s_i|$ diverge to infinity, we have that $\lim_{i \to \infty} \rm{diam}(\mathcal{D}^{-1}(s_i)) = 0$.

    If $S$ is the set of configurations of a Turing machine, we say that the decoder is \textbf{hierarchically shrinking} when for every configuration $s= (s_a, s_b)$ where $s_a$ is the internal configuration of the turing machine and $s_b$ is the state of the tape, there exist compact subsets $C'_s \subset M$ satisfying the following condition. Let $s' = (s'_a, s'_b)$ be a configuration of the Turing machine for which the internal state agrees with $s$ (so $s'_a = s_a$) and the region of the tape around the head agrees with $s_b$ (so $s'_a$ is a substring of $s'_b$, and the head of both configurations lies in $s'_a$). Then, the condition is that $C'_s \supset C_{s'}$; and such that the following two conditions hold:
    \begin{enumerate}
        \item $\rm{diam}(C'_s) \to 0$ as the length of $s_b$ diverges; and 
        \item Writing $s_1 = (s_{a_1}, s_{b_1})$ and $s_2 = (s_{a_2}, s_{b_2})$ for a pair of configurations of the Turing machine, we have that either 
        $C'_{s_1}$ is disjoint from $C'_{s_2}$, or that either $C'_{s_1} \subset C'_{s_2}$ or $C'_{s_2} \subset C'_{s_1}$; and that $C'_{s_1} \subset C'_{s_2}$ exactly when $s_{a_1} = s_{a_2}$ and $s_{b_2}$ is is a substring of $s_{b_1}$ with the head of the Turing machine lying in $s_{b_2}$. 
    \end{enumerate}
\end{definition}

\begin{remark}
\label{rk:robust-universal-system-can-be-shrinking}
    The decoder in Section \ref{sec:example-of-robust-cds} associated to a robustly Turing-universal CDS is both shrinking and hierarchically shrinking. 
\end{remark}

\begin{remark}
    It is easy to see that a hierarchically shrinking decoder is shrinking.
\end{remark}

\subsection{Overview of dynamical systems}
\label{subsec:overviewofdynsys}

In this section, we will review basic notions about continuous and differentiable dynamical systems theory. These concepts will be used in the subsequent constructions and proofs, and are very helpful for expressing properties of dynamical systems. We will also briefly review several standard examples of simple dynamical systems which are used as sources of intuition for the possible types of dynamical behavior. 

\paragraph{Basic notions.}

We have repeatedly used the well-known concept of a (discrete-time) \emph{dynamical system}, which is simply a map $f: M \to M$ for some set $M$. We say $f$ system is \emph{reversible} when $f^{-1}$ exists. For any point $x \in M$, we can consider the \emph{forwards orbit} $\{x, f(x), f^2(x), f^3(x), \ldots\}$ of $x$ and if $f$ is reversible, we can consider its \emph{orbit} $\{\ldots, f^{-2}(x), f^{-1}(x), x, f(x), f^2(x), \ldots\}$  or its \emph{backwards orbit} $\{\ldots, f^{-2}(x), f^{-1}(x), x\}$.  Taking a union of forwards orbits produces an \emph{invariant set}: a subset $A \subset M$ such that $f(A) \subset A$. (Note that this does not imply that $f(A) = A$!) 

Many of the basic questions regarding dynamical systems are about the decomposition of $M$ into orbits, and the action of $f$ on the orbits themselves. In particular, the simplest kinds of orbits are the orbits of \emph{periodic points}, i.e.~of $x \in M$ such that $f^n(x) = x$ for some $n$. One often hopes to understand the global dynamics of $f$ by first understanding its periodic points. 

Usually, the set $M$ carries more structure, and $f$ is compatible with this structure in a natural way. For example, $M$ may be a topological space or a metric space and $f$ may be continuous, in which case we are studying \emph{topological dynamics}; or $M$ may be a manifold and $f$ may be differentiable with some degree of regularity, i.e.~$f$ may be once-differentiable, twice-differentiable, H\"older-continuous, or smooth. Alternatively, $M$ may be a measurable space and $f$ may preserve a measure $\mu$ on this space. 

\begin{remark}
A phenomenon that may be surprising to those new to differentiable dynamics is that smooth dynamical systems produce associated objects which have a much lower degree of regularity. For example, when studying hyperbolic dynamical systems, as we will in Section \ref{sec:axiom-a}, the stable and unstable manifolds associated to points $x$ in hyperbolic invariant sets only vary H\"older continuously with $x$ \cite{katok1995introduction}, even when the underlying dynamics are smooth. Similarly, certain fundamental results like the closing lemma \cite{pugh1983c1} are known to be true for $C^1$ small perturbations but not for $C^\infty$-small perturbations, and with energy-conservation constraints on $f$ they may even be false \cite{herman1991exemples}. This sensitivity of dynamical results to the regularity of the mathematical objects under consideration often underlies deep dynamical structure, but the relevance of such phenomena to scientific and modeling problems is still unclear. 
\end{remark}

\paragraph{Discrete- and continuous-time dynamics.}
Instead of the discrete-time systems $f: M \to M$ defined above, many physical systems are instead specified by a differential equation 
\[ \frac{dx}{dt} = g(x,t), \quad x \in M, \quad g: M \times \R \to TM\,,\]
where $g$ is a time-dependent vector field on a manifold $M$. The analog of the map $f$ above is then the \emph{continuous-time} dynamical system
$f_t: M \to M$ for $t \in \R$, associated to solutions of the differential equation above. Thus, $f$ is specified by the property that $\frac{\partial}{\partial t}f_t(x) = g(x,t)$. When $g$ is $t$-independent or $1$-periodic (i.e.~$g(x,t) = g(x,t+1)$), one can extract a discrete-time system from this continuous-time system by specifying $f: M \to M$ via $f(x) = f(x,1)$; this will satisfy the property that $f^2(x) = f(x,2)$, and so forth, and thus we will be studying the behavior of the flow $f(x,t)$ of the defining differential equation at a discrete set of times. Then the dynamics of $f: M \to M$ corresponds to viewing the dynamics specified by the ODE viewed at a discrete set of times, as in Remark \ref{rem:modification1}.

\begin{figure}[t!]
    \centering
    \includegraphics[scale = .5]{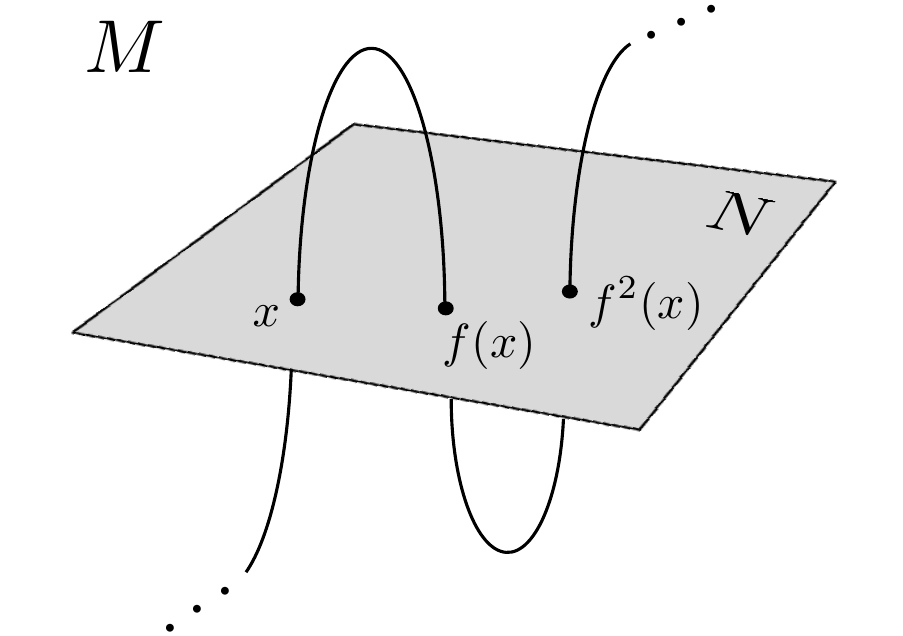}
    \caption{A sketch of the Poincar\'{e} section of a map.  We consider time-dependent dynamics $f_t : M \to M$ with $f_0 = \text{Id}$.  A Poincar\'{e} section of the map is a submanifold $N$ of $M$ such that for any $x \in N$ and any $k \in \mathbb{Z}_{\geq 0}$, we can define $f^k(x)$ as the $k$th time the trajectory $f_t(x)$ intersects with $N$ (which we require to occur for each $k$).}
    \label{fig:poincare1}
\end{figure}

\begin{remark}
Definitions made for discrete-time systems usually have analogs for continuous-time systems. In this paper, we will largely stick to discrete-time systems because of their comparative mathematical simplicity. It is worth noting that the behavior of discrete-time systems in $n$-dimensions models in many ways resembles the behavior of continuous time-independent systems (i.e.~with $g(x,t) = g(x)$) in $(n+1)$-dimensions, because in many cases one can find a \emph{Poincar\'{e} section} for the dynamics.  A Poincar\'{e} section is a hypersurface $N \subset M$ such that $g(x)$ is transverse to $N$ for all $x$.  Then one defines a discrete dynamical system $f: N \to N$ by setting $f(x) =  f_{t^*}(x)$, where $t^* = \inf \{t > 0\,:\, f_t(x) \in N\}$.  See Figure~\ref{fig:poincare1} for a depiction.
\end{remark} 

\paragraph{Notions of transitivity, ergodicity, and chaos.}
There are several basic notions that help to articulate and characterize the behavior dynamical systems, and which we will use repeatedly throughout this paper. We quickly introduce them here, and then illustrate them with standard examples in the subsequent section. Specifically, there are various ways of discussing the extent to which the dynamics of $f$ `mix up' certain regions in $M$, all of which highlight different aspects of the dynamical behavior of $f$. 

Let $f:M \to M$ be a continuous dynamical system on a  metric space $M$.  

\begin{definition}[Topological transitivity]
Let $C \subset M$ be closed and $f$-invariant, i.e.~$f(C) \subset C$.  We say that $f$ is \textbf{topologically transitive} on $C$ if there is an $x \in C$ such that the orbit of $x$ is dense in $C$. This implies that $f(C) = C$. 
\end{definition} 

\noindent With  notation as in the previous definition, we write 
\begin{equation}
    \label{eq:stable-set}
    W^s(C) = \left\{ x \in M \,:\, \lim_{n \to \infty} d(f^n(x), C) = 0\right\},\; W^u(C) = \left\{ x \in M \,:\, \lim_{n \to \infty} d(f^{-n}(x), C) = 0\right\}\,.
\end{equation}
The set $W^s(C)$ \emph{stable set} of $C$; points in $W^s(C)$ are eventually attracted to $C$.  Similarly, $W^u(C)$ is the \emph{unstable set} of $C$. One may want to find invariant sets $C$ which can be thought of as `attractors' \cite{smale1967differentiable}, such that the entire dynamics is decomposed, in some sense, into the interaction between the attractors of the dynamics. 

\begin{definition}[Topologically mixing]
Let $C \subset M$ be closed and $f$-invariant. We say that $f$ is \textbf{topologically mixing} on $C$ if for every two open sets $U$ and $V$ in $C$ (in the subspace topology) there is an $N$ such that for all $n>N$, $f^n(U) \cap V \neq \emptyset$. 
\end{definition}

\begin{lemma}[\cite{katok1995introduction}]
    If $X$ is compact and $f$ is topologically mixing on $C$ then $f$ is topologically transitive on $C$.
\end{lemma}

There are yet stronger notions of mixing besides the topological ones described above. To define these one asks for the existence of a Borel probability measure $\mu$ such that $f^*\mu =\mu$, i.e.~$\mu$ is an \emph{invariant measure} for $f$. In this case, $(f, X, \mu)$ is a \emph{continuous, measure-preserving system.}
\begin{definition}[Ergodic]
    Given a continuous, measure-preserving system $(f, X, \mu)$, we say that $f$ is \textbf{ergodic} if for any $A \subset X$ such that $f^{-1}(A) = A$ we have that $\mu(A)=0$ or $\mu(A)=1$. 
\end{definition}
\begin{remark}
    One can speak of course of discontinuous measure-preserving systems and make such definitions in that context, but such examples will not be the focus of this work.
\end{remark}

\noindent It is easily shown that if $f$ is ergodic then $\mu$-almost-all points of $X$ have dense orbits in $X$, and thus that $f$ is topologically transitive on $X$. However, ergodicity does not imply topological mixing  of $f$ on $X$, while there are various measure-theoretic notions of mixing that do imply topological mixing of $f$ on $X$ \cite{katok1995introduction}.

Standard examples, reviewed below, show that while mixing and ergodicity arise in chaotic systems, they are not necessary hallmarks of chaotic behavior. Instead, chaotic behavior is better quantified using other concepts. One important notion is that of \emph{hyperbolicity}: the idea that the dynamics of $f$ are sensitive to initial conditions, i.e.~that small errors in initial condition get exponentially larger as time goes on. An invariant set on which there is a quantitative bound on the amount of such sensitivity is called a \emph{hyperbolic set}:

\begin{definition}[Hyperbolic sets]
\label{def:hyperbolic-set}
Let $C \subset M$ be closed and $f$-invariant. We say that $C$ is \textbf{hyperbolic} if there there is a decomposition 
\[ TM|_{C} = T^+_C \oplus T^-_C \]
together with constants $c, d> 0$, $\lambda > 1$, such that 
\begin{equation}
    \label{eq:instability-equation-for-hyperbolic-set}\|df^n(v)\| > c \lambda^n \|v\|\,\text{ for }\,v \in T^+_C\,,
\end{equation}
and $\|df^n(v)\| < d \lambda^{-n}\|v\|$ for $v \in T^-_C$, where in each of these statements $n$ runs over the natural numbers. This property does not depend on the choice of Riemannian metric that we are using to measure the lengths of tangent vectors. 
\end{definition} 

\begin{remark}
    Note that if $T^-_C=0$, which is allowed in the definition of a hyperbolic set, there is indeed no sensitivity to initial conditions. In particular, if $f$ is the time-$1$ gradient flow of some Morse function $g$, then a minimum of $g$ is a hyperbolic set for $f$. Nonetheless, in typical examples $T^+_C \neq \emptyset$.
\end{remark}

\noindent Clearly if $T^+_C \neq 0$ then the maximal constant $\lambda$ such that the condition \eqref{eq:instability-equation-for-hyperbolic-set} holds is a measure of the instability of the dynamics of $f$ on $C$. The measure of infinitesimal instability of $f$ along a trajectory of $f$ is the \emph{Lyapunov exponent}:

\[ L(f, x) = \sup_{v \in T_x} \lim_{n \to \infty} \frac{ \log \|Df^n(v)\|}{n}.  \]
Here one uses any Riemannan metric on $M$ to measure the lengths of vectors, and this quantity does not depend on the metric chosen. 
 It is a fundamental theorem of Oseledets \cite{oseledets1968multiplicative} that given an invariant measure $\mu$ for $f$, the Lyapunov exponents are the same for $\mu$-almost-every $x \in M$. (One can define  Lyapunov exponents as functions of $v$ instead, in which case one has several Lyapunov exponents for each $x \in M$; \cite{oseledets1968multiplicative} states that this set of exponents is the same for $\mu$-almost everywhere $x$.)

\paragraph{Fundamental examples.}
Basic intuition about the behavior of dynamical systems is captured in well-known examples. Here we give a lightning review of what may be considered from a mathematical perspective to be the simplest kinds of dynamical systems; see Figure~\ref{fig:manysystems} for a depiction. The rest of this paper focuses on proving some results that try to analyze the \emph{computational capabilities} of natural generalizations of these simplest examples. 

\begin{figure}[t!]
    \centering
    \includegraphics[scale = .5]{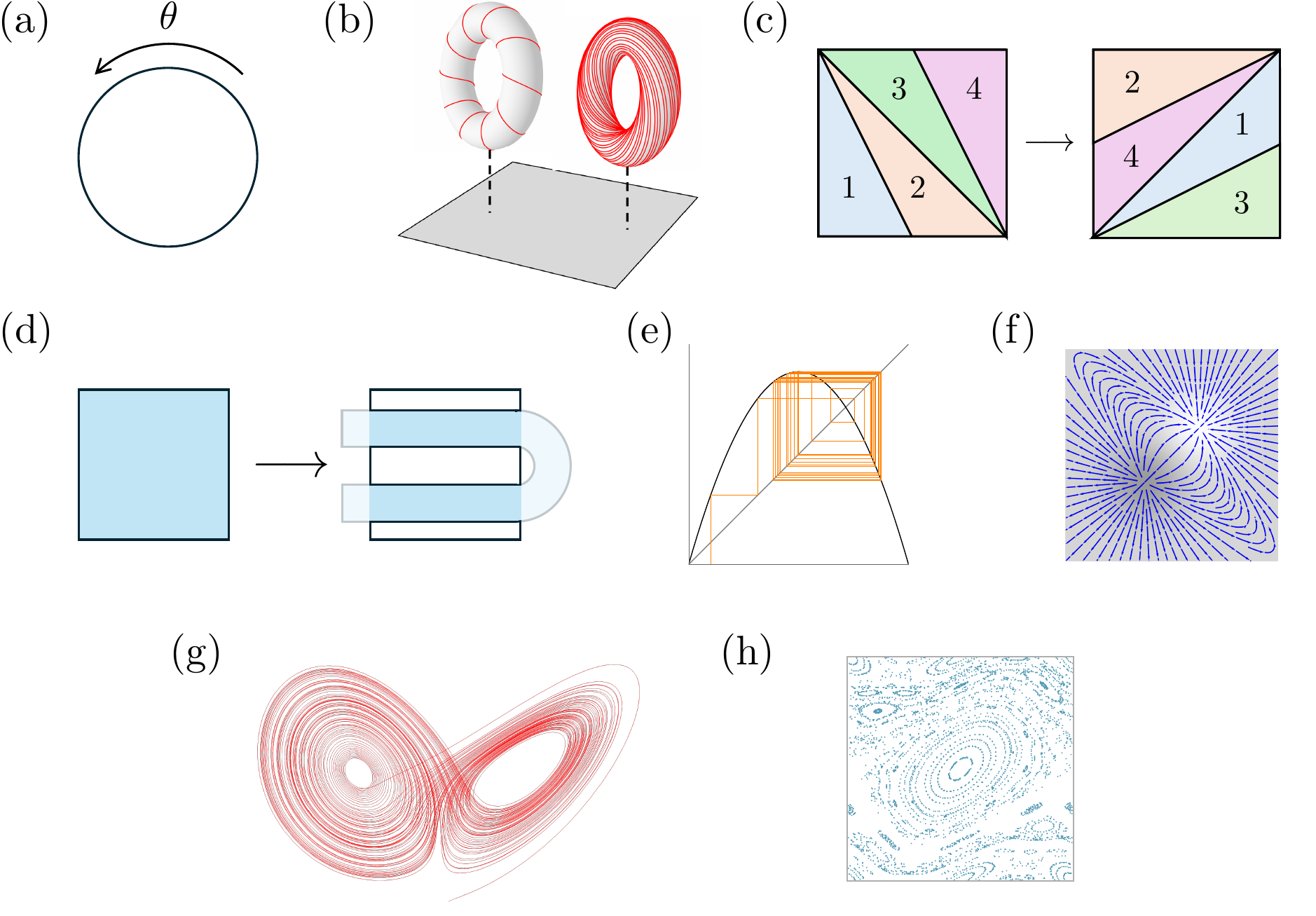}
    \caption{Various fundamental examples in dynamical systems theory. (a) A circle map. (b) An integrable system represented as a base manifold with tori fibered over it.  The red curves are either periodic or irrational orbits, contingent on the corresponding point on the base manifold. (c) The Arnold cat map. (d) One iteration of the Smale horseshoe map. (e) A member of the quadratic family.  Depicted is a cobweb diagram for an orbit. (f) A gradient flow system, with various flow lines depicted.  (g) An orbit of the Lorenz system, tending towards the Lorenz attractor.  (h) An instance of the standard map, with various orbits displayed. }
    \label{fig:manysystems}
\end{figure}

\begin{enumerate}
\item \textbf{Circle rotations.} This is a dynamical system $f: S^1 \to S^1$ which rotates the circle by a fixed angle. Parameterizing $S^1$ as $\mathbb{R}/\mathbb{Z}$, the real numbers modulo $1$, a circle rotation is given by the map $t_\tau(x) = x + \tau$ for some $\tau$. We can organize circle rotations into the \emph{rational} rotations for which $\tau \in \mathbb{Q}/\mathbb{Z}$, and the complement of the rational rotations, namely the irrational ones. Rational rotations are periodic (a power of $f$ is the identity) and so are not topologically transitive; while irrational rotations are topologically transitive, ergodic (with respect to the uniform measure on $S^1$), but not topologically mixing. 
\item \textbf{Torus translations and integrable systems.} One can generalize circle rotations to \emph{torus translations}, which are maps $f: (S^1)^n \to (S^1)^n$ given by $f(x) = x+v$ for some vector $v \in \R^n$. These have similar behavior to circle rotations but are higher-dimensional; moreover, they are pieces of the dynamics of \emph{integrable systems}, which are  the (time-1 flows) of Hamiltonian dynamical systems associated to $H: M \to \R$ where $M$ is a manifold with a symplectic form $\omega$ such that $H$ is part of a system of pairwise Poisson-commuting Hamiltonians $(H_1, \ldots, H_n)$ with $H_1 = H$. The latter arise in many important physical problems, e.g.~the Euler equations for a a free rotating body in $3$d \cite{arnol2013mathematical}. The result connecting these two settings is the Liouville-Arnold theorem \cite{arnol2013mathematical}, which states that if there is a point $x \in M$ such that $(dH_i)_x$ is nonzero for each $H_i$, then there is a neighborhood $U$ of $x$ and a diffeomorphism $\psi: U \to (S^1)^n \times V$ where $V$ is some open subset in $\R^n$ and the dynamics satisfy $\psi \circ f\circ \psi^{-1}(a, b) = a+b$. Later in this paper, we will strongly constrain the computations that can be robustly performed by integrable systems; see Theorem \ref{thm:measure-preserving-not-universal}. The Lyapunov exponent computed with respect to any $x \in M$ when $f$ is a torus translaton or an integrable system is zero.

\item \textbf{Doubling map.} A qualitatively different type of dynamical system arises by replacing the \emph{translation} defining a circle map $f: S^1 \to S^1$ with a \emph{multiplication} of the coordinate. Again parameterizing $S^1 = \mathbb{R}/\mathbb{Z}$, let $f(x) = 2x$. This is called the \emph{doubling map}, and is the simplest Anosov map. Writing real numbers $x \in [0,1)$ in their $2$-adic expansions $x = 0.a_1a_2a_3\ldots$, we have that the $2$-adic expansion of $f(x)$ is $0.a_2a_3a_4\ldots$; thus, the dynamics of $f$ is essentially that of the \emph{shift map} on the set $\{0,1\}^{\mathbb{N}}$ of sequences of binary digits which simply drops the first element of a half-infinite string of digits. One sees immediately from this description that $f$ is topologically mixing, since $f^n(I)$ contains $(0,1)$ for any nonempty interval $I$ and any sufficiently large $n$. In fact, all of $S^1$ is an invariant closed hyperbolic set, the map preserves the uniform measure on $S^1$, and the Lyapunov exponent with respect to this measure is $\log 2$. The description of the dynamics of $f$ in terms of the shift map on decimal sequences is the simplest instance of \emph{symbolic dynamics} (further reviewed in Appendix \ref{app:symbolic-dynamics}).

\item \textbf{Anosov maps.} One drawback of the doubling map is that it is not invertible. However, the non-invertiblity can be removed by passing to higher dimensions. For example, the \emph{Arnold cat map}, which is a diffeomorphism $f: T^2 \to T^2$ of the two-torus $T^2$, is defined via 
\[ f(x,y) = (2x + y, x+y), \text{ where }\, T^2 = S^1 \times S^1 = (\mathbb{R}/\mathbb{Z})^2. \]
This map has a representation via symbolic dynamics, where one divides $T^2$ up into certain blocks consisting of certain symbols, such that almost all trajectories of the cat map can be understood in terms of the restriction of the shift map to a simple subset of $\{0,1\}^{\mathbb{Z}}$. All of $T^2$ is a hyperbolic invariant set, there is a $1$-dimensional stable bundle and a $1$-dimensional unstable bundle, and the Lyapunov exponent (with respect to the area measure on $T^2$, which is an invariant measure) is $(3 + \sqrt{5})/2$. More generally, systems for which the entire domain is a single hyperbolic transitive invariant set are called \emph{Anosov}, and are the prototypical examples of systems which are \emph{strongly chaotic} -- so chaotic that their dynamics, while unpredictable, is considered by researchers in dynamical systems theory to be completely understood  \cite[Part 4]{katok1995introduction}.
\item \textbf{1-dimensional dynamics.} Beyond Anosov maps, it is fruitful to first study low-dimensional examples. There is a rich theory describing the dynamics of $1$-dimensional maps $f:S^1 \to S^1$ or maps $f: [-\beta, \beta] \to [0,1]$. Famously, the quadratic family of maps $f_a(x) = a - x^2$ has an invariant interval $[-\beta, \beta]$, and as discovered by Feigenbaum \cite{feigenbaum1978quantitative} the resulting maps show a transition from periodic dynamics to chaotic dynamics as one increases $a$ past a critical threshold. There is an elaborate theory of smooth maps of the interval without flat critical points \cite{de2012one}. 
\item \textbf{Gradient flows.} Generalizing in a different direction, a simple higher-dimensional example of a nonlinear dynamical system is the time-1 flow of the gradient flow of a potential function $U: M \to \R$. The invariant sets are all contained in the locus of critical points $\text{Crit}(U)$ of $U$. Moreover, when $U$ is Morse, i.e.~when the Hessian of $U$ has full rank at every critical point of $U$, the invariant sets are the (isolated) critical points, and each critical point comprises a hyperbolic invariant set. In that case, $M$ becomes decomposed into regions which are attracted to the various critical points: we have 
\[ M = \bigsqcup_{p \,\in\, \text{Crit}(U)} W^s(\{p\}),\]
and there is an associated graph that describes possible gradient flows from one critical point to another. 

In Section \ref{sec:axiom-a} we will strongly constrain computational properties of the class of systems which encompass the features of gradient flows and of Anosov diffeomorphisms, called \emph{Axiom A systems}, and include both as special cases. 
\item \textbf{Beyond the simplest examples.} Outside of mild generalizations of the class of systems described above, the behavior of differentiable dynamical systems becomes extremely complex and difficult to analyze mathematically. For example, there are non-uniformly hyperbolic systems \cite{pesin2010open}, which include famous examples such as the Lorenz system \cite{lorenz1963deterministic}, and there is the puzzling behavior of essentially all non-integrable Hamiltonian systems, even small perturbations of integrable examples. The latter type of behavior is modeled by the \emph{standard map}, and the related H\`enon family \cite{young1998developments}, which can contain Lorenz-like attractors. Many basic questions, such as the existence of any values of the parameter of the standard map for which the Lyapunov exponent is nonnegative, are famous and difficult open problems in the theory of differentiable dynamical systems \cite{pesin2010open}.
\end{enumerate}

We wish to use the notion of a computational dynamical system introduced in Section \ref{sec:defining-CDS} to ask questions about the computational capacity of each of these types of dynamical systems.

\section{Autonomous CDSs}

In this section, we employ our definition of a CDS to examine the topological and geometric characteristics of autonomous dynamical systems that either facilitate or prevent their ability to simulate universal Turing machines.  To demonstrate the possibility of robustly Turing-universal CDSs, we begin by constructing an example on a disk.  Next we study Axiom A dynamical systems as a well-studied class of `chaotic'  systems, and prove that they cannot be extended to robustly Turing-universal CDSs.  We the prove another non-universality result for measure-preserving systems and thus for integrable systems.  Taken all together, our non-universality articulate a sense in which `chaos' and `order' are in tension with Turing-universality.  Thereafter we prove more refined statements about how mixing in certain dynamical systems constrains the computational complexity of the Turing machines it can realize as a CDS.

\subsection{Example of a Turing-complete CDS}
\label{sec:example-of-robust-cds}
Here we construct a robustly Turing-universal CDS $(f,\mathcal{E}, \mathcal{D},\tau,\textsf{T}_{\!\text{univ}})$ where $f : D^2 \to D^2$ is a smooth map on the disk, $\tau(x) = 1$, and $\mathcal{E}, \mathcal{D}$ have optimal complexity $\Theta(n)$.  Our construction is inspired in part by~\cite{moore1991generalized} and~\cite{cardona2021constructing}.  The latter work constructs an area-preserving, smooth map $f : D^2 \to D^2$ which forms a Turing-universal CDS, but as discussed previously, their map is not robust in our sense.  (In fact, as we show later in Section~\ref{subsec:nonunivmeas1}, it is not possible for an area-preserving map to be extended to a robustly Turing-universal CDS.)  To make a robust version of the construction in~\cite{cardona2021constructing}, some new ideas are involved which we will explain shortly.

For our purposes, we will consider a single-tape universal machine $\textsf{T}_{\!\text{univ}}$ given by the triple $(Q, \Gamma, \delta)$ where $\Gamma \cup \{\sqcup\} = \{\sqcup, 0, 1, 2\}$, and $|Q| = 2^k$ for some $k$.  Recalling that $\sqcup$ is the blank symbol, it will be useful to associate $\sqcup \to 00$, $0 \to 01$, $1 \to 10$, $2 \to 11$, so that we can write $\Gamma \cup \{\sqcup\} = \{00, 01, 10, 11\}$.  Moreover we will label each state $q \in Q$ by a $k$-bit string $z \in \{0,1\}^k$.  The Turing machine $\textsf{T}_{\!\text{univ}}$ can read blank symbols, but it cannot write blank symbols.  The basic strategy in constructing a CDS for this $\textsf{T}_{\!\text{univ}}$ is to re-express $\textsf{T}_{\!\text{univ}}$ in terms of a generalized shift map.  To do so, we need to establish some notation, following~\cite{moore1991generalized}.

Let $A$ be some finite alphabet and let $\Sigma = A^{\mathbb{Z}}$ be the set of bi-infinite sequences of $A$.  We further let $\sigma : \Sigma \to \Sigma$ be the shift map, such that for $s \in \Sigma$ we have $\sigma(s)$ denote $s$ shifted one spot to the right.  We will use $\sigma$ as an ingredient to define `generalized shift maps' below, but first we require a useful definition.

\begin{definition}[Domain of effect and domain of dependence, adapted from~\cite{moore1991generalized}] 
Let $h : \Sigma \to \Sigma$ where we write $\Sigma = \prod_{i \in \mathbb{Z}} A_i$.  Let $E^e$ be the subset of $\mathbb{Z}$ containing all $j$ such that $h(s)_j \not = s_j$ for at least one $s$.  Here $h(s)_j$ is the $j$th coordinates of $h(s)$, and similarly for $s_j$.  Then the \textbf{domain of effect} of $h$ is $D^e := \prod_{i \in E^e} A_i$.  Informally, the domain of effect of $h$ contains the $A_i$'s in $\Sigma$ which are affected by the inputs to $h$.

Moreover, let $E^d$ be the smallest subset $\mathbb{Z}$ containing all $j$ such that $h$ can be written as 
\[ h(s)_j = s_j \,\text{ for }\, j \notin E^e,\quad h(s)_j = \bar{h}_j((s_i)_{i \in E^d}) \,\text{ for }\, j \in E^e.\]
Then the \textbf{domain of dependence} of $h$ is $D^d := \prod_{i \in E^d} A_i$. More informally, the domain of dependence of $h$ contains the $A_i$'s in $\Sigma$ on which outputs of $h$ depend.

Similarly, if $h$ is a map $\Sigma \to \Z$ then the domain of dependence is $\prod_{i \in E^d} A_i$ where $E^d$ is the smallest subset of $\Z$ such that $h$ factors through the projection to $\prod_{i \in E^d} A_i$.
\end{definition}

\noindent Having defined the domain of effect and domain of  dependence, we can now define generalized shift maps as follows.

\begin{definition}[Generalized shift map, adapted from~\cite{moore1991generalized}]  Let $F : \Sigma \to \mathbb{Z}$ have a finite domain of dependence which we denote by $D_F^d$.  Moreover, let $G : \Sigma \to \Sigma$ have a finite domain of dependence $D_G^d$ and a finite domain of effect $D_G^e$.  Then we call
\begin{align}
\Phi(a) = \sigma^{F(s)}\!(G(s))
\end{align}
a \textbf{generalized shift map}.
\end{definition}

There is a close connection between Turing machines and generalized shift maps, articulated by~\cite{moore1991generalized} in the following theorem:
\begin{theorem}[Turing machines are conjugate to generalized shift maps, from~\cite{moore1991generalized}]\label{thm:Turingtogeneral1}
For any Turing machine $\textnormal{\textsf{T}}$, there is a generalized shift $\Phi = (F,G)$ conjugate to $\textnormal{\textsf{T}}$.
\end{theorem}
\noindent Let us work out the correspondence explicitly for our universal Turing machine $\textsf{T}_{\!\text{univ}}$ mentioned above.  Recall that $\textsf{T}_{\!\text{univ}}$ can be expressed as a map $\textsf{T}_{\!\text{univ}} :  S \to S$ where $S = \Gamma^* \times Q \times \Gamma^*$.  As such a configuration in $S$ can be written as $\gamma_1\,q\,\gamma_2$ for $\gamma_1, \gamma_2 \in \Gamma^*$ and $q \in Q$. 
 Noting that in our setting $\Gamma = \{00,01,10,11\}$ and $Q = \{0,1\}^k$ where $00$ corresponds to the blank symbol, we can treat $S$ as a subset of $\{0,1\}^\mathbb{Z}$ which contains finitely many $1$'s.  Accordingly, we can re-express an $s \in S$ by 
 \begin{align}
s = \cdots 0\,0\underbrace{s_{-2m+2} \cdots s_{-1}\,.\,s_{0}\,s_{1}}_{\gamma_1} \,\underbrace{s_2 \cdots s_{k+1}}_{q}\,\underbrace{s_{k+2} \cdots s_{k+2n+1}}_{\gamma_2} 0\,0 \cdots
\end{align}
for all $s_i \in \{0,1\}$.  We have put in a decimal point between $s_{-1}$ and $s_{-0}$ for later convenience.  If for our Turing machine $\delta(s_2 \cdots s_{k+1},\,s_{k+2}\,s_{k+3}) = (s_2' \cdots s_{k+1}',\,s_{k+1}'\,s_{k+3}',\,a)$ for $a \in \{\text{\rm L},\text{\rm R}, \text{\rm S}\}$, then $\textsf{T}_{\!\text{univ}}$ acts on $s$ as 
\begin{align}
\label{E:niceTuniv1}
\textsf{T}_{\!\text{univ}}(s)
 = \begin{cases}
\cdots 0\,0\,s_{-2m+2} \cdots s_{-3}\, .\,s_{-2}\,s_{-1} \,s_2' \cdots s_{k+1}'\,s_{0}\,s_{1}\,s_{k+2}'\,s_{k+3}' \cdots s_{k+2n+1}\,0\,0 \cdots &\text{if }a = \text{\rm L} \\
\cdots 0\,0\,s_{-2m+2} \cdots s_{-1}\, .\,s_{0}\,s_{1} \,s_0' \cdots s_{k+1}'\,s_{k+2}'\,s_{k+3}'\cdots s_{k+2n+1}\,0\,0 \cdots &\text{if }a = \text{\rm S} \\
\cdots 0\,0\,s_{-2m+2} \cdots s_{0}\,s_{1}\,.\,s_{k+2}'\,s_{k+3}'\,s_2' \cdots s_{k+1}'\,s_{k+4}\cdots s_{k+n+1}\,0\,0 \cdots &\text{if }a = \text{\rm R}
\end{cases}.
\end{align}
Now let us define the maps $F : \{0,1\}^{\mathbb{Z}} \to \mathbb{Z}$ and $G : \{0,1\}^{\mathbb{Z}} \to \mathbb{Z}$ by
\begin{align}
\label{E:Ffunction1}
F(s) &:= \begin{cases}
    -2 &\text{if }a = \text{\rm L} \\
    0 &\text{if }a = \text{\rm S} \\
    2 &\text{if }a = \text{\rm R}
    \end{cases}, \\ \nonumber \\
\label{E:Gfunction1}
G(s) &:= \begin{cases}
\cdots 0\,0\,s_{-2m+2} \cdots s_{-1}\,.\,s_2'\,s_{3}'\,s_4' \cdots s_{k+1}'\,s_{0}\,s_{1}\,s_{k+2}'\,s_{k+3}' \cdots s_{k+2n+1}\,0\,0 \cdots &\text{if }a = \text{\rm L} \\
\cdots 0\,0\,s_{-2m+2} \cdots s_{-1}\, .\,s_{0}\,s_{1} \,s_2' \cdots s_{k+1}'\,s_{k+2}'\,s_{k+3}'\cdots s_{k+2n+1}\,0\,0 \cdots &\text{if }a = \text{\rm S} \\
\cdots 0\,0\,s_{-2m+2} \cdots s_{-1}\,.\,s_{0}\,s_{1}\,s_{k+2}'\,s_{k+3}' \,s_2' \cdots s_{k+1}'\,s_{k+4}\cdots s_{k+2n+1}\,0\,0 \cdots &\text{if }a = \text{\rm R}
\end{cases}.
\end{align}
We note that the domain of dependence $D_F^d$ of $F$ and the domain of dependence $D_G^d$ of $G$ are bits $s_2$ through $s_{k+3}$.  Similarly, the domain of effect $D_G^e$ of $G$ are bits $s_0$ through $s_{k+3}$.  As such, $D_F^d$, $D_G^d$, and $D_G^e$ are all finite.  Inspecting~\eqref{E:niceTuniv1},~\eqref{E:Ffunction1}, and~\eqref{E:Gfunction1}, we find that $G(s) = \sigma^{- F(s)} (\textsf{T}_{\!\text{univ}}(s))$, giving us the equivalence
\begin{align}
\label{E:Tgenshiftequiv1}
\textsf{T}_{\!\text{univ}}(s) = \sigma^{F(s)}\!(G(s))\,.
\end{align}
Thus, we see that $\textsf{T}_{\!\text{univ}}$ can be written as a generalized shift map in accordance with Theorem~\ref{thm:Turingtogeneral1}.  Note that while at any finite time step a Turing machine will only act on the subset $S$ which we noted has finitely many ones (since the initial conditions to the Turing machine are likewise in $S$), the dynamics of the Turing machine extends to all of $\{0,1\}^{\mathbb{Z}}$ by e.g.~\eqref{E:Tgenshiftequiv1}.

With the above definitions at hand, we can now provide the promised construction of a robustly Turing-universal CDS.
\begin{theorem}[Existence of a smooth, robustly Turing-universal CDS]\label{thm:TuringUniversalCDS1}
There is a robustly Turing-universal CDS $(f, \mathcal{E}, \mathcal{D}, \tau, \textnormal{\textsf{T}}_{\!\text{\rm univ}})$ where $f : D^2 \to D^2$ is a smooth map on the disk, $\tau(x) = 1$, and $\mathcal{E}, \mathcal{D}$ have optimal complexity $\Theta(n)$.
\end{theorem}

\begin{proof}
Let $D_{\text{tot}} := D_F^d \cup D_G^d \cup D_G^e$, which consists of the first bit to the left of the decimal in $s$ and first $k+2$ bits to the right of the decimal in $s$.  We let the dynamics $f : D^2 \to D^2$ take place on the unit disk centered at $(\frac{1}{3},\frac{1}{3})$, containing the square $[-\frac{1}{3},1]^2$.  The configurations of the Turing machine will be encoded and decoded from this square.  Accordingly, let us define the encoder $\mathcal{E}$ and decoder $\mathcal{D}$.

Let $\vec{s} = s_0 s_1 s_2 \cdots$ and $\cev{s} =  \cdots s_{-2} s_{-1}$ be half-infinite sequences which we take as having finitely many $1$'s.  We can put a $\vec{s}$ and $\cev{s}$ together to get a bi-infinite string $s$ with the notation $s = \cev{s}\,.\,\vec{s}$. Then the functions
\begin{align}
\label{E:abss1}
|\vec{s}|_x &:= \begin{cases}
\max\{k \geq 1 \, : \, s_{k-1} = 1\} & \text{if  }\vec{s} \not = 000\cdots \\
0 & \text{if  }\vec{s} = 000\cdots
\end{cases}\,, \\ \nonumber \\
\label{E:abss2}
|\cev{s}|_y &:= \begin{cases}
\max\{k \geq 1 \, : \, s_{-k} = 1\} & \text{if  }\cev{s} \not = \cdots000 \\
0 & \text{if  }\cev{s} = \cdots 000
\end{cases}\,,
\end{align}
respectively locate the position of the right-most $1$ in $\vec{s}$ before there is an infinite string of $0$'s to the right, and the position of the left-most $1$ in $\cev{s}$ before there is an infinite string of zeros to the left.  While the definitions of $|\vec{s}|_x$ and $|\cev{s}|_y$ in~\eqref{E:abss1} and~\eqref{E:abss2} are clear, it is not obvious as written how to compute them. 
 To remedy this, we write down two manifestly computable (but harder to parse at a glance) formulae below, which are tailored to our encoding:
\begin{align}
\label{E:abss3}
|\vec{s}|_x &:= \begin{cases} \min\{j \geq k+2\, : \, s_{j+1}\,s_{j+2}\,s_{j+3} = 0\,0\,0\} &\text{if } s_0 \cdots s_{k+1} \not = 0\cdots 0\,\text{ and }\,s_{k+2}\,s_{k+3} \not = 0\,0 \\
\max\{1 \leq j \leq k+2 \, : \,s_{j-1} = 1\} &\text{if } s_0 \cdots s_{k+1} \not = 0\cdots 0\,\text{ and }\,s_{k+2}\,s_{k+3} = 0\,0 \\
0 &\text{if } s_0 \cdots s_{k+1} = 0\cdots 0
\end{cases}\,,
\\ 
\label{E:abss4}
|\cev{s}|_y &:= \min\{j \geq 0\,:\,s_{-j-3}\,s_{-j-2}\,s_{-j-1} = 0\,0\,0\}\,.
\end{align}
These definitions rely on the fact that the initialized part of the tape of our Turing machine must be a contiguous string of non-blank symbols, and that the Turing machine can only write non-blank symbols.  As such, there will never be a non-blank symbol followed by blank symbols followed by a non-blank symbol.

We further define ternary encodings of $\vec{s}$ and $\cev{s}$ by the functions
\begin{align}
\label{E:XY1}
X(\vec{s}) := \sum_{j=0}^{|\vec{s}|_x} 2s_{j}\,3^{-(j+1)}\,,\qquad Y(\cev{s}) := \sum_{j = 1}^{|\cev{s}|_y + 1} 2 s_{-j}\,3^{-j}\,.
\end{align}
Using~\eqref{E:abss1} and~\eqref{E:XY1}, we can define the sets $C_s = \mathcal{D}^{-1}(s)$ by
\begin{align}
\label{eq:cantor-decoder}
C_s := [X(\vec{s})- \frac{1}{3^{|\vec{s}|_x + 1}},\, X(\vec{s}) ] \times [Y(\cev{s})- \frac{1}{3^{|\cev{s}|_y + 1}},\, Y(\cev{s}) ]\,,
\end{align}
which evidently are closed sets with non-trivial interior (and as such, our CDS will be robust in the sense we previously defined).  Moreover all $C_s$'s are contained in the square $[-\frac{1}{3},1]^2$, and are depicted in Figure~\ref{fig:Cantor1}.  We can think of $C_s$ as a thickened version of a two-dimensional Cantor set encoding of $s$.  The encoder $\mathcal{E}$ is then given by 
\begin{align}
\mathcal{E}(s) := (X(\vec{s})- \frac{1}{2}\,\frac{1}{3^{|\vec{s}|_x + 1}},\,Y(\cev{s})- \frac{1}{2}\,\frac{1}{3^{|\cev{s}|_y + 1}})\,,
\end{align}
and satisfies $\mathcal{D} \circ \mathcal{E}(s) = s$.  Moreover, both $\mathcal{E}$ and $\mathcal{D}$ can be implemented by hybrid BSS$_{\textsf{C}}$ machines with complexity $\Theta(n)$.

\begin{figure}[t!]
    \centering
    \includegraphics[scale = .5]{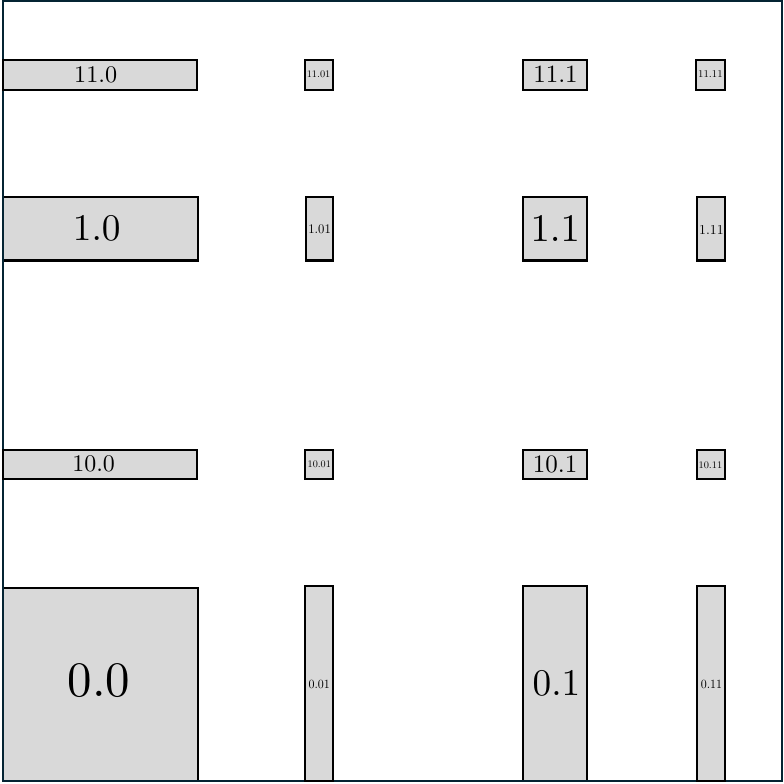}
    \vspace{15pt}
    \caption{An illustration of the some of the sets $C_s$ in the square $[-\frac{1}{3},1]^2$.  The numerical labels of the sets indicate the corresponding $s$, with the understanding that there are infinitely many zeros to the left and infinitely many zeros to the right (i.e.~$1.01$ is shorthand for $\cev{0}1.01\vec{0}$).  The sets are organized according to a thickened version of a Cantor encoding in two dimensions.}
    \label{fig:Cantor1}
\end{figure}

Having constructed $\mathcal{E}$ and $\mathcal{D}$, we now turn to constructing appropriate dynamics $f : D^2 \to D^2$.  Before we proceed, let us define the smooth, monotonic functions
\begin{align}
q_{a,b}(x) := \begin{cases}
a\,x &\text{for }x \leq - \frac{b}{2} \\
\phi_{a,b}(x) &\text{for }- \frac{b}{2} \leq x \leq 0 \\
x &\text{for } x \geq 0
\end{cases}\,,\qquad r_{a,b}(x) := \begin{cases}
x &\text{for }x \leq - \frac{b}{2} \\
\psi_{a,b}(x) &\text{for }- \frac{b}{2} \leq x \leq 0 \\
a\,x &\text{for } x \geq 0
\end{cases}\,,
\end{align}
where $\phi_{a,b}(x)$ and $\psi_{a,b}(x)$ are smooth, monotonic interpolating functions (which can be constructed explicitly or otherwise are guaranteed to exist by the Whitney extension theorem).  The functions $q_{a,b}(x)$ and $r_{a,b}(x)$ will play important roles in the definition of $f$.

Our strategy in the remainder of the proof is to construct smooth maps $\textsf{S}$ and $\textsf{G}$ from a nice subset of $D^2$ (containing $\bigcup_{s \in S} C_s$) to $D^2$ such that $\textsf{S}(C_s) = C_{\sigma(s)}$ and $\textsf{G}(C_s) = C_{G(s)}$.  Then we can compose these maps in an appropriate way and extend the domain to all of $D^2$ to arrive at a mapping $f : D^2 \to D^2$ such that $f(C_s) = C_{\sigma^{F(s)}(G(s))} = C_{\textsf{T}_{\!\text{univ}}(s)}$.  We begin by constructing $\textsf{G}$.

Recall from~\eqref{E:Gfunction1} that $G$ only depends on $s_0\,s_1 \cdots s_{k+3}$.  Let us denote the $2^{k+4}$ possible bit strings $s_0\, s_1 \cdots s_{k+3}$ by $z_{1},...,z_{2^{k+4}}$.  It will be convenient to find nice sets $B_{z_m}$ such that $C_{s} \subset B_{z_m}$ for all $s$, with $s_0 \,s_1 \cdots s_{k+3} = z_m$.  If $\vec{0} := 000 \cdots$, an example of nice sets $B_{z_m}$ is given by
\begin{align}
B_{z_m} = [X(z_m \vec{0}) - \frac{1}{3^{|z_m \vec{0}|_x + 1}},\, X(z_m \vec{0}) + \frac{1}{3^{k+4}}] \times [-\frac{1}{3},1]\,.
\end{align}
These $B_{z_m}$'s are pairwise disjoint, and have non-zero distance apart from one another. 
 In a slight abuse of notation, we can treat $G$ as a map $G : \{z_1,...,z_{2^{k+4}}\} \to \{z_1,...,z_{2^k}\}$.  Then we construct the map
\begin{align}
\label{E:sfGeq1}
\textsf{G}(x,y) = \left(q_{a(z_m), \,b(z_m)}(x - X(z_m \vec{0})) + X(G(z_m)\vec{0}),\,\,y\right) \qquad \text{if} \quad (x,y) \in B_{z_m}\,,
\end{align}
where $a(z_m) := \frac{3^{|z_m \vec{0}|_x}}{3^{|G(z_m) \vec{0}|_x}}$ and $b(z_m) := \frac{1}{3^{|z_m \vec{0}|_x + 1}}$.  We note several features of the above map $\textsf{G} : \bigcup_{m = 1}^{2^{k+4}} B_{z_m} \to D^2$.  The map smoothly rearranges and resizes a finite number of vertical strips in $[-\frac{1}{3},1]^2$ so that $\textsf{G}(C_s) = C_{G(s)}$ for all $s \in S$.  Due to the re-sizing, $\textsf{G}$ is not area-preserving.  

We can similarly define a smooth map $\textsf{S}$ satisfying $\textsf{S}(C_s) = C_{\sigma(s)}$, namely
\begin{align}
\label{E:Shift1}
\textsf{S}(x,y) = \begin{cases}
(r_{\frac{1}{3},\frac{1}{3}}\!(x),\,r_{3,\frac{1}{3}}\!(x)) & \text{for }(x,y) \in [-\frac{1}{3},1] \times [-\frac{1}{3},\frac{1}{3}] \\
(\frac{1}{3}x + \frac{2}{3},\,3(y - \frac{2}{3})) & \text{for  } (x,y) \in [-\frac{1}{3},1] \times [\frac{2}{3}-\frac{1}{9}, 1]
\end{cases}\,,
\end{align}
which is a nonlinear version of the Baker's map. 
The map $\textsf{S}$ is not area-preserving, and is only defined on the two blocks $[-\frac{1}{3}, 1] \times [-\frac{1}{3},\frac{1}{3}]$ and $[-\frac{1}{3},1]\times [\frac{5}{9},1]$ which jointly contain all of the $C_s$'s.  The two blocks comprise both the domain and co-domain of $\textsf{S}$, and so $\textsf{S}$ can be composed with itself.  A depiction of the action of $\textsf{S}$ on $[-\frac{1}{3},1]^2$ is given in Figure~\ref{fig:Shift1}.  It is straightforward to use $\textsf{S}$ to define a smooth map $\textsf{H}$ satisfying $\textsf{H}(C_s) = C_{\sigma^{F(a)}(s)}$.  To this end we note that, in a slight abuse of notation, $F$ in~\eqref{E:Ffunction1} can be viewed as a map $F :\{z_1,...,z_{2^k}\} \to \mathbb{Z}$, and we define
\begin{align}
\label{E:Heq1}
\textsf{H}(x,y) = \textsf{S}^{F(z_m)}\big|_{B_{z_m}}\!(x,y) \quad \text{if} \quad (x,y) \in B_{z_m}\,.
\end{align}
Above, $\textsf{S}^{F(z_m)}|_{B_{z_m}}$ denotes the map $\textsf{S}^{F(z_m)}$ restricted to the domain $B_{z_m}$.  Indeed, the above mapping $\textsf{H} : \bigcup_{m = 1}^{2^{k+4}} B_{z_m} \to D^2$ is smooth and satisfies $\textsf{H}(C_s) = C_{\sigma^{F(a)}(s)}$ for all $s \in S$.

\begin{figure}[t!]
    \centering
    \includegraphics[scale = .4]{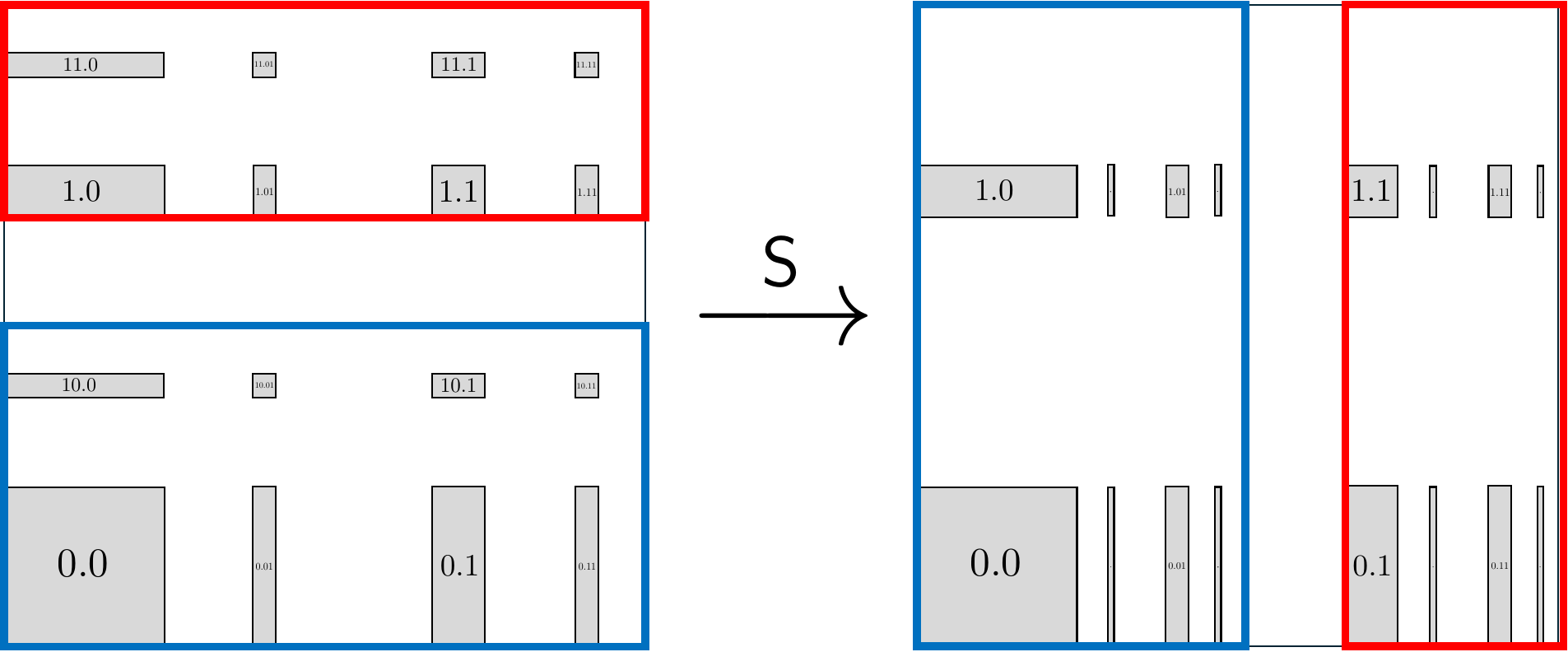}
        \vspace{15pt}
    \caption{A depiction of how the map $\textsf{S}$ in~\eqref{E:Shift1} acts on two subsets of the square $[-\frac{1}{3},1]^2$.  Some of the sets $C_s$ are depicted (akin to Figure~\ref{fig:Cantor1}) to clarify the action of the mapping $\textsf{S}$.  The red rectangle is shifted to the right and nonlinearly compressed in the $x$-direction, as well as nonlinearly stretched in the $y$-direction.  The blue rectangle is also nonlinearly compressed in the $x$-direction as well as nonlinearly stretched in the $y$-direction.  Hence $\textsf{S}$ is a nonlinear version of the Baker's map, and indeed acts as a shift map since it takes $C_s$ to $C_{\sigma(s)}$.}
    \label{fig:Shift1}
\end{figure}

Putting~\eqref{E:sfGeq1} and~\eqref{E:Heq1} together, we find
\begin{align}
\label{E:HGeq1}
\textsf{H} \circ \textsf{G}(C_s) = C_{\sigma^{F(s)}(G(s))}\,,
\end{align}
where $\textsf{H}$ and $\textsf{G}$ each are maps from $\bigcup_{m = 1}^{2^{k+4}} B_{z_m} \to D^2$. 
 Next we will extend the domain to all of $D^2$ (which we recall we have taken to be the unit disk centered at $(\frac{1}{3}, \frac{1}{3})$, containing the square $[-\frac{1}{3},1]^2$).  We first note that the $B_{z_m}$'s each have a non-zero pairwise distance from one another, and from the boundary of $D^2$.  Then we can enlarge each $B_{z_m}$ slightly into a closed set $\widetilde{B}_{z_m}$ with smooth boundary (e.g.~no corners), such that the $\widetilde{B}_{z_m}$'s each have a non-zero pairwise distance from one another, and non-zero distance to the boundary of $D^2$.  Then we extend the domains of the functions comprising $\textsf{H}$ and $\textsf{G}$ to $\bigcup_{m = 1}^{2^{k_4}}\widetilde{B}_{z_m}$, resulting in functions $\widetilde{\textsf{H}}$ and $\widetilde{\textsf{G}}$ which also satisfy~\eqref{E:HGeq1}.  But since $\widetilde{\textsf{H}} \circ \widetilde{\textsf{G}} : \bigsqcup_{m = 1}^{2^{k+4}} \widetilde{B}_{z_m} \to \bigsqcup_{m = 1}^{2^{k+4}} \widetilde{B}_{z_m}  $
is a smooth map and $\bigsqcup_{m = 1}^{2^{k+4}} \widetilde{B}_{z_m}$ is contained in $D^2$ and has finite distance from the boundary of $D^2$, by the Whitney extension theorem we can extend $\widetilde{\textsf{H}} \circ \widetilde{\textsf{G}}$ to a smooth function $f : D^2 \to D^2$.  This function satisfies
\begin{align}
f(C_s) = C_{\sigma^{F(s)}(G(s))} = C_{\textsf{T}_{\!\text{univ}}}(s)\,,
\end{align}
and therefore we have
\begin{align}
\mathcal{D} \circ f(C_s) = \textsf{T}_{\!\text{univ}}(s)
\end{align}
for all $s \in S$.  Evidently from the above equation, $\tau(x) = 1$.  This completes the construction of our robustly Turing-universal CDS.
\end{proof}

Although the previous proof is somewhat elaborate, the basic ingredients are straightforward.  The first idea is that since we need to encode finite strings into regions with non-trivial interior so that the CDS is robust, the $C_s$'s need to become smaller as the `size' of $s$ becomes larger.  These considerations motivate a natural guess for the encoded regions (and hence the encoder and decoder), where the $C_s$'s are thickened version of the Cantor encoding of an infinite binary string into two dimensions.  The second idea is that by appropriately stretching and compressing regions of the disk (in a non-area-preserving manner), we can implement a proxy for $G$ which acts on a finite number of regions.  The third idea is that a nonlinear generalization of the Baker's map serves as a proxy for the shift map on the $C_s$'s.  By combining the previous ingredients appropriately, we obtain a function on a subset of the disk implementing Turing-universal dynamics on the $C_s$'s.  We then smoothly extend said function into $f$, which gives us our desired CDS.  

\begin{remark}[Generalizing Theorem~\ref{thm:Turingtogeneral1} to a diffeomorphism]
We can generalize our CDS in Theorem~\ref{thm:Turingtogeneral1} to a diffeomorphism if we use a Turing machine $\textsf{T}_{\!\text{univ}}$ which is reversible.  The same exact construction presented in the proof goes through without any other change.
\end{remark}

\begin{remark}
Our robustly Turing-universal CDS is defined on a compact set in $\R^2$.  The work~\cite{cardona2023computability} gives a fascinating example of a gradient flow on $\R^2$ which does not preserve any compact subset and yet is Turing-universal, although under a less robust definition of simulation than ours.  However, we suspect that their example, or a slight variation thereof, may also furnish a robustly Turing-universal CDS.  The same paper also gives a related example of a gradient flow on the sphere with zero topological entropy which is Turing-universal under a less robust definition of simulation.  It would be interesting to understand if this example is a robustly Turing-universal CDS; this would involve a non-trivial slowdown function $\tau$ implicit in the construction in~\cite{cardona2023computability} (arising from the function $G$ of~\cite[Section 7]{cardona2023computability}).
\end{remark}

\subsection{Non-universality of Axiom A systems}
\label{sec:axiom-a}

We now state another central result of our paper: 

\begin{theorem}
\label{thm:axiom-a-non-universal}
 Let $f: M \to M$ be an Axiom A system and assume that $M$ is compact. There is no extension of $f$ to a robustly Turing-universal CDS $(f,\mathcal{E}, \mathcal{D}, \tau, \machine{T})$. 
\end{theorem}

\noindent Below we recall the definition of an Axiom A system (originally due to Smale \cite{smale1967differentiable}), as well as its basic properties. We must first introduce some auxiliary definitions.

\begin{definition}[Wandering and nonwandering sets]
The \textbf{wandering set} of $f$ is the set of points $x \in M$ such that there is a neighborhood $U$ of $x$ and an $N$ such that for all $n > N$, $f^n(U) \cap U = \emptyset$. The \textbf{nonwandering set} $\Omega_f$ of $f$ is the complement of the wandering set. 
\end{definition}

\noindent Clearly, any periodic point of $f$ is in the non-wandering set of $f$.

\begin{definition}[Axiom A, Smale \cite{smale1967differentiable}]
\label{def:axiom-A-diffeo}
We say that $f$ is \textbf{Axiom A} if its nonwandering set $\Omega_f$ is compact, hyperbolic, and if the periodic points of $f$ are dense in $\Omega_f$. 
\end{definition}

Much of the reason for the importance of Axiom A systems is that they form a rich class of differentiable dynamical systems with dynamics that can be characterized in great detail. Our proof of Theorem \ref{thm:axiom-a-non-universal} will rely on this precise understanding. As such, we will recall the basic structural results on Axiom A below. 

The first result is the \emph{spectral decomposition} of Axiom A systems:

\begin{proposition}[Spectral Decomposition of Axiom A systems \cite{smale1967differentiable, bowen2008equilibrium}]
\label{prop:spectral-decomposition}
Let $f: M \to M$ be an Axiom A system. Then there is a decomposition of the non-wandering set
\[ \Omega_f = \Omega_1 \cup \cdots \cup \Omega_k \]
into disjoint closed hyperbolic $f$-invariant subsets such that $f$ is topologically transitive on each $\Omega_i$. Each $\Omega_i$ can be written as a union of pairwise disjoint closed sets 
\[ \Omega_{i} = \bigsqcup_{j=1}^{r_i} \Omega_{i, j}\]
with $f(\Omega_{i, j}) = \Omega_{i, j+1}$ for $j=1, \ldots, r_i-1$, and $f(\Omega_{i, r_i}) = \Omega_{i, 1}$, and such that 
\[ f^{r_i}|_{\Omega_{i, j}} \text{ is topologically mixing. }\]
  Moreover, 
there is a decomposition into disjoint subsets 
\[ M = \bigsqcup_{i=1}^k W^s(\Omega_i)\,. \]
\end{proposition}

Each of the $\Omega_i$ should be thought of as a kind of `attractor' for $f$. However, since the expanding part of the tangent bundle $T^+_{\Omega_i}$ (see Definition \ref{def:hyperbolic-set}) may be nontrivial, points near $\Omega_i$ may not be attracted to $\Omega_i$ under the dynamics of $f$; thus these `attractors' may be `unstable'. An illustrative example of an Axiom A system may be obtained by letting $f$ be a time-1 gradient flow of a proper Morse function $g: M \to \R$; in this setting, the $\Omega_i$ are exactly the critical points $x$ of $g$, and the decomposition of $TM|_{\Omega_i}$ is exactly the decomposition of $T_xM$ into the positive and negative eigenspaces of the Hessian of $g$ at the critical point $x$. Note that in contrast to this example, in general the geometry of the $\Omega_i$ may be quite complicated, and possibly fractal: the Cantor set associated to a Smale horsehoe (see Figure~\ref{fig:manysystems}(d) for a depiction) is an Axiom A attractor \cite{katok1995introduction}. Moreover, the dynamics on each $\Omega_i$ can be fairly chaotic, as we will discuss below. 

Another basic result regarding Axiom A systems is that they are essentially the same as the  \emph{structurally stable} systems:
\begin{theorem}[\cite{robbin1971structural, robinson1976structural, mane1987proof}]\label{thm:axiom-a-is-structurally-stable}
Suppose that $f$ is \textbf{structurally stable}: there exists an $\epsilon > 0$ such that for all diffeomorphisms $g: M \to M$ which are $\epsilon$-$C^1$-close to $f$ (i.e.~those such that $\sup_{x \in M} d(f(x), g(x)) + \|df_x-dg_x\| < \epsilon$), $g$ is \emph{topologically conjugate} to $f$ (i.e.~there exists a homeomorphism $h_g: M \to M$ such that $f = h_g \circ g \circ h^{-1}_g$). Then $f$ is Axiom A.

In fact, structural stability is equivalent to the condition that $f$ is Axiom A and that $f$ satisfies the \emph{Strong Transversality} condition: for every $x, y \in \Omega$, the stable manifold $W^s(x)$ is transverse to the unstable manifold $W^u(y)$. \end{theorem}

In part the motivation for the study of Axiom A systems was the hope that \emph{generic} systems might be Axiom A, and thus one might be able to give a tractable description of generic dynamics. Unfortunately, this latter statement turns out to be false \cite{newhouse1970nondensity}. There exists a rich set of conjectures and expectations regarding the dynamics of a generic differentiable dynamical system due to Palis \cite{palis2000global}, and an enormous collection of results and techniques for their study. The analysis of Axiom A systems is the simplest application of these methods. 

In particular, the spectral decomposition for Axiom A systems is proven using another tool, the \emph{stable manifold theorem for hyperbolic sets}:
\begin{proposition}[\cite{smale1967differentiable}, proven in \cite{hirsch-pugh}; see also Theorem 3.2 of~\cite{bowen2008equilibrium}]
\label{prop:stable-manifold-theorem}
Let $C$ be a compact hyperbolic set for $f$. Then $W^s(C)$ is a union of sets 
\[ W^s(x) = \{ y \in M : d(f^n(y), f^n(x)) \to 0 \text{ as } n \to \infty\} \]
over $x \in C$. Each of these sets is the image of a smooth injective immersion of a smooth manifold of dimension $\dim T^-_C$, and if $y \in C$ then the tangent space to $W^s(x)$ at $y$ is $(T^-_C)_y$. Moreover, the $W^s(x)$ are a \emph{continuous family of smooth submanifolds}, i.e.~for every $x \in C$ and any $z \in W^s(x)$, there exists a neighborhood $U_x \subset C$ and a continuous map $\phi: U_x \to C^\infty(D^r, M)$ (where $r = \dim W^s(x)$, and $D^r$ is the unit disk of degree $r$) such that $\phi(y)$ is an equidimensional immersion from $D^r$ into $W^s(y)$ sending $0$ to $y$, and such that $\phi(x)(D^r)$ contains $z$. 
\end{proposition}

\begin{remark}
In fact, the theorem proven in \cite{hirsch-pugh} is not quite the result above; we explain the (standard) derivation of the variant above in Appendix \ref{App:stable-manifold-theorem-variants}. 
\end{remark}

To prove Theorem \ref{thm:axiom-a-non-universal}, we take the heuristic perspective that each `attractor' $\Omega_i$ acts as a `memory bank' which can store only a bounded amount of `data'. But the existence of arbitrarily many disjoint sets $C_s$ with $f^L(C_s) \subset C_s$ suggests that a universal CDS should be able to implement an arbitrarily large `memory'. Let us imagine that $\tau(x)=1$. By the spectral decomposition, given a point $x_s \in C_s$, we have that $f^{kL}(x) \to \Omega_i$ for some $i$; since $C_s$ is closed, this implies that $C_s$ contains the closure of an orbit of a $z_s \in \Omega_i$. We then use the precise form of the stable manifold theorem to show that we can perturb $x_s$ a little bit to $\tilde{x}_s \in C_s$ such that the corresponding $\tilde{z}_s \in \Omega_i \cap C_s$ has an orbit dense in $\Omega_i$. Thus each $C_s$ eats up one unit of `memory'; since the $C_s$ are disjoint but their number is arbitrarily large, this is a contradiction.  A visual summary of the proof is depicted in Figure~\ref{fig:axiom-a-proof}.

\begin{figure}[t!]
    \centering
    \includegraphics[scale = .85]{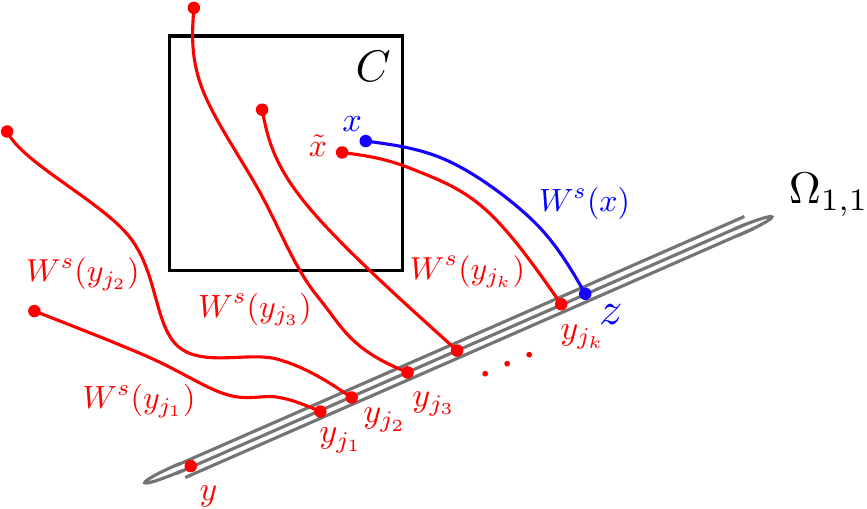}
        \vspace{15pt}
    \caption{Visual summary of the proof of Theorem \ref{thm:axiom-a-non-universal}.}
    \label{fig:axiom-a-proof}
\end{figure}

\begin{proof}[Proof of Theorem \ref{thm:axiom-a-non-universal}]

Let $\tilde{k} = \sum_{i=1}^k r_i$ with $k$ and $r_i$ as in the spectral decomposition theorem. 
Choose $N > \tilde{k}$, and using Lemma \ref{lemma:UTM-has-periodic-points}, choose $L$ sufficiently large such that there exist $N$ periodic configurations $s_1, \ldots, s_N$ of our universal Turing machine of period $L$, such that the orbits of these configurations are pairwise distinct. The robustly Turing-universal CDS condition then produces $N$ subsets $C_1, \ldots, C_N$ of $M$ which are closed and have nonempty interior, and such that $f^{n_{C_i}}(C_i) \subset C_i$ for each $i=1, \ldots, N$ and some integers $n_{C_i}$. Indeed, these sets can be taken to be a choice of connected component of $\mathcal{D}^{-1}(s_j)$ for each $s_j$, for the robust CDS condition implies that $(f^\tau)^L(C_i) \subset C_i$ for each $i$.  Replace $f$ by $f^{(\prod_{i=1}^N C_i)r_1 \cdots r_k}$; this diffeomorphism is still Axiom A, and its spectral decomposition is simply the decomposition 
\[ \Omega = \Omega_{1, 1} \sqcup \cdots \sqcup \Omega_{k, r_k}\,.\]
In particular, $N$ is still greater than the number of basic sets in the spectral decomposition.  Moreover, now $f(C_i) \subset C_i$. We will show that for each $j=1, \ldots, N$, $C_j$ contains one of the basic sets by the pidgenhole principle, and this is a contradiction. Without loss of generality let us set $j=1$.

Choose an $x$ in the interior of $C_1$. Then by the spectral decomposition of $f$, $x \in W^s(\Omega_{i, j})$ for some $i, j$; without loss of generality let us set $i=1, j=1$. By Proposition \ref{prop:stable-manifold-theorem}, there is a $z \in \Omega_{1,1}$ such that $d(f^{n}(x), f^{n}(z)) \to 0$ as $n \to \infty$. Now there is a point $y \in \Omega_{1,1}$ such that the orbit of $y$ is dense in $\Omega_{1,1}$. In particular, there is a sequence $n_i$ of distinct natural numbers diverging to infinity such that $y_i = f^{n_i}(y) \to z$ as $i \to \infty$. For all sufficiently large $i$, $y_i \in U_z$ with the notation as in Proposition \ref{prop:stable-manifold-theorem}. Let $p \in D^n$ be the preimage of $z$ under $\phi(x)$. By continuity of $\phi$, for sufficiently large $i$, we have that $\phi(y_i)(p) \to \phi(x)(p) = z$; in particular for sufficiently large $i$, $\phi(y_i)(p) \in C$. Fix any such $i$ and write $\tilde{x} = \phi(y_i)(p)$, $\tilde{z} = y_i$. 

The orbit of $\tilde{z}$ under $f$ is dense in $\Omega_{1,1}$. Thus, for any $w \in \Omega_{1,1}$, there exists a sequence $m_i$ of natural numbers diverging to infinity such that $f^{m_i}(\tilde{z}) \to w$ as $i \to \infty$. But $d(f^{m_i}(\tilde{x}), f^{m_i}(\tilde{z})) \to 0$ as $i \to \infty$. So $w \in C_1$ is in the closure of $C_1$. We have thus proven that $\Omega_{1,1}\subset C_1$. Then there is a relabeling of the $\Omega_{i,j}$ as 
$\Omega_k = \Omega_{i(k), j(k)}$ such that $\Omega_k \subset C_j$. But the number of $C_j$'s is larger than the number of $\Omega_{i,j}$'s, so this is not possible. We have proven the theorem.
\end{proof}

\begin{proof}[Proof of Theorem \ref{thm:axiom-a-not-universal}]
The only condition used in the proof of Theorem \ref{thm:axiom-a-non-universal} is the existence of, for an arbitrarily large $N$, a collection of disjoint sets $C_i$ that are closed, have non-empty interior, and satisfy $f^{n_{C_i}}(C_i) \subset C_i$ for each $i=1, \ldots, N$. But this is exactly the condition that $f$ simulates the corresponding machine $\machine{T}_N$ as defined in the statement of Theorem \ref{thm:axiom-a-not-universal}.
\end{proof}

\subsubsection{Conjectures about generic diffeomorphisms}

We have proven that structurally stable systems cannot robustly implement universal computation (in the sense of this paper). As mentioned previously, the original motivation for the study of Axiom A systems was, in part, the hope that they would be generic, that they would form an open and dense subset of the set of diffeomorphisms. This turns out not to be the case \cite{newhouse1970nondensity}; however, there are still hopes that generic diffeomorphisms have nice structure theorems \cite{palis2000global}. 

\begin{conjecture}
For any compact computable manifold $M$ (possibly with boundary), there is a $C^\infty$ generic set $\mathcal{S}$ of diffeomorphisms $f: M \to M$ such that no diffeomorphism $f \in \mathcal{S}$ can be extended to a robustly Turing-universal CDS $(f, \mathcal{E}, \mathcal{D}, \tau, \textnormal{\textsf{T}}_{\!\text{\rm univ}})$. 
\end{conjecture}

\noindent Here a $C^\infty$-generic set is a set that is the intersection of a countable collection of open dense sets in the space of smooth diffeomorphisms of $M$. We now prove this conjecture under a strong condition on the decoder as well as a condition on the slowdown function, using a periodic point argument:

\begin{theorem}
\label{thm:generic-diffeo-not-universal-under-hierachical-shrinking}
    For any compact computable manifold $M$ (possibly with boundary), there is an open dense set $\mathcal{S}$ of diffeomorphisms $f: M \to M$ such that no diffeomorphism $f \in \mathcal{S}$ can be extended to a robustly Turing-universal CDS $(f, \mathcal{E}, \mathcal{D}, \tau, \textnormal{\textsf{T}}_{\!\text{\rm univ}})$ where $\mathcal{D}$ satisfies the shrinking condition (see Definition~\ref{def:shrinking-decoders-and-so-forth}) and the slowdown function is \emph{constant}.
\end{theorem}

\begin{remark}
Note that if one drops robustness and genericity \cite{koiran1999closed} or one allows oneself to work on noncompact domains \cite{gracca2023analytic}, then the above statement is false. The results of~\cite{koiran1999closed, gracca2023analytic} are under somewhat different formalizations, and we do not explicate the comparison here. 
\end{remark}

\begin{proof}Let us assume that the slowdown function has constant value $a$. 
    We can choose the set of diffeomorphisms in the theorem statement to be the set where all periodic points are hyperbolic; this is open and dense by the Kupka-Smale theorem \cite[Theorem~7.2.6]{katok1995introduction}. Note that $\machine{T}_{\!\text{univ}}$ has some periodic configuration $s_0$ of period $n$ by Lemma \ref{lemma:UTM-has-periodic-points}; moreover, we can find such a periodic configuration such that if we append symbols to the input those symbols are never read in the periodic motion of $\machine{T}_{\!\text{univ}}$. Thus for every finite string $t$ on the alphabet of tape symbols $\Sigma$ we have a corresponding periodic configuration $s_t$ of $\machine{T}_{\!\text{univ}}$. Moreover, the decoder defines for us a closed (and thus compact) set $C_{s_t} \subset M$ such that $f^{an}(C_{s_t}) \subset C_{s_t}$ where $n$ is the period of $s_t$ under the dynamics of $\machine{T}_{\!\text{univ}}$. Since the lengths of the corresponding configurations of $\machine{T}$ go to infinity as the length of $s$ goes to infinity, the shrinking condition implies that given an infinite string $\hat{s}$ on $\Sigma$ and writing its finite truncations as $\hat{s}_n$, we have that the diameters of $C_{t_{\hat{s}_n}}$ converge to zero. Since each $C_{t_{\hat{s}_n}}$ is closed, it is compact, and so it defines a subspace of the Hausdorff metric space of all compact subspaces of $M$. Since $M$ is compact, this latter space is compact as well, and so there is a subsequence $C_{t_{\hat{s}_{n_j}}}$ converging to some point $p_{\hat{s}} \in M$. Since $f$ is continuous, we must have that $f^{an}(p_{\hat{s}}) \subset \{p_s\}$, i.e.~$f^{an}(p_{\hat{s}}) = p_{\hat{s}}$. Thus, if we show that the points $p_{\hat{s}}$ are distinct for different infinite strings $\hat{s}$ we will be done. Indeed, the hyperbolicity of all periodic points and the compactness of $M$ implies that there are only a countable number of periodic points. (This is because every periodic point of period $k$ has an open neighborhood containing no other periodic points of period $k$ by the local normal form for hyperbolic periodic points; so by compactness of $M$ there are finitely many periodic points of period $k$ for each $k$, and thus countably many periodic points in general.) But the distinctness of the points $p_{\hat{s}}$ follows from the hierarchical shrinking condition: we have that if $\hat{s}^a \neq \hat{s}^b$, then $\hat{s}^a_n \neq \hat{s}^b_n$ for some $n$; and thus we have that for all sufficiently large $j$, $C_{t_{\hat{s}^a_{n_j}}} \subset C'_{t_{\hat{s}^a_n}}$  and similarly $C_{t_{\hat{s}^b_{n_j}}} \subset C'_{t_{\hat{s}^b_n}}$, so $p_{\hat{s}^a} \in C'_{t_{\hat{s}^a_n}}$ while $p_{\hat{s}^b} \in C'_{t_{\hat{s}^b_n}}$. Since the hierarchical shrinking condition implies that $C'_{t_{\hat{s}^a_n}} \cap C'_{t_{\hat{s}^a_n}} = \emptyset$, we know that the points $p_{\hat{s}^a}$ are disjoint. In other words, we have produced an uncountable number of periodic points of $f$, which is a contradiction.
\end{proof}

We do not know how to implement an analog of this argument without any conditions on the decoder, or even for a shrinking decoder. We hope the different flavor of arguments in this section vis-\'{a}-vis the previous section highlight how changing the condition on the decoder brings out different ways in which information can be encoded into the dynamics of a smooth dynamical system.

\begin{remark}
    Remark \ref{rk:robust-universal-system-can-be-shrinking} thus implies that the diffeomorphism underlying the construction of Section \ref{sec:example-of-robust-cds} is highly non-generic. 
\end{remark}

\begin{remark}
    The argument in the proof fails completely when we allow for a non-constant slowdown function. Attempting the same argument in that setting, the periodic points of constant period $n$ of $\machine{T}_{\!\text{univ}}$ then correspond to periodic points of \emph{unbounded} period of $f$. But Kupka-Smale diffeomorphisms typically have infinitely many periodic points; it is only the number of periodic points of any \emph{bounded} period which is finite.
\end{remark}
\subsection{Non-universality of measure-preserving and integrable systems}
\label{subsec:nonunivmeas1}

Above we showed that Axiom A systems cannot be extended to robustly Turing-universal CDSs.  One interpretation of that result is that certain kinds of chaotic behavior are incommensurate with robust Turing universality.  Below we will pursue a complementary set of results, namely that measure-preserving systems on compact domains, as well as certain measure-preserving systems on non-compact domains, are likewise incapable of furnishing robustly Turing-universal CDSs.  Such systems are at the opposite end of the extreme of chaotic systems.  As such, we are in total establishing that neither `very chaotic' systems nor `very non-chaotic' systems are capable of being robustly Turing universality.

For all of the results in this section we need the following lemma.
\begin{lemma}
\label{lemm:intersection}
Let $(f, \mathcal{E}, \mathcal{D}, \tau, \textnormal{\textsf{T}})$ be a CDS.  If for any $n \in \mathbb{Z}_{\geq 0}$ and any $s, s' \in S$ we have $(f^{\tau})^n(C_s) \cap C_{s'} \not = \emptyset$, then $(f^\tau)^n(C_s) \subseteq C_{s'} = C_{\textnormal{\textsf{T}}^n(s)}$. 
\end{lemma}

\begin{proof} 
By the definition of a CDS we have $\mathcal{D} \circ (f^\tau)^n(C_s) = \textsf{T}^n(s)$, and so applying $\mathcal{D}^{-1}$ to both sides we find $(f^\tau)^n(C_s) \subseteq C_{\textsf{T}^n(s)}$.  So to complete the proof, it suffices to show that $C_{\textsf{T}^n(s)} = C_{s'}$.  Since $(f^\tau)^n(C_s) \cap C_{s'} \not = \emptyset$, we have $C_{\textsf{T}^n(s)} \cap C_{s'} \not = \emptyset$.  If $x$ is a point in the intersection, then $\mathcal{D}(x) = \textsf{T}^n(s) = s'$, but by definition the inverse image under $\mathcal{D}$ is $C_{\textsf{T}^n(s)} = C_{s'}$.
\end{proof}

\noindent We can use the above Lemma to prove our result about measure-preserving maps:

\begin{theorem}
\label{thm:measurepreserve}
Let $\mu$ be a Borel measure on a compact set $M \subset \R^n$ which assigns nonzero measure to all nonempty open sets and such that $\text{\rm supp}(\mu) = M$ and $\mu(M) < \infty$.  If $f : M \to M$ is a measure-preserving map with respect to $\mu$, then $f$ cannot be extended to a robust CDS for the machine $\machine{Plus}: \{1\}^* \to \{1\}^*$ defined by $\machine{Plus}([n]_1) = [n+1]_1$, where $[n]_1$ denotes $n$ in unary.
\end{theorem}
% This property is preserved under pushforwards of homeomorphisms.

\begin{proof}
By contradiction, let $s_0 = 1$ be the initial configuration of $\textsf{Plus}$, which which visits infinitely many distinct configurations $\{s_{n}\}_{n=0}^\infty$ with $s_n := \textsf{Plus}^n(s_0) = [n+1]_1$.  Define $\widetilde{C}_{s_{n}} := (f^{\tau})^n(C_{s_0})$, and notice that on account of Lemma~\ref{lemm:intersection} the $\widetilde{C}_{s_{n}}$'s are pairwise disjoint.  Since $C_{s_0}$ has non-trivial interior and $\mu$ is supported on all of $M$, we have $\mu(C_{s_0}) = \mu(\widetilde{C}_{s_{n}}) > 0$.  Consequently $\sum_{n=0}^\infty \mu(\widetilde{C}_{s_{n}}) = \infty$, but $\sum_{n=0}^\infty \mu(\widetilde{C}_{s_{n}}) = \mu\!\left(\bigcup_{n=0}^\infty \widetilde{C}_{s_{n}} \right) \leq \mu(M) < \infty$  which is a contradiction.
\end{proof}
\noindent Since $\textsf{Plus}$ is a sub-machine of every universal Turing machine, using the fundamental theorem of sub-machines we have the immediate corollary:
\begin{corollary}
\label{corr:measurepreserve}
Let $\mu$ be a Borel measure  which  compact set $M \subset \R^n$ which assigns nonzero measure to all nonempty open sets and such that $\text{\rm supp}(\mu) = M$ and $\mu(M) < \infty$.  If $f : M \to M$ is a measure-preserving map with respect to $\mu$, then $f$ cannot be extended to a robustly Turing-universal CDS.
\end{corollary}

Recall that a diffeomorphism $f: M \to M$ when $M$ is a symplectic manifold with symplectic form $\omega$ is a symplectomorphism when $f^*\omega = \omega$. Such maps thus satisfy $f^*\omega^n = \omega^n$, and hence preserve the Borel measure associated to the volume form $\omega^n$. Thus Corollary \ref{corr:measurepreserve} proves in particular the following corollary: 

\begin{corollary}
\label{corr:symplectic-compact}
Symplectomorphisms of \emph{compact} symplectic manifolds cannot be extended to a robustly Turing-universal CDS.
\end{corollary}

In the next theorem, we partially drop the compactness assumption on the domain. Before we state the theorem, we will describe a prototypical class of physical dynamical systems which motivate this latter generalization.

\begin{definition}[Continuous-time integrable system]
Let $M$ be a $2n$-dimensional symplectic manifold, and let $f_t : M \to M$ be a Hamiltonian flow generated by the Hamiltonian function $H_1: M \to \R$, such that there exists complete system of commuting Hamiltonian functions $H_2, \ldots, H_n: M \to \R$, i.e.~we have that
\[ \frac{d}{dt}f_t = X_{H_1}\,, \quad i_{X_{H_j}} \omega = dH_j\,, \quad \omega(X_{H_j}, X_{H_k}) = 0 \quad \text{ for }\,\,\,i,j=1, \ldots, n\,.\] 
Here $d$ and $i$ are the exterior derivative and interior product \cite{arnol2013mathematical}. We suppose that there is a dense set of points where differentials of the $n$ conserved Hamiltonian functions are linearly independent, and also that the sets $\bigcap_{i=1}^n H^{-1}_i(c_i)$ are compact for $(c_i)_{i=1}^n \in \R^n$.  The dynamics instantiated by $f_t$ define a \emph{continuous-time integrable system}.
\end{definition}
\noindent A nice feature of continuous-time integrable systems is that due to the condition on the differentials of the conserved Hamiltonian functions, it is possible to use a diffeomorphism $h$ to locally transform into simple action-angle variables using the Liouville-Arnold theorem. More formally, for an open dense set of points $p \in M$, there are $f_t$-invariant neighborhoods $V$ of $p$ together with diffeomorphisms $h: V \to U \times T^n$ for some open set $U \subset \R^n$, satisfying the following properties. Writing $\tilde{f}_t := h^{-1} \circ f_t \circ h$, the map $\tilde{f}_t$ acts as $\tilde{f}_t : (x, y) \mapsto (x, y + t\,\omega(x)\,(\text{mod }1))$, where $\omega(x)$ is a frequency vector depending on the action variables $x$. In physical terminology, the sets $U$ are the action variables, $T^n$ are angle variables, and the dynamics translates the angle variables linearly with a rate depending on the action variables, instantiating periodic or quasiperiodic motion on each $T^n$. 

Recalling Remark~\ref{rem:modification1}, we can generalize our definition of a CDS to a continuous-time dynamical system by letting $\tau$ be a partial function $\tau : M \rightharpoonup \mathbb{R}_{\geq 0}$ instead of $M \rightharpoonup \mathbb{Z}_{\geq 0}$.  With this in mind, we present the theorem below, which in particular shows that continuous-time integrable systems are not robustly Turing-universal:

\begin{theorem}[Integrable systems cannot be extended to a robust CDS for the $\textsf{Plus}$ machine.]
\label{thm:integrable-systems}
Let $f_t$ be a smooth family of diffeomorphisms of a manifold $M$, and suppose that $f_t^*\mu = \mu$ for a measure $\mu$ of the form $g\,dV$ where $g: M \to M$, $g > 0$, and $dV$ is the volume form on $M$. Suppose also that $M$ can be written as a union of compact invariant submanifold (with boundary) for the flow of $f_t$.  Then $f_t$ cannot be extended to a robust CDS for the machine $\machine{Plus}: \{1\}^* \to \{1\}^*$ defined by $\machine{Plus}([n]_1) = [n+1]_1$.
% a robust Turing-universal CDS.
\end{theorem}
\begin{proof}
We argue by contradiction. Let $s_0 = 1$ be an initial configuration of $\text{Plus}$ where $\{s_n\}_{n = 0}^\infty$ with $s_n := \textsf{Plus}^n(s_0) = [n+1]_1$ are all distinct.  Since $C_{s_0}$ has a non-trivial interior, we can find a closed set $\widetilde{C}_{s_0}$ with non-trivial interior inside $C_{s_0}$ such that $\widetilde{C}_{s_0}$ is small enough to be contained within some compact invariant submanifold (with boundary) $K$ for the flow of $f_t$.  Defining $\widetilde{C}_{s_n} := (f^{\tau})^n(\widetilde{C}_{s_0})$, due to Lemma~\ref{lemm:intersection} we have that the $\widetilde{C}_{s_n}$ are pairwise disjoint. Then, recapitulating the argument of Theorem~\ref{thm:measurepreserve} with $f_t$ restricted to $K$, we find that $\mu(\widetilde{C}_{s_0}) = \mu(\widetilde{C}_{s_n}) > 0$ for all $n$, and so $\sum_{n = 0}^\infty \mu(\widetilde{C}_{s_n}) = \mu(\bigcup_{n=0}^\infty \widetilde{C}_{s_n}) \leq \mu(K) < \infty$ which is a contradiction.
\end{proof}

\noindent Again, using that $\textsf{Plus}$ is a sub-machine of every universal Turing machine, we can again use the fundamental theorem of sub-machines to attain the following corollary:

\begin{corollary}[Integrable systems cannot be extended to a robustly Turing-universal CDS]
If $f_t$ is an integrable system then it cannot be extended to a robustly Turing-universal CDS.
\end{corollary}

\begin{remark}
Even though integrable systems cannot furnish even the basic $\textsf{Plus}$ machine, they are capable of serving as `memories', i.e.~they can encode information in the values of their conserved charges.  As such, integrable systems serve as a model of memory \textit{without} the ability to perform certain forms of more sophisticated computation.
\end{remark}

\subsection{Time complexity bounds in one dimension}
\label{sec:time-complexity}
We now proceed to go \emph{beyond} obstructions to universality, and to prove time complexity bounds for various kinds of differentiable systems. Indeed, many previous works show that universality is often too strong to hope for in the vast majority of differentiable systems under most formalizations. The real task of a theory of computational dynamical systems seems to us to find a correspondence between various \emph{complexity classes} in the sense of classical computability theory and various \emph{dynamical classes} defined in terms of natural conditions on smooth dynamical systems. Thus, a goal is to be able to say that `dynamical systems of this kind are this computationally powerful'.  We discuss how to  formalize this notion in Appendix \ref{App:compclass}. The role of this paper is in part to make precise the gap between what is known in the enormous literature on differentiable dynamics and what would be a satisfying realization of the above goal.

We first proceed with the one-dimensional setting, where we can prove a relatively strong theorem. Recall the following definition:

\begin{definition} A differentiable map $f: I \to I$ is \textbf{Axiom A} if 
\begin{enumerate}
\item All periodic points of $f$ are hyperbolic (and thus classified into \emph{attracting} or \emph{repelling} periodic points),  and 
\item Let $B(f) = \bigcup_p W^s(p)$ where the union runs over all attracting periodic points. Then $[0,1] \setminus B(f)$ is a hyperbolic set.
\end{enumerate}
\end{definition}

\begin{remark}
    This is a standard generalization of the notion of $1$-dimensional Axiom A diffeomorphism to the case of maps which are not necessarily invertible. Note that diffeomorphisms of the interval are relatively simple dynamical objects \cite{katok1995introduction}, and in particular, the condition that an interval diffeomoprhism is Axiom A is extremely restrictive. In contrast, Axiom A maps are generic among all interval maps by the theorem below. The theory of Axiom A maps of domains of dimension greater than one is much less developed than the corresponding theory for diffeomorphisms \cite{moosavi2023smooth}. 
\end{remark}

\begin{theorem}[\cite{kozlovski2007density}]
Axiom A maps $f: I \to I$ of class $C^k$ where $I = [0,1]$ are dense in the space of $C^k$ maps $C^k([0,1], [0,1])$ for $k=1, 2,\ldots, \infty$, or even $\omega$ (the real analytic maps); this density result even holds for polynomial maps of the interval. Moreover in each of the cases the Axiom A maps are exactly the structurally stable maps. 
\end{theorem}

\begin{theorem}
\label{thm:time-complexity-1-d-b}
    Let $f: [0,1] \to [0,1]$ be an Axiom A map. Then for any CDS $(f, \mathcal{E}, \mathcal{D}, \tau, \,\machine{T})$ where $t(n) = n$, $\tau$ is constant, and the encoder and decoder can be implemented with \textnormal{BSS$_{\textsf{C}}$} machines, then if $\machine{T}$ halts on a configuration then it will halt in time $O(F(n))$, where 
    \[ F(n) =D^{2^n} \]
    for some constant $D$.
    In other words, in the sense of Appendix \ref{App:compclass}, $f$ can recognize languages at most in $\textnormal{\textbf{DTIME}}(O(F(n)))$. 
\end{theorem}

Before proceeding with the proof of this theorem, we first prove several sublemmas. First, note that by \cite[Theorem~16.1.1]{katok1995introduction} we have that $f$ induces dynamics on $R(f) := [0,1] \setminus B(f)$ which are topologically conjugate to a topological Markov chain \cite{katok1995introduction}; thus this subset of the interval is a Cantor set with the induced topology, and is thus totally disconnected. In particular, any region $C_s$ associated to a robust decoder for a CDS involving $f$ will contain points outside of $B(f)$, which will thus have dynamics that asymptote to the attracting hyperbolic periodic points. To get a halting bound we will understand how long it will take for those points to get there. This will follow from the following two lemmas. The first is purely dynamical:

\begin{lemma}
\label{lemma:convenient-cover}
    There exists a finite collection of closed intervals $\{I_j\}_{j=1}^r$ whose union covers $R(f)$ and an $n$ such that $f^{-kn}(I_j)$ is a disjoint union of closed intervals all contained in $I = \bigcup_j I_j$ and such that the lengths of each of these intervals is bounded by $\lambda^kC$ for some $\lambda < 1$ and $C > 0$. 
\end{lemma}
\begin{proof}
    This is essentially proven in \cite[Theorem~16.1.1]{katok1995introduction}; our $I_j$'s are those of this theorem, as is our $n$. The only point to make is that if we chose the $I_j$'s small enough, so that on each $I_j$, we have that $f^n$ is monotone and has  $|(f^n)'|$ lower bounded by $\lambda^{-1} > 1$ on all the $I_j$. This is possible by hyperbolicity and the total disconnectedness of $R(f)$. Then, by the derivative bound and and the monotonicity of $f^n$, $f^n$ must stretch the length of each component of $f^{-n}(I_j)$ by a factor of at least $\lambda$. But then an induction on $k$ proves the lemma.
\end{proof}
Beyond this covering lemma, we have the following helpful definition and result from real algebraic geometry:

\begin{definition}[Standardized polynomial]
We say that a polynomial $p(x) = \sum_i a_i x^i$ is \textbf{standardized} if all of its non-zero coefficients are greater than or equal to $1$. For any $p(x)$, we let
\begin{align}
\widetilde{p} := p/\min\{1,\min\{|a_i|\, : \, a_i \not = 0\}\}\,,
\end{align}
which is standardized.  Indeed, if $p$ is already standardized, then $\widetilde{p} = p$.
\end{definition}

\begin{proposition}[Generalization of Theorem 3 of~\cite{rump1979polynomial}]\label{prop:real-algebraic-bound}
Let $p(x) = \sum_i a_i x^i$ be a real polynomial with integral coefficients. Denote $s = \sum_i |a_i|$ and $d = \deg p$. Let $\rm{rsep}_{[-1,1]}(p)$ be the minimum distance between two real roots of $p$ which are each in the interval $[-1,1]$. Then 
\[ \text{\rm{rsep}}_{[-1,1]}(p) \geq \frac{2 \sqrt{2}}{d^{d/2 + 1}(\tilde{s}(p) + 1)^d}\,. \]
\end{proposition}
\noindent This result is proved in Appendix \ref{App:root-separation}.  We now proceed with the proof of Theorem \ref{thm:time-complexity-1-d}.
\begin{proof}[Proof of Theorem \ref{thm:time-complexity-1-d-b}]

Denote the domain of $f$ by $M$. There is a region $\mathcal{H} \subset M$ which is a union of closed nonempty intervals of strictly positive length and should be thought of as the `halting region', corresponding to the union of all the decoded regions associated to configurations for which the underlying state of the Turing machine is the halting state.

Let $p_1, \ldots, p_k$ be the attracting periodic points. Choose an $\epsilon$ such that the neighborhood $U_\epsilon$ of size $\epsilon$ of the attracting periodic points is a union of intervals, one around each $p_i$, with standard dynamics given by the local normal forms \cite{kato2015global} for attracting periodic points. 

Now let $V = \bigcup_i I_j$ with $I_j$ as in Lemma \ref{lemma:convenient-cover}. We know that $R(f)^c = \bigcup_{i=1}^\infty f^{-i}(U_\epsilon)$; in particular, there is some $L$ such that $f^{-L}(U_\epsilon)$ contains $V^c$. Lemma \ref{lemma:convenient-cover} then implies that $f^{-L+k}(U_\epsilon)$ can be covered by disjoint intervals of size exponentially decreasing in $k$. 

Let $s$ be a configuration of $\machine{T}$. Then because $C_s$ is given by a BSS$_{\textsf{C}}$ machine, it is given via a union of the set of solutions to several equations $P_{s_i}(x)/Q_{s_i}(x) \geq 0$ for some polynomials $P_{s_i}(x)$ and $Q_{s_i}(x)$ which can be taken to be standardized (by rescaling). In particular, it is a union of intervals of length bounded by the distance between the roots of the polynomials $P_{s_i}$. We will bound the quantities $s$ and $d$ for these polynomials in terms of the length $n$ of $s$. 

The fact that $t(n)=n$ means that the formula for $C_s$ is computed in $O(n)$ steps. Indeed, the formula for this semialgebraic set is computed by running the BSS$_{\textsf{C}}$-machine defining $\mathcal{D}$. At each step the machine doing this computation has in its registers several real quantities which in terms of the original quanties fed into $\mathcal{D}$ take the form of ratios of polynomials in these quantities. At each step, either the computation of $\mathcal{D}$ branches on the positivity of one of these quantities; or the computation is in the reject state; or otherwise the computation continues adding, multiplying, and dividing the available quantities (which are all rational functions of the original inputs). If we have $P_1, P_2, Q_1, Q_2$ all satisfying $\deg P_i = O(w), \deg Q_i = O(w)$, while $s(P_i)$ and $s(Q_i)$ are both $\leq D^n$ for some large $D$, then  $P_1/Q_1 + P_2/Q_2 = P'/Q'$, where $\deg P'$ and $\deg Q'$ are $O(2w)$, $s(P')$ and $s(Q')$ are $\leq D^{2n}$, and $P'$ and $Q'$ are still standardized if $P_i$ and $Q_i$ were. Multiplying the available rational functions gives similar bounds. We can also add a constant or multiply a rational function by a constant; to preserve standardization of the numerator and denominator, we will have to rescale $P$ and $Q$ by the relevant constant, which affects e.g.~$s(P')$ by scaling it by $\deg P' C = O(w)C$.

Finally, if the computation underlying the formula at some point branches on the positivity of some value of some rational function $P_1/Q_1$ of the original inputs to $\mathcal{D}$, and while the other quantities available to the machine computing $\mathcal{D}$ are of the form $P_2/Q_2, \ldots$, then the computation of $\mathcal{D}$ continues, then this simply corresponds in the end to considering regions which are a union of regions in between the roots of $(P_1/Q_1) \cdot F(P_2/Q_2, \ldots)$ for some final formula $F$. Combining these observations together with Proposition \ref{prop:real-algebraic-bound}
    tells us that $C_s$ consists of intervals of length bounded by a constant times  
    \[ 2 \sqrt{2}((2^n)^{(2^n)/2+1} (2^{n(n+1)/2}C^nD^{2^n}+1)^{D^{2^n}+1})^{-1}.\]
    Here the factor $2^{n(n+1)/2}C^n = 2^nC2^{n-1}C2^{n-2}C\ldots$ is a lossy bound associated to the condition that we might have to to `restandardize' the polynomials at each step.  Thus after we iterate $f$ $L+k(n)$ times, where 
    \[k(n) = O(\log ((2^n)^{(2^n)/2+1} (2^{2n}C^nD^{2^n}+1)^{D^{2^n}+1}))= O(2^n n + D^{2^n}2^n)= O(D^{2^n})\]
     there  will be a point $p$ of $f^{L+k(n)}(C_s)$ which lies in $U_\epsilon$, say in a neighborhood of a periodic point of period $\ell$. 

     At this stage, it is comparatively easy to figure out if $\machine{T}$ will halt. The first possibility is that for $2\ell$ more iterations of $T$, there will still be a point of $f^{L+k(n)+i}(C_s)$, $i=1,..., \ell$ which stays in the neighborhood $V$, at which point by the $2\ell$-th iteration we will have that $f^{L+k(n)+2\ell}(C_s)$ overlaps $f^{L+k(n)}(C_s)$ which is a contradiction if we have not yet halted. The other possibility is that if the whole neighborhood leaves $V$ during these $2\ell$ iterations, we can wait another $f^L$ iterations to have the whole interval lie in $U_\epsilon$; we see in fact that at this stage, the image of the whole interval must lie in the same component of $U_\epsilon$ as the image of the original point entering $U_\epsilon$ does, for otherwise the interval will never decrease in size; but this could only happen if the interval in fact contained a point of $R(f)$, which would contradict the whole interval leaving $V$ eventually. In any case, once the whole interval lies in a single component of $U_\epsilon$, we must halt in some fixed time and overall bounded time $Q$ depending on how the halting set overlaps $U_\epsilon$, or never halt at all. This concludes the proof of the case when $\tau=1$.
\end{proof}

\begin{remark}
\label{rk:annoying-time-complexity-1d-remark}
    Note that much of the proof of Theorem \ref{thm:time-complexity-1-d-b} goes through for general $\tau$. Indeed,  the proof above really proves that $f^k(C_s)$ reaches the region near the hyperbolic attracting periodic point once $k$ is greater than the bound in the Theorem for $n=|s|$. Thus, if $\tau \geq 1$, we have that iterating $f$ on $C_s$ reaches the hyperbolic attracting periodic point \emph{faster}, and subsequently $f^k(C_s)$ stays near this hyperbolic attracting point. Thus that part of the bound of Theorem \ref{thm:time-complexity-1-d-b} holds for general $\tau$, as the bound is only \emph{improved} by increasing $\tau$. The reason we cannot state the theorem for general $\tau$ is because of the possibility that $f^k(C_s)$ eventually gets close to the hyperbolic periodic point but the slowdown function $\tau$ conspires so that on the collection of iterations $k=k_i$ of $f^k(C_s)$ on which we are to evaluate the dynamics of $f$ for the purposes of simulation, the $k_i$ always have the wrong value modulo the period of the periodic point. Thus, if one has an oracle which can resolve this challenge (and one can produce one for many slowdown functions that are much more general than the constant function), then one can continue to conclude the bound of Theorem \ref{thm:time-complexity-1-d-b}.
\end{remark}

\subsection{Time complexity bounds in many dimensions}
\label{sec:time-complexity-2}

In the theory of Axiom A systems (see Section \ref{sec:axiom-a}), the notion of a `topologically mixing' system occurs. There is an elementary statement showing that topologically mixing regions cannot encode non-halting computations:
\begin{lemma}
\label{lemma:topologically-mixing-implies-halting}
    Let $f$ be topologically mixing. Then any robust CDS with underlying dynamical system $f$ must halt on all inputs. 
\end{lemma}
\begin{proof}
    Any configuration $s$ of the machine $\machine{T}$ underlying the CDS defines a closed set with nonempty interior, namely $C_s = \mathcal{D}^{-1}(s)$. Similarly, there is a set $\mathcal{H}$ with nonempty interior which corresponds to the configurations of $\machine{T}$ in the halting state. Let $C^o_s$ denote the interior of $C_s$. Since $f$ is topologically mixing, then there is an $N$ such that for all $n > N$, $f^n(C^o_s) \cap \mathcal{H} \neq \emptyset$. Suppose that $\machine{T}$ never halts on $s$; then, choosing some $n$ large enough we must have that $(f^{\tau})^n(C^o_s) \cap \mathcal{H} = \emptyset$, which is a contradiction with the previous statement. 
\end{proof}

Now, for Axiom A diffeomorphisms, the spectral decomposition of Proposition \ref{prop:spectral-decomposition} decomposes the dynamics into basic sets on which the dynamics are topologically mixing in a highly quantitative way, and regions where the dynamics flows from one basic set to another. Thus to establish a higher-dimensional analog of Theorem \ref{thm:time-complexity-1-d} one must understand the `amount of computation' that can be done in a single basic set. The special case where there is one basic set which is exactly the entire the domain of $f$ is the setting where $f$ is Anosov. Already in this setting the problem is nontrivial; below we give an inexplicit computable bound on the halting time of any CDS with underlying Anosov dynamics in order to clarify the difficulties of the problem. First we must introduce a restriction on the type of decoder allowed which abstracts away the properties of the robust Cantor decoder of \eqref{eq:cantor-decoder}.

\begin{definition}
\label{def:cantor-like-decoder}
    A \emph{Cantor-like} decoder $\mathcal{D}: M \rightharpoonup S$ is one that is implemented as follows: there is a semialgebraic map $P: M \times \R^k \to M \times \R^k$, a semialgebraic (thus piecewise constant) map $\mathcal{D}_0: M \times \R^k \to \Sigma \cup \{\tilde{p}\}^2$ (here $\tilde{p}$ is a new symbol) and an initial condition $r_0 \in \R^k$, such that 
    \[ \mathcal{D}(x) = s_{-k_1}\ldots s_0 \ldots s_{k_2}\]
    is outputted as 
    \begin{align*}
    x &= x_0 \\
    (s_{-1},s_0) &= \mathcal{D}_0(x_0, r_0) \\
    x_i, r_i &= P(x_{i-1}, r_{i-1}) \\
    (s_{-i-1}, s_i) &=\mathcal{D}_0(x_i, r_i)
    \end{align*}
    where if one of $s_{-i-1}$ or $s_i$ is not defined in $\mathcal{D}(x)$, then $\mathcal{D}_0$ outputs $\tilde{p}$ for that value; and the output halts when one of $P$ or $\mathcal{D}_0$ reaches a value for which its output is not defined.

    Note that every Cantor-like decoder is optimal, i.e.~it takes at most $O(1)=C$ steps to compute each next pair of symbols of $\mathcal{D}(x)$. We call this constant $C$, the largest number of steps it might take to compute one more symbol of $\mathcal{D}(x)$, the \emph{rate constant} of $\mathcal{D}$. 
\end{definition}

\begin{theorem}
\label{thm:time-complexity-anosov}
    Let $f: M \to M$ be Anosov and volume-preserving. Consider a CDS $(f, \mathcal{E}, \mathcal{D}, \tau, \textnormal{\textsf{T}})$ where $t(n) = n$, the decoder is Cantor-like (Definition \ref{def:cantor-like-decoder}), and the encoder and decoder are implemented by \textnormal{BSS$_{\textsf{C}}$} machines with the finite set of computable constants being $\{a_1, \ldots, a_\ell\}$.  Then the CDS halts on all configurations  in time $O(\mathcal{C}(n))$, where $\mathcal{C}(n)$ is a computable function depending on $a_1, \ldots, a_\ell$.
\end{theorem}
\begin{proof}
This follows from the ergodic theory of Anosov diffeomorphisms, together with an elementary finiteness argument. Knowing the constants used in $\mathcal{D}$, there are only a finite number of possible semialgebraic maps $P$ and $\mathcal{D}_0$ such that $(P, \mathcal{D}_0)$ can be represented by a formula which takes time at most $C$ to compute \cite{smale1997complexity}. Thus the sets $\mathcal{D}^{-1}(s)$ must each be a union of an \emph{explicit} finite set of semialgebraic sets with nonempty interior (the finiteness is by the fact that the number of connected components of a semialgebraic set is bounded by the complexity of the defining formula \cite{smale1997complexity}, together with the finiteness of the number of maps $(P, \mathcal{D}_0)$). In particular, there is a computable lower bound on the inradius of the sets $C_s = \mathcal{D}^{-1}(s)$ as a function of the length $|s|$. For example, for every possible choice of $(P, \mathcal{D}_0)$, one can compute the cylindrical algebraic decomposition (see \cite[Chapter~5]{Basu2006-pd} for a textbook treatment, or \cite{jirstrand1995cylindrical} for a gentle introduction) of the corresponding set $C_s$; the defining property of the cylindrical algebraic decomposition lets us find a cube  in the corresponding cell of the cylindrical algebraic decomposition of $C_s$ (which will be of full dimension in the corresponding component), at which point finding a computable lower bound on the inradius is elementary. Now, given a ball $B_r$ in $C_s$ of radius $r$, the exponential mixing properties of $f$ \cite{bowen2008equilibrium} imply that there are constants $C$ and $\kappa$ independent of $B_r$ or $r$ such that $f^N(B_r) \cap \mathcal{H} \neq \emptyset$ for all $N \geq N_0$ where $N_0$ is the minimum quantity such that $r^n > Ce^{-N_0\kappa}$, i.e.~such that 
\[ N_0 \geq (\log(C)-n\log r)/\kappa \]
(note that $r \ll 1$ so the quantity on the right-hand side grows large with small $r$). Combining this with the previous observation proves the theorem.
\end{proof}

It is quite interesting to try to get an explicit value for the computable function $\mathcal{C}(n)$. The primary challenge is to get an explicit bound for the volume (or inradius) of a semialgebraic set  (which is the closure of its interior) defined by polynomial inequalities with integer coefficients, in terms of a bound on the coefficients of the defining inequalities. Such a bound, which would be the higher-dimensional generalization of Rump's bound \cite{rump1979polynomial} stated in Proposition \ref{prop:real-algebraic-bound} above, does not appear to have been established. This is a basic question in real algebraic geometry that must require a fusion of the methods of \cite{rump1979polynomial} with the mathematics of the Fuschian differential equations satisfied by periods which control the volumes of semialgebraic sets defined by BSS$_{\mathbb{Q}}$ machines \cite{lairez2019computing}. 

\section{Dynamical mechanisms for computation}

Throughout the paper, we have used features of families of dynamical systems to constrain the kinds of computation that they can encode.  In doing so, we have identified a number of dynamical `mechanisms' which facilitate or manifest certain types of computation.  In this section we collect and summarize these mechanisms with the hope that they will be useful in future inquiries, and that the list will be enlarged in future works.

\begin{enumerate}
    \item \textbf{Thickened Smale horseshoes provide robust shift maps for finite strings.} In our construction of robust Turing-universal CDS in Section~\ref{sec:example-of-robust-cds}, we considered a nonlinear variant of the Baker's map giving rise to a modified version of the Smale horseshoe.  Ordinarily, the Smale horseshoe provides a dynamical mechanism for implementing the shift map on bi-infinite string encoded into a two-dimensional Cantor set (see e.g.~\cite{cardona2021constructing} for an application involving the construction of a Turing-universal, but not robust, dynamical system on the disk).  The problem is that the ordinary Smale horseshoe does not give rise to robust computation in our sense, since the Cantor encoding encodes bi-infinite bit strings into individual points, which accordingly do not have a non-trivial interior.  However, when it comes to constructing a Turing machine, we only need to consider the set of two-sided bit strings with \textit{bounded} length, which is a countable set.  In particular, if a Turing machine starts with a tape that has a finite number of blank symbols, then at any finite time step it will still have only a finite number of blank symbols.  As such, we constructed an encoding of corresponding two-sided bit strings into \textit{closed sets with non-empty interior} associated with the Cantor construction, and then constructed a thickened version of the Smale horseshoe (via a nonlinear Baker's map) which permutes the aforementioned closed sets.  As such, the thickened Smale horseshoe provides a `robustified' version of the shift map on finite two-sided strings.
    \item \textbf{Attractors with topologically mixing dynamics limit memory storage.} In Section~\ref{sec:axiom-a}, we proved that Axiom A systems are not robustly Turing-universal.  One of the key steps of the proof is to use the spectral decomposition of Axiom A systems~\cite{smale1967differentiable, bowen2008equilibrium} to decompose the non-wandering set of the dynamics into a finite disjoint union of sets $\Omega_{i}$ on which (an iterate) of the dynamics is topologically mixing.  The structural stability of Axiom A systems~\cite{robbin1971structural, robinson1976structural, mane1987proof} in tandem with the stable manifold theorem for hyperbolic sets~\cite{smale1967differentiable, hirsch-pugh, bowen2008equilibrium} allows us to prove that if there is a bit string that we want to encode in the dynamics and store for all time, it will encode a point that gets soaked up into a dense orbit of an $\Omega_i$; as such the $\Omega_i$ soaks up one finite bit string's worth of memory.  Then the number of $\Omega_i$'s limits the memory storage capacity of any computation instantiated by an Axiom A system.  More broadly, topologically mixing dynamics in a subset of the configuration space of a dynamical system is highly constrained, forcing any encoded strings within that subset to dynamically map into one another.  So for example, if the halting set overlaps with a part of the subset of the configuration space that is topologically mixing, any other bit strings which ultimately land in that subset will at some time land in the halting set. 
    %Units of memory
    \item \textbf{Measure-preserving dynamics on a compact set prevents computation accessing infinitely many states.} In measure-preserving dynamics, a region of the configuration space which encodes a string can never shrink.  If the configuration space of the dynamical system is compact, this means that the forward orbit of the encoded set can only visit the encoded sets of a finite number of strings before it exhausts the size of the compact space.  Accordingly, in a measure-preserving system, the computation is such that the forward orbit of every string is a finite set.  This prevents us from encoding e.g.~the $\textsf{Plus}$ machine into measure-preserving dynamics (see Theorem~\ref{thm:measurepreserve}), thus obstructing Turing-universality as well as other kinds of computations.
    % Finite state space
    % contractive dynamics too
    \item \textbf{Periodic orbits can store memory.} Part of our intuition in Section~\ref{subsec:nonunivmeas1} is that periodic orbits can be used to store memories.  One way is to use the period of the orbit itself to store information.  For instance, consider a family of integrable Hamiltonian systems, such as pendulums with varying tensions.  The tension dictates the period of oscillations in a robust manner in the regime of small oscillations (i.e.~the harmonic oscillator regime), and thus the `memory' stored in the pendulum could be regarded as its period.  Even for fixed tension, in the nonlinear regime of large oscillations the period of the pendulum depends more strongly on the initial conditions, which can also serve to encode the `memory'.  These considerations generalize to integrable Hamiltonian systems more broadly, and also to dynamical systems with multiple periodic orbits of varying lengths in their configuration space.
    %Contingent on slowdown; memories
    \item \textbf{Exponential mixing forgets the initial input to the computation.} In Sections~\ref{sec:time-complexity} and~\ref{sec:time-complexity-2} where we discuss time complexity bounds of Anosov systems, we exploit exponential mixing properties to bound the halting time of encoded dynamics.  In essence, the exponential mixing gives quantitative bounds on the timescale at which an initial encoded region will ultimately overlap with another encoded region.  That is, if we start in a configuration $s$ encoded into a region of known size, we can bound how quickly that region maps into another one corresponding to $s'$ with a known size.  In other words, we are bounding how quickly $s$ will be mapped into $s'$ (e.g.~we can imagine $s'$ being in the halting set).  As such, the exponential mixing mediates how quickly the computation `forgets' the initial string and maps it onto another arbitrary string $s'$.
\end{enumerate}

\section{Open Problems}
In this final section, we describe some natural open problems and research directions in the theory of computational dynamical systems. The problems we outline below involve an interplay between aspects of circuit theory or real algebraic geometry, aspects of dynamical systems theory, and questions that are natural from a computer science perspective. We hope that these problems will be of interest to others.

\paragraph{Quantitative analysis of Anosov diffeomorphisms and related problems.}
In Section \ref{sec:axiom-a}, we showed that an Axiom A diffeomorphism, in any dimension, cannot robustly implement machines which have infinitely many cycles in their state space. In fact, the proof implies a sharp bound on the number of possible distinct cycles in the state space of any machine implemented by an Axiom A diffeomorphism. This naturally leads to the question of a more quantitative analysis of the computational power of Axiom A diffeomorphisms and maps, which we began to investigate in Section \ref{sec:time-complexity}.

Anosov diffeomorphisms, which are the simplest Axiom A diffeomorphisms, constitute an interesting case for questions in computational dynamical systems theory, because establishing complexity bounds on Anosov systems is connected to the theory of error correcting codes. On an informal level, the problem is whether it is possible to efficiently encode and decode states into strongly chaotic systems such that these chaotic systems, which statistically behave like pure random number generators which rapidly `forget' their state, are nonetheless able to perform sophisticated computations. For example: can an ideal hard-ball gas in a box serve as a computer, given the correct interpretation of its state space? In order to show that in some sense the answer is \emph{no}, one can use the language of this paper to formalize this problem, and a mathematical answer is not straightforward. 

\begin{remark}
    The formalization of computational dynamical systems in this paper is connected to to certain discussions in the analytic philosophy literature on the computational theory of mind,  referred to as the `computational triviality problem' of Hinkman, Searle, Putman, Chalmers, and others (e.g.~`Does a rock implement every finite-state automaton' \cite{chalmers1996does}; see also \cite{sprevak2018triviality} for  a review). In some sense this paper gives a mathematical formalization and partial resolution of this problem. We thank Rosa Cao for explaining this connection.
\end{remark}

As described in Section \ref{sec:time-complexity-2}, the basic reason we are not able to prove a quantitative complexity bound for Anosov diffeomorphisms is because the method of Theorem \ref{thm:time-complexity-1-d-b} relies on a lower bound of the size of a real-algebraic set in dimension $1$ (Proposition \ref{prop:real-algebraic-bound}). There appears to be no known analog of this statement for regions in $\R^n$ defined by real algebraic circuits if $n>1$.

\begin{problem}
\label{problem:semialgebraic-complexity-bound}
Give a generalization of Proposition \ref{prop:real-algebraic-bound} to the setting of semialgebraic subsets of $\R^n$. Namely, let $K$ be a connected component of a semialgebraic set defined by $m$ algebraic inequalities of $n$ variables $p_i(x) \geq 0$, $i =1, \ldots, m$, and let $w_1, \ldots, w_\ell$ be a collection of explicit functions of the coefficients of the polynomials $p_i$. Write $B(x, r)$ for the ball of radius $r$ around $x$. Give an explicit lower bound
\[ C(n, m, w_1, \ldots, w_\ell) \leq \inrad(K) = \sup_{r > 0: B(x, r) \subset K} r\]
that holds whenever $K$ has nonempty interior. 
\end{problem}

\noindent Unfortunately, without a methodological improvement, progress on Problem \ref{problem:semialgebraic-complexity-bound} will give at best tower-exponential bounds for the Anosov setting of Theorem \ref{thm:time-complexity-anosov}. However, note that this setting, while higher-dimensional, is in some ways \emph{simpler} than the setting of 1-dimensional Axiom A maps, since the Anosov diffeomorphism system is exponentially mixing everywhere. However, even though one has an excellent understanding of the dynamics of an Anosov diffeomorphism from a symbolic dynamics perspective (Appendix \ref{app:symbolic-dynamics}), the basic challenge is that the encoding and decoding of states in a CDS does not have to be \emph{compatible} with the symbolic dynamics in any elementary fashion.

\begin{remark}[Decidable halting problem for measure-preserving dynamics]
If Problem~\ref{problem:semialgebraic-complexity-bound} can be solved with an explicit bound, then it would imply that measure-preserving dynamics on a compact set can only encode Turing machines with a decidable halting problem.  The argument is as follows.  Suppose we have a putative CDS involving $f : M \to M$ which instantiates measure-preserving dynamics for a measure $\mu$ on a compact set $M$, and let $\mathcal{D}^{-1}(s) = C_s \subset M$ for some $s$.  The forward orbit of $C_s$ under $f$ can only intersect finitely many $C_{s'}$'s since $M$ is compact and $C_s$ has non-trivial interior.  As such, in the Turing machine the forward orbit of any $s$ can only visit a finite number of configurations.  Then the Turing machine initialized with $s$ either halts or enters into a periodic orbit.  Contingent on a solution to Problem~\ref{problem:semialgebraic-complexity-bound}, we have an explicit lower bound on the radius of a ball inside $C_s$ as a function of $|s|$.  Let us call the lower bound $r_{C_s}$.  Then the number of points in the forward orbit of $s$ is at most $\mu(M)/\mu(B_{r_{C_s}})$, and as such the halting problem is decidable for the CDS.
\end{remark}

\begin{problem}
\label{problem:anosov-algorithm-for-halting-time}
Given a CDS with an underlying differentiable dynamical system $f$ which is Anosov, is there an efficient algorithm which will predict how long it takes for the system to halt on an input string? For example, is there a polynomial-time algorithm for this problem? Here, we do not ask that $f$ is computable, and we are simply asking for the \emph{existence} of such an algorithm for any $f$, rather than for an algorithm for finding this algorithm given $f$.
\end{problem}

\begin{problem}
\label{problem:anosov-bound-on-halting-time}
Given a CDS with an underlying differentiable dynamical system $f$ which is Anosov, will it halt in polynomial time in the input size?
\end{problem}
\noindent Relatedly, one would like to improve the bound of Theorem \ref{thm:time-complexity-1-d-b} to something much better than the tower-exponential bound given in the theorem statement. This may be difficult without putting additional restrictictions on the architecture of the encoder and decoder, which we turn to next. 

\paragraph{Varying conditions on encoders and decoders.}
Much of the difficulty of the questions of the previous section (e.g.~Problems \ref{problem:anosov-algorithm-for-halting-time} and \ref{problem:anosov-bound-on-halting-time}) is that not so much is known about the theory of circuits, whether Boolean or real-algebraic. Indeed, even for a Cantor-like decoder (Definition \ref{def:cantor-like-decoder}) as used in Theorem \ref{thm:time-complexity-anosov}, the structure of the encoder and decoder is essentially controlled by an auxiliary high-dimensional real-semialgebraic dynamical system, which may have arbitrarily complex dynamics; this is odd when compared to the  comparatively simple (e.g.~Anosov, Axiom-A) dynamics of the system being encoded into! Thus, while our definition of a optimal complexity encoder-decoder is  `simple'  from the perspective of computer science, it still allows for extraordinarily complex dynamical behavior. 

In fact, with the current definitions of encoders and decoders used in this paper, Problems \ref{problem:anosov-algorithm-for-halting-time} and \ref{problem:anosov-bound-on-halting-time} have corresponding variants where $f$ is an integrable system, and already the corresponding results are not straightfoward to establish. Although one can can certainly establish such bounds by solving Problem \ref{problem:semialgebraic-complexity-bound},  a more fruitful strategy may be do modify the definition of a Cantor-like decoder in a natural manner such that better methods of proof become available. 

\begin{problem}
Find a natural strengthening of the notions of a Cantor-like decoder which forces it to behave `like' the robust Cantor decoder defined in Theorem \ref{thm:TuringUniversalCDS1}. This notion should not be too `specific', i.e.~it should allow for behavior significantly more general than that of the robust Cantor decoder; however, using this notion should make it possible to significantly strengthen the bounds in Problems \ref{problem:anosov-algorithm-for-halting-time} and \ref{problem:anosov-bound-on-halting-time}, and the other problems of this section.
\end{problem}

More generally, varying the condition on the encoders and decoders allowed for any given problem about computational
dynamical systems tends to highlight different aspects of the problem, and can make a problem either interesting, trivial, or extremely difficult. For example:

\begin{problem}
Prove Theorem  \ref{thm:generic-diffeo-not-universal-under-hierachical-shrinking} while dropping the shrinking condition on the decoder. 
\end{problem}

\noindent It is also natural to approach Conjecture \ref{conj:generic-maps-are-not-universal} with various conditions on the encoder-decoder pair. 

Finally, one can conceive of a weaker notion of robustness where one works with a hierachically shrinking decoder (Definition \ref{def:shrinking-decoders-and-so-forth}), and one only asks (in the notation of that definition) that if $\machine{T}^n(s) = s'$ then $\mathcal{D}((f^{\tau})^n(C'_{s_{\text{padded}}})) = C'_{s'}$, where $s_{\text{padded}}$ denotes the string $s$ prepended and appended with some number of zeros, i.e.~there is some amount of `zero padding' around $s$. This notion encodes the idea that one may need to `encode a string more and more accurately to simulate it for longer and longer'. Modifying the problems in this paper to utilize this latter notion highlights yet a different aspect of the computational properties of continuous dynamical systems, which deserves to be explored in future work. 

\paragraph{Problems on finite state machines and other complexity classes.}
Below, a finite state machine is taken to be a Turing machine that only moves to the right on its tape.  We expect that there are positive answers to the following two problems:
\begin{problem}
Is there an Axiom A diffeomorphism $f$ which is universal for finite state machines, e.g.~such that for any finite state machine $F$ there is a Cantor-like encoder-decoder pair which makes $f$ implement $F$? 
\end{problem}

\begin{problem}
Fix an analog of the Cantor decoder $\mathcal{D}$ (as well as an analog of the encoder $\mathcal{E}$) of Theorem \ref{thm:TuringUniversalCDS1}. For every finite state machine $F$, does there exists an Axiom A diffeomorphism $f_F: \R^2 \to \R^2$ such that $(f_F, \mathcal{E}, \mathcal{D}, \tau, F)$ is a robust CDS for $F$?
\end{problem}
\noindent These two problems would clarify the more traditional correspondence between Axiom A systems and finite state machines via symbolic dynamics.  From the perspective of computational dynamical systems, the class of Axiom A systems should be `Finite-State-Machine-complete'.

We can generalize these problems to other complexity classes and other classes of dynamical systems. For example:

\newcommand{\complexityC}{\mathbf{C}}
\newcommand{\complexityD}{\mathbf{D}}

\begin{problem}
\label{problem:C-universal-diffeomorphism}
Does there exist a single Anosov diffeomorphism $f: M \to M$  (where $M$ is a computable manifold) such that for every decision problem $L \in \complexityC$ (where $\complexityC$ is a complexity class like $\textnormal{\textbf{P}}$, $\textnormal{\textbf{NP}}$, $\textnormal{\textbf{NEXPSPACE}}$, etc.) there exists a Turing machine $\textnormal{\textsf{T}}$ solving $L$ and an optimal complexity encoder-decoder pair $\mathcal{E}, \mathcal{D}$ such that $f, \mathcal{E}, \mathcal{D}$ implements $\textnormal{\textsf{T}}$? In this case, we may say that $f$ is $\complexityC$-complete. For example, we can also ask if the Arnold cat map in particular is $\complexityC$-complete? \end{problem}

\noindent There is an obvious relationship between Problems \ref{problem:anosov-bound-on-halting-time} and Problem \ref{problem:C-universal-diffeomorphism}; however, the requirement that $f$ is $\complexityC$-complete is a much stronger property than e.g.~a runtime bound, and there may be correspondingly stronger results. If a given $f$ is \emph{not} $\complexityC$-complete then there is some problem in $\complexityC$ that $f$ cannot implement.  If a given class of dynamical systems $\texttt{D}$ contains no elements which are $\complexityC$-complete, then this gives a precise meaning to the idea that `systems in class $\texttt{D}$ cannot implement computations in class $\complexityC$', i.e. that the classical computational complexity class $\complexityC$ is `more complex' than the dynamics available in $\texttt{D}$. 

One of the most satisfying aspects of the basic theory of computational complexity is the existence of various complexity classes and the (only very partially understood)
inclusions and separations between these classes. In differentiable dynamics, researchers also consider a
‘complexity hierarchy’, but rather than being governed by computational considerations, this hierarchy is
structured largely by ‘infinitesimal’ or ‘dynamical’ conditions, e.g.~a loss of uniformity of hyperbolicity (see \cite{hasselblatt2006partially} for a textbook review, and \cite{pomeau1980intermittent} for the connections to intermittency and turbulence). 

The language above lets us formalize a problem which we feel is particularly interesting:
\begin{problem}
    \label{problem:intermediate complexity}
    Do there exist natural classes of differentiable dynamical systems which are universal for finite state machines (or another intermediate complexity class) but cannot be extended to a robustly Turing-universal CDSs?
\end{problem}
\noindent This problem highlights the significant gap between the dynamical perspective and the computational perspective. This paper was largely written with the motivation of rephrasing this conceptual gap into a series of precise mathemtatical problems. We hope that future researchers can  help clarify in a rigorous fashion the correspondence between dynamical and computational notions of complexity.

\subsection*{Acknowledgements}
We thank Will Allen, Peter G\'{a}cs, Boris Hasselblatt, Felipe Hern\'{a}ndez, and Jensen Suther for valuable discussions.

\appendix

\section{Real computation and BSS$_{\textsf{C}}$ machines}
\label{App:BSSC}

In~\cite{blum1989theory} Blum, Shub and Smale introduced what is now known as the BSS machine as a model of computation over an arbitrary ring $R$, notably including the reals $R = \mathbb{R}$.  In the setting where $R$ is a finite field, a BSS machine reduces to the usual paradigm of Turing machines.  A detailed exposition is given in the original paper~\cite{blum1989theory} as well as the book~\cite{blum1998complexity}, where BSS machines are presented in terms of computational graphs, although some other equivalent formulations are given.  For our purposes, we provide an equivalent formulation that emphasizes the connection with ordinary Turing machines.  To build up the definition, we require some prelimary definitions from real algebraic geometry and the theory of semialgebraic sets (see e.g.~\cite{bochnak2013real, coste2000introduction}).

First let us define semialgebraic subsets of $R^n$.  Since we will consider $R$ to be a finite field, $\mathbb{R}$, or Cartesian products thereof, we can suppose that $R$ is a real closed field with a canonical ordering.  We have the following definition.
\begin{definition}[Semialgebraic subsets of $R^n$, adapted from~\cite{bochnak2013real}]\label{def:semialg1}
A \textbf{semialgebraic subset} of $R^n$ is a subset of the form
\begin{align}
\bigcup_{i=1}^s \bigcap_{j=1}^{r_i} \{x \in R^n \,:\,f_{i,j}(x) \triangleright_{i,j} 0\}
\end{align}
where $f_{i,j} \in R[X_1,...,X_n]$ and $\triangleright_{i,j}$ is either $>$ or $=$, for $i = 1,...,s$ and $j = 1,...,r_i$.
\end{definition}
\noindent Having defined a semialgebraic subset, we can now define a semialgebraic function.
\begin{definition}[Semialgebraic function, adapted from~\cite{bochnak2013real, coste2000introduction}]\label{def:semialgfn1} Let $A \subseteq R^m$ and $B \subseteq R^n$ be two semialgebraic sets.  A mapping $f : A \to B$ is a \textbf{semialgebraic function} if its graph
\begin{align}
\Gamma_f = \{(x,y) \in A \times B\, : \, y = f(x)\}
\end{align}
is a semialgebraic subset of $R^m \times R^n$.
\end{definition}
\noindent While Definitions~\ref{def:semialg1} and~\ref{def:semialgfn1} may be slightly hard to parse upon an initial glance, they have a simple interpretation as remarked below.
\begin{remark}
A semialgebraic subset of $R^n$ is one that is carved out be a finite boolean combination of polynomial equalities and inequalities.  Similarly, a semialgebraic function is a function that is built out of a finite boolean combination of addition, multiplication, and inequality comparisons.
\end{remark}
\noindent Moreover, we notice the notion of semialgebraic subsets of functions becomes trivial in the setting that $R = \mathbb{Z}_d$ for some $d$.  We capture this in the following remark.
\begin{remark}[Semialgebraic subsets and functions for $R = \mathbb{Z}_d$]
If $R = \mathbb{Z}_d$, all subsets of $\mathbb{Z}_d^n$ are semialgebraic and thus all functions $f : A \to B$ for $A \subseteq \mathbb{Z}_d^m$ and $B \subseteq \mathbb{Z}_d^n$ are semialgebraic functions.  That is, the semialgebraic conditions put no constraints on subsets or functions in the setting of finite fields.
\end{remark}

Before proceeding, it is useful to consider a refinements of Definition~\ref{def:semialg1}.
% and~\ref{def:semialgfn1}.
\begin{definition}[$R'$--semialgebraic subsets of $R^n$]\label{def:semialg2}
Let $R' \subseteq R$ by a real closed subfield of $R$ with an ordering induced by that of $R$. 
Then an \textbf{$R'$--semialgebraic subset} of $R^n$ is a subset of the form
\begin{align}
\bigcup_{i=1}^s \bigcap_{j=1}^{r_i} \{x \in R^n \,:\,f_{i,j}(x) \triangleright_{i,j} 0\}
\end{align}
where $f_{i,j} \in R'[X_1,...,X_n]$ and $\triangleright_{i,j}$ is either $>$ or $=$, for $i = 1,...,s$ and $j = 1,...,r_i$.
\end{definition}
%\noindent We similarly have:
%\begin{definition}[$R'$--semialgebraic function]\label{def:semialgfn2} Let $A \subseteq R^m$ and $B \subseteq R^n$ be two semialgebraic sets.  A mapping $f : A \to B$ is an \textbf{$R'$--semialgebraic function} if its graph is an $R'$--semialgebraic subset of $R^m \times R^n$.
%\end{definition}

We are now equipped to define BSS machines and some useful generalizations thereof. Let us first define the BSS analog of a Turing machine over $R$, akin to Definition~\ref{def:Turingmachine1}.
\begin{definition}[$R$--Turing machine]\label{def:RTuringmachine1} An \textbf{$R$--Turing machine} is given by a triple $(Q, \Gamma, \delta)$ where $Q = R^m$ and $\Gamma = R$ and:
\begin{enumerate}
    \item $Q$ is the set of states, containing a start state labeled $q_0$ and at least one halt state labeled $q_{\text{\rm halt}}$;
    \item $\Gamma$ is the tape alphabet; and
    \item $\delta : Q \times (\Gamma \times \mathbb{Z}_2)\to Q \times (\Gamma \times \mathbb{Z}_2)
    \times \{\text{\rm L}, \text{\rm R}, \text{\rm S}\}$ is a semialgebraic function, interpreted as one from $R^m \times (R \times \mathbb{Z}_2) \to R^m \times (R \times \mathbb{Z}_2) \times \mathbb{Z}_3$.  We identify $\mathbb{Z}_2 \simeq \{0,1\}$ with $\{\sqcup, 1\}$ and call $\sqcup$ the blank symbol.  We additionally require that $\delta(q, \gamma, \sqcup) = \delta(q', \gamma', 1)$ and $\delta(q, \gamma, 1) = \delta(q'', \gamma'', 1)$ for all $q, \gamma$.  In other words, transitions always take blank symbols to non-blank symbols, and non-blank symbols to non-blank symbols.
\end{enumerate}
The configuration space $S$ of an $R$--Turing machine is given by $S := \Gamma^* \times Q \times \Gamma^*$, where a configuration is denoted by $s = x_1 \cdots x_{m-1} \,q\,x_m \cdots x_n$ for $x_i \in \Gamma \times \mathbb{Z}_2$ and $q \in Q$, with the understanding that all symbols to the left of $x_1$ are equal to $(\gamma_0, \sqcup)$ for some fixed $\gamma_0$, and all symbols to the right of $x_n$ are equal to $(\gamma_0, \sqcup)$ for the same fixed $\gamma_0$.  Our notation expresses that the head is above $x_m$, and that all of the symbols on the tape to the left of $x_1$ and to the right of $x_n$ are $(\gamma_0, \sqcup)$ which is `blank'.  We can define a map $\textnormal{\textsf{T}}_R : S \to S$ as follows.  If $\delta(q, x_m) = (q', x_m', a)$ for $a \in \{\text{\rm L},\text{\rm R}, \text{\rm S}\}$, then
\begin{align}
\textnormal{\textsf{T}}_R(s)
 = \begin{cases}
x_1 \cdots x_{m-2} \,q'\, x_{m-1} x_{m}' \cdots x_n &\text{if }a = \text{\rm L} \\
x_1 \cdots x_{m-1} \,q'\, x_{m}' \cdots x_n &\text{if }a = \text{\rm S} \\
x_1 \cdots x_m' \,q'\, x_{m+1} \cdots x_n &\text{if }a = \text{\rm R}
\end{cases},
\end{align}
which describes one time step of the $R$--Turing machine.
\end{definition}
\noindent The above immediately generalizes to the $k$-tape setting.  We make the following important remark about terminology:
\begin{remark}[BSS machine]
Definition~\ref{def:RTuringmachine1} is different than the definition of a BSS machine over $R$ in e.g.~\cite{blum1989theory, blum1998complexity}.  However, our $R$--Turing machine model is computationally equivalent to a BSS machine over $R$.  As such, we will sometimes refer to an $R$--Turing machine as a BSS machine over $R$.
\end{remark}
\noindent Moreover, we note that Definition~\ref{def:RTuringmachine1} reduces to the ordinary Turing machine setting of Definition~\ref{def:Turingmachine1} when $R$ is a finite field.

Considering the setting of $R$--Turing machines for $R = \mathbb{R}$, there is a notable feature of Definition~\ref{def:RTuringmachine1}.  Since $\delta : \mathbb{R}^m \times (\mathbb{R} \times \mathbb{Z}_2) \to \mathbb{R}^m \times (\mathbb{R} \times \mathbb{Z}_2) \times \mathbb{Z}_3$ is a semialgebraic function, it entails polynomials and polynomial inequalities involving a finite number of arbitrary real numbers (e.g.~the coefficients of the polynomials).  As emphasized in~\cite{braverman2005complexity}, BSS machines over $\mathbb{R}$ can have certain pathological features stemming from the aforementioned arbitrary real constants.  As a remedy, Braverman suggests constraining those real numbers to be computable.  Let us make this precise in our language.  First, we recall the definition of computable real numbers.
\begin{definition}[Computable real numbers~\textsf{C}]
The \textbf{computable real numbers} \textnormal{\textsf{C}} are real numbers $x \in \mathbb{R}$ such that for each $x$, there exists a ($\mathbb{Z}_2$--)Turing machine $\textnormal{\textsf{T}}_x$ with $\Gamma = \{0,1\}$ that has the following property.  For each $n \in \mathbb{N}$, if the Turing machine is given the initial configuration $q_{0}\,[n]_2$ where $[n]_2$ is the representation of $n$ in binary, then the Turing machine halts with the configuration $q_{\text{\rm halt}}\,[a(n)]_2$ where $\frac{a(n) - 1}{n} \leq x \leq \frac{a(n) + 1}{n}$.
\end{definition}
\noindent We note that $\textsf{C}$ is a real closed subfield of $\mathbb{R}$ with the ordering induced by $\mathbb{R}$.  As such, we have the following definition:
\begin{definition}[BSS$_{\textsf{C}}$ machine, after~\cite{braverman2005complexity}]\label{def:BSSC1} A \textbf{BSS$_\textnormal{\textsf{C}}$ machine} is an $\mathbb{R}$--Turing machine $(Q, \Gamma, \delta)$ such that the graph of $\delta$ is a $(\textnormal{\textsf{C}}^m \times (\textnormal{\textsf{C}} \times \mathbb{Z}_2) \times \textnormal{\textsf{C}}^m \times (\textnormal{\textsf{C}} \times \mathbb{Z}_2) \times \mathbb{Z}_3)$--semialgebraic subset of $\mathbb{R}^m \times (\mathbb{R} \times \mathbb{Z}_2) \times \mathbb{R}^m \times (\mathbb{R} \times \mathbb{Z}_2) \times \mathbb{Z}_3$, and the distinguished symbol $\gamma_0$ is an element of $\textnormal{\textsf{C}}$.  This is to say that the (finite number of) real numbers entailed in the definition $\delta$ are all computable.
\end{definition}
\noindent In the above Definition, we can canonically take $\gamma_0 := 0$.  We use the Definition of a BSS$_{\textsf{C}}$ machine extensively throughout the paper.

Another useful definition is the one below.
\begin{definition}[Finite state BSS$_{\textsf{C}}$ machine]\label{def:FSBSSC1}
A \textbf{finite state BSS$_\textnormal{\textsf{C}}$ machine} is a map $f : M_1 \times M_2 \to M_2$ where $M_1 \subseteq \mathbb{R}^{n_1}$, $M_2 \subseteq \mathbb{R}^{n_2}$, and such that the graph of $f$ is a ($M_1|_{\textnormal{\textsf{C}}} \times M_2|_{\textnormal{\textsf{C}}} \times M_2|_{\textnormal{\textsf{C}}}$)--semialgebraic subset of $M_1 \times M_2 \times M_2$.
Here $M_1|_{\textnormal{\textsf{C}}}$ and $M_2|_{\textnormal{\textsf{C}}}$ mean the restriction of $M_1$ and $M_2$ to their computable points.
\end{definition}
\noindent Note that we can think of an ordinary BSS$_{\textsf{C}}$ as a Turing machine with an $\mathbb{R}$--valued tape, where the `head' is a type of finite state BSS$_{\textsf{C}}$ machine.

\section{Complexity classes and dynamical systems}
\label{App:compclass}

In this appendix we explore the relationship between complexity classes and computational dynamical systems.  We focus on deterministic time and space complexity, and do not discuss the non-deterministic setting here.  First we review some standard complexity classes along the lines of~\cite{arora2009computational}.  Recall that a language $L$ over a finite alphabet $\Gamma$ is a subset $L \subseteq \Gamma^*$.  Then we have the following definition.

\begin{definition}[Turing machine deciding a language]
Consider a (single tape) Turing machine $(Q, \Gamma, \delta)$ with two halt states called $q_{\text{\rm accept}}$ and $q_{\text{\rm reject}}$.  We say that the Turing machine \textbf{decides} a language $L \subseteq \Gamma^*$ if, when initiated in the configuration $q_0\,x$, the Turing machine halts with the head in the state $q_{\text{\rm accept}}$ if and only if $x \in L$.  We further say that a Turing machine which decides $L$ has \textbf{runtime} $O(T(n))$ if for each $x \in L$ such that $|x| = n$, the Turing machine halts in time $O(T(n))$.
\end{definition}

\noindent With the above definition, we can define the following complexity classes.

\begin{definition}[\textbf{DTIME} complexity classes, adapted from~\cite{arora2009computational}]
\label{def:DTIME1} If $T : \mathbb{N} \to \mathbb{N}$ is a monotonically increasing function, we say that $L$ is in $\textnormal{\textbf{DTIME}}(T(n))$ if there is a Turing machine that decides $L$ with runtime $O(T(n))$.
\end{definition}

\noindent While in some cases it is natural to consider e.g.~\textbf{DTIME}$(n)$ or \textbf{DTIME}$(n^2)$, it is often most useful to consider the class of languages that is decidable by a Turing machine with any polynomial runtime.  To this end, we define the following:

\begin{definition}[\textbf{P} complexity class]
\label{def:P1}
We define $\textnormal{\textbf{P}} := \bigcup_{k \in \mathbb{N}} \textnormal{\textbf{DTIME}}(n^k)$.
\end{definition}
\noindent 

\noindent This complexity class \textbf{P} is one of the central objects of study in computational complexity theory.  It is also useful to define:
\begin{definition}[\textbf{EXPTIME} complexity class]
\label{def:EXPTIME1}
We define $\textnormal{\textbf{EXPTIME}} := \bigcup_{k \in \mathbb{N}} \textnormal{\textbf{DTIME}}(2^{n^k})$.
\end{definition}

In addition to considering time-bounded complexity, we can also consider space-bounded complexity.  As such, a spatial analog of Definition~\ref{def:DTIME1} is:

\begin{definition}[\textbf{DTIME} complexity classes]
\label{def:DSPACE1} Let $S : \mathbb{N} \to \mathbb{N}$ be a monotonically increasing function. We say that $L$ is in $\textnormal{\textbf{SPACE}}(S(n))$ if there is a Turing machine that decides $L$ for which the following condition holds.  For each $x \in L$ with $|x| = n$, when the Turing machine is initiated in the configuration $q_0\,x$, it halts with its head in the state $q_{\text{\rm accept}}$ with the head visiting $O(S(n))$ cells of the tape during the computation.
\end{definition}

\noindent Then the spatial analogs of Definitions~\ref{def:P1} and~\ref{def:EXPTIME1} are:

\begin{definition}[\textbf{PSPACE} complexity class, adapted from~\cite{arora2009computational}]
\label{def:PSPACE1}
We define $\textnormal{\textbf{PSPACE}} := \bigcup_{k \in \mathbb{N}} \textnormal{\textbf{SPACE}}(n^k)$.
\end{definition}

\begin{definition}[\textbf{EXPSPACE} complexity class]
\label{def:EXPSPACE1}
We define $\textnormal{\textbf{EXPSPACE}} := \bigcup_{k \in \mathbb{N}} \textnormal{\textbf{SPACE}}(2^{n^k})$.
\end{definition}

\noindent There are standard relationships between the time and space complexity classes (see e.g.~\cite{sipser2012introduction, arora2009computational}), but we will not review them here.

Now we turn to dynamical systems.  Let us consider e.g.~differentiable maps $f : M \to M$ for $M$ some fixed manifold $M$.  Then we have the following definitions.

\begin{definition}[Simulation by a collection of dynamical systems]
Let $\textnormal{\textbf{C}}$ be some collection of languages and $\textnormal{\texttt{D}}$ be some collection of differentiable dynamical systems $f : M \to M$ for some fixed $M$.  Then we say that $\textnormal{\textbf{C}}$ is simulated by $\textnormal{\texttt{D}}$ with $O(t(n))$--complexity if for each language $L \in \textnormal{\textbf{C}}$, there exists a dynamical system in $\textnormal{\texttt{D}}$ which extends to a $O(t(n))$--complexity CDS of a Turing machine that decides $L$.
\end{definition}
\noindent The above Definition allows us to formulate certain interesting mathematical statements in a compact manner. For instance, let $\textbf{R}$ be the set of recursive languages, i.e.~the set of languages decidable by a Turing machine.  Moreover, let $\texttt{DiffDisk}$ be the set of diffeomorphisms of the disk.  Then by Theorem~\ref{thm:TuringUniversalCDS1}, we find that $\textbf{R}$ is simulated by $\texttt{DiffDisk}$ with $\Theta(n)$ complexity.  We can also formulate questions like:
\begin{question}
Let $\texttt{ErgodicDisk}$ be the set of ergodic diffeomorphisms on the disk.  Then is $\textbf{P}$ or even $\textbf{EXPTIME}$ simulated by $\texttt{ErgodicDisk}$ with $\Theta(n)$ complexity?
\end{question}
\noindent Or similarly we can ask:
\begin{question}
Is $\textbf{PSPACE}$ or even $\textbf{EXPSPACE}$ simulated by $\texttt{ErgodicDisk}$ with $\Theta(n)$ complexity?
\end{question}
\noindent We anticipate that there are many other interesting questions along these lines.

\section{Symbolic dynamics}
\label{app:symbolic-dynamics}
In this appendix we give a review of symbolic dynamics.  An excellent textbook treatment can be found in \cite{kitchens2012symbolic}; see also \cite{kato2015global}.

Let $\Sigma_1$ be a finite alphabet (distinct from the alphabet $\Sigma$ associated to any CDS with underyling dynamical system $f$). The set $\Sigma_1^\Z$ of bi-infinite strings of symbols from $\Sigma_1$ has a natural topology that makes it homeomorphic to the Cantor set: one declares a sub-basis of open sets to be the collection $U_s$, where $s$ ranges over $\Sigma^*$ of even length, such that writing $s = s_{-m}\cdots s_{-1}s_0s_1\cdots s_{m-1}$ (where the length of $s$ is $2m$) one has that 
\[ U_s = \{s' \in \Sigma_1^\Z \,:\,s'_i = s_i \,\text{ for }\, i=-m, \ldots, m-1\}.\]
The shift map 
\[ \sigma: \Sigma_1^\Z \to \Sigma_1^\Z,\quad \sigma(s)_i = \sigma(s)_{i+1}\,,\]
is a continuous homeomorphism of $\Sigma_1^\Z$. By the previous description of the topology on $\Sigma^1_\Z$, one sees that closed $\sigma$-invariant subsets of $\Sigma^1_\Z$ correspond to those sets of bi-infinite strings which do not contain some collection of finite substrings.

A closed $\sigma$-invariant subset $F \subset \Sigma_1^\Z$ such that $F = U_{s_1}^c \cap \cdots \cap U_{s_k}^c$, i.e. the set of strings which do not contain any of a \emph{finite} list of banned substrings, is called a \emph{shift of finite type.}

Now, a map of shift spaces 
\[ \Sigma_1^\Z \supset F_1 \longrightarrow F_2 \subset \Sigma_2^\Z\]
is simply a continuous map $\phi: F_1 \to F_2$ that commutes with the shift operations on each side. It is a theorem \cite[Theorem 1.4.9]{kitchens2012symbolic} that a map of shift spaces is defined by a \emph{sliding block code}: that is, $\phi$ takes the form 
\[ \phi(s)_i = \bar{\phi}(s_{i-n}, \ldots, s_{i+m}) \,\text{ for some }\, \bar{\phi}: \Sigma_1^{n+m+1} \longrightarrow \Sigma_2\,.\]
A map of shift spaces is a \emph{factor map} if it is surjective. 

Let $F = U_{s_1}^c \cap \cdots \cap U_{s_k}^c$ be a shift of finite type. Supposing $N = \max_j |s_j|$, if we know the last $N-1$ symbols of a substring of $s$ then we know what the allowed possible current symbols are. Thus, using a sliding block code which simply introduces a new symbol corresponding to every element of $\Sigma_1^N$, we can encode the data of $F$ into a graph, with vertices given by allowed words of length $N$ in $F$, and edges given by allowed transitions (corresponding to forgetting the earliest symbol and adding an allowed current symbol). We can think of the vertices of this graph as states and the edges of this graph as transitions between states. This discussion shows that using the sliding block code above, one can find an isomorphism from $F$ to another shift of finite type $F_0 = U_{s'_1}^c \cap \cdots \cap U_{s'_k}^c$ with $|s'_j|=1$ for all $j$. 

Using this representation, we see that we can think of a shift space $F_{\text{FSM}}$ (where `FSM' stands for `finite state machine') which is a \emph{factor} of a shift of finite type as the collection of strings outputted by walks on such a graph, where we compute the output string via a sliding block code on the sequence of vertices visited on the walk.  In other words, the symbol output at time $t$ corresponds to the vertices visited at times $t-n, \ldots, t+m$. By building a new graph whose vertices are labeled by sequences of $n+m+1$ vertices of the previous graph and whose transitions correspond to forgetting the earliest vertex and going to a valid subsequent vertex, we see that this set $F_{\text{FSM}}$ is exactly the set of walks on a finite graph with certain symbols labeling the vertices where the \emph{same symbol} may label \emph{multiple vertices}. Thinking of the vertices as states and the symbols labeling the vertices as `outputs', we see that we have exactly the set of possible outputs of a finite state machine with no halting state. Thus factors of shifts of finite type, also known as `sofic systems', correspond to collections of strings which can be output by a nondeterministic finite state machine without a distinguished halting state.

This review of symbolic dynamics should convince the reader that factors of shifts of finite type are a convenient description of a finite state machine. Now, given an Axiom A system $f$, a theorem of Bowen gives us an analogous understanding of the dynamics of $f$ on each basic set $\Omega_i$ in terms of an associated shift of finite type, due to the existence of \emph{Markov partitions} for $f$. This is a fundamental result of Bowen which we we state below:

\begin{theorem}[Bowen \cite{bowen1970markov}]
Let $C$ be a hyperbolic set for $f$. 
There exists a shift of finite type $F$ on the alphabet $\Sigma_1$, and a continuous map
\[ \pi: F \to C \]
such that:
\begin{enumerate}
    \item $\pi$ is surjective;
    \item $\pi$ is finite-to-one, with the sizes of preimages being finite; and
    \item The map $\pi$ is defined as follows: there are a finite collection of \emph{rectangles} $R_{a_1}, \ldots R_{a_k}$ with $\Sigma_1 = \{a_1, \ldots, a_k\}$, which are closed subsets of $\Omega$ with nontrivial, nonoverlapping interiors, satisfying some axioms (those of a \emph{Markov partition}, see \cite{bowen1970markov}). Then $\pi$ is defined to be the map 
    \[ s \mapsto \pi(x), \,\text{ where }\, \{\pi(x)\} = \bigcap_{i \in \Z} f^{k}(R_{s_{-k}})\,.\]
    Moreover, on the set of $x \in \Omega_i$ such that $f^k(x)$ always lies in the interior of some rectangle, $\pi^{-1}(x)$ consists of a single element.
\end{enumerate}

\end{theorem}

\noindent Thus we see that although $f$ is a continuous dynamical system rather than a shift space, its dynamics are in a precise sense a factor of a shift of finite type.

\begin{remark}
The combination of the spectral decomposition and Bowen's result on Markov partitions for $f$ suggests immediately that $f$ behaves like a non-deterministic finite state machine. One of our motivations for this work was to make this mathematical result have a more precise interpretation from a computational perspective. Part of the challenge is that the sets $R_j$ of the Markov partition are generally \emph{not} computable subsets. 

Indeed, if there is only one basic set and it is equal to all of $M$, we say that $f$ is \emph{Anosov}. Already in the relatively simple case of linear Anosov diffeomorphisms of higher-dimensional tori, one can see that the Markov partition $R_j$ does not consist of computable subsets \cite{bowen1978markov}.  See Problem  \ref{problem:anosov-bound-on-halting-time} above for questions about the computational complexity of $f$ which existing dynamical results on Axiom A systems do not immediately resolve. 
\end{remark}

\section{Variations of the stable manifold theorem}
\label{App:stable-manifold-theorem-variants}
In this technical appendix we derive the variant of the stable manifold theorem for hyperbolic sets that we use (Proposition \ref{prop:stable-manifold-theorem}) from the result of Hirsch-Pugh \cite{hirsch-pugh}. We will freely use the statement of Theorem 3.2 of \cite{hirsch-pugh}. First, with notation $W_x$ as in Hirsch-Pugh, note that  $x \in W_x$, so \cite[Theorem 3.2(c)]{hirsch-pugh} implies that $W_x$ contains $W^s(x) \cap B(x, r_x)$ for some $r_x$, where $W^s(x)$ is defined as in our paper.  Since the $W_x$ are a continuous family of $C^k$ submanifolds (in the sense of \cite{hirsch-pugh}), the numbers $r_x$ can be chosen to vary continuously with $x$. By compactness of the hyperbolic set, this means that $r_x$ achieves a minimum $r_{\min}$. But by \cite[Theorem 3.2(c)]{hirsch-pugh} together with the fact that $f(W^s(x)) = W^s(f(x))$, this implies that $f^{-n}(W_{f^n(x)}) \subset W^s(x)$ contains $W^s(x) \cap B(x, r_{\min}/(K\lambda)^n)$ for some $K > 0$ and $\lambda < 1$. In particular we have that $\bigcup_n f^{-n}(W_{f^n(x)}) = W^s(x)$. Moreover, the same argument lets us generalize \cite[Theorem 3.2(b)]{hirsch-pugh} to our notion of the continuity of the manifolds $W^s(x)$. Indeed, $f^n(z) \in W_{f^n(x)}$ for all sufficiently large $n$ by Theorem \cite[Theorem 3.2(c)]{hirsch-pugh}, so we can take our $\phi: U_x \to C^\infty(D^r, M)$ to be $f^{-n} \circ \phi' \circ f^n$ with $\phi': U_{f^n(x)} \to C^\infty(D^r, M)$ the map produced by \cite[Theorem 3.2(b)]{hirsch-pugh}.

\section{Proof of polynomial root separation bound}
\label{App:root-separation}

In this appendix we give a self-contained treatment of the polynomial root separation bound used in the proof of Theorem~\ref{thm:Anosov1}.  Our desired bound is essentially an adaptation Theorem 3 of~\cite{rump1979polynomial}, and we will follow that proof closely.

Let us begin with some definitions which allow us to state the results and proofs.  Denote by $P(x)$ a polynomial in $\mathbb{C}[x]$ of degree $n > 0$ of the form
\begin{align}
P(x) = \sum_{k = 0}^n a_k\,x^k = a_n \prod_{k = 1}^n (x - \lambda_k)
\end{align}
where $a_n \not = 0$.  Above, the $\lambda_k$'s are the roots of $P$.  Then we define the minimal root separation as follows.
\begin{definition}[Minimal root separation]
We denote by $\text{\rm rsep}(P)$ the minimal root separation of $P$, namely
\begin{align}
\text{\rm rsep}(P) := \min\{|\lambda_i - \lambda_j|\,\,\,\text{\rm for real }\,\lambda_i \not = \lambda_j\}\,.
\end{align}
We further denote
\begin{align}
\text{\rm rsep}_{[-1,1]}(P) := \min\{|\lambda_i - \lambda_j|\,\,\,\text{\rm for real }\,\lambda_i \not = \lambda_j\,\,\,\text{\rm and }\,|\lambda_i|, |\lambda_j| \leq 1\}\,.
\end{align}

\end{definition}
\noindent We will also need to define a $p$-seminorm on polynomials.
\begin{definition}[$p$-seminorm on polynomials] For $P \in \mathbb{C}[x]$, we define the $p$-seminorm $|P|_p := \left(\sum_{k = 0}^n |a_k|^{p}\right)^{1/p}$. Additionally consider $Q \in \mathbb{C}[x,y]$, and suppose $Q$ takes the form $Q(x,y) = \sum_{k = 0}^n \sum_{\ell = 0}^m a_{k,\ell}\,x^k y^\ell$.  Then we further define the $p$-seminorm $|Q|_{p} := \left(\sum_{k = 0}^n \sum_{\ell = 0}^m |a_{k,\ell}|^p\right)^{1/p}$.
\end{definition}
\begin{remark}
Note that $\mathbb{C}[x], \mathbb{C}[y] \subset \mathbb{C}[x,y]$, and as that the $p$-seminorm on $\mathbb{C}[x,y]$ induces the $p$-seminorm on $\mathbb{C}[x]$ and $\mathbb{C}[y]$.  As such, we not distinguish between these seminorms.
\end{remark}
\noindent We also use the following definitions to simplify notation.
\begin{definition}[Standardized polynomial]
We say that a polynomial $P$ is \textbf{standardized} if all of its non-zero coefficients are greater than or equal to $1$. For any $P(x)$, we let
\begin{align}
\widetilde{P} := P/\min\{1,\min\{|a_i|\, : \, a_i \not = 0\}\}\,,
\end{align}
which is standardized.  Indeed, if $P$ is already standardized, then $\widetilde{P} = P$.
\end{definition}
\begin{definition}[Size of a polynomial]
We define the \textbf{size} of $P$ to be $s(P) := |P|_1$.  We further let $\tilde{s}(P) := s(\widetilde{P})$, noting that $\tilde{s}(P) \geq s(P)$.    
\end{definition}
\noindent Finally, we also need to define the resultant of two polynomials.
\begin{definition}[Resultant of two polynomials]
Let $P,Q \in \mathbb{C}[x]$ with $P(x) = \sum_{k = 0}^n a_k\,x^k = a_n \prod_{k = 1}^n (x - \lambda_k)$ and $Q(x) = \sum_{\ell = 0}^m b_\ell\,x^\ell = b_m \prod_{\ell = 1}^m (x - \mu_\ell)$ with $a_n, b_m \not = 0$.  Then the \textbf{resultant} of $P$ and $Q$ is
\begin{align}
\text{\rm res}(P,Q) := a_n^m b_m^n \prod_{k = 1}^n \prod_{\ell = 1}^m (\lambda_k - \mu_\ell)\,.
\end{align}
The resultant has the equivalent expressions $\text{\rm res}(P,Q) = a_n^m \prod_{k = 1}^n Q(\lambda_k) = (-1)^{mn} b_m^n \prod_{\ell = 1}^m P(\mu_\ell)$.
\end{definition}

With the above definitions at hand, we can state the desired bound.
\begin{theorem}[Adapted from Theorem 3 of~\cite{rump1979polynomial}]\label{thm:mainrootbound1}
Let $P \in \mathbb{C}[x]$ have degree $n$.  Then we have the root separation bound
\begin{align}
\label{E:mainrootbound1}
\text{\rm rsep}_{[-1,1]}(P) \geq \frac{2\sqrt{2}}{n^{n/2 + 1} (\tilde{s}(P) + 1)^n}\,.
\end{align}
\end{theorem}
\noindent To establish this Theorem, we require a number of lemmas.  We begin with a generalization of Hadamard's bound.
\begin{lemma}[Adapted from Theorem 1 of~\cite{collins1974minimum}]\label{lemma:Hadamard1} Let $M$ be an $N \times N$ matrix with entries valued in $\mathbb{C}[x,y]$ (or just $\mathbb{C}[x]$ or $\mathbb{C}[y]$). If $M_j$ denotes the $j$th row of $M$, then we have the bound
\begin{align}
\label{E:Hadamard1}
|\det(M)|_1 \leq \prod_{j = 1}^N |M_j|_1\,.
\end{align}
\end{lemma}
\begin{proof}
We proceed with a proof by induction, noting that the $N = 1$ case holds automatically.  Supposing~\eqref{E:Hadamard1} holds for $N$, let us establish the bound for $N+1$.  Letting $M_{ij}'$ be the submatrix of $M$ with row $i$ and column $j$ removed, we have $\det(M) = \sum_{i = 1}^{N+1} (-1)^{i+1} M_{i1}\,\det(M_{i1}')$ and so
\begin{align}
|\det(M)|_1 \leq \sum_{i = 1}^{N+1}|M_{i1}|_1\,|\det(M_{i1}')|_1 \leq \prod_{j = 2}^{N+1} |M_j|_1 \sum_{i = 1}^{N+1}|M_{i1}|_1\,.
\end{align}
In the first inequality we have used Cauchy-Schwarz, and in the second inequality we applied~\eqref{E:Hadamard1} to $\det(M_{i1}')$.  Since $\sum_{i = 1}^{N+1}|M_{i1}|_1  = |M_1|_1$, we have established the bound for $N+1$, which completes the proof.
\end{proof}
\noindent Next we use the bound in the previous Lemma to bound a particular resultant of two polynomials.
\begin{lemma}[Adapted from Theorem 2 of~\cite{collins1974minimum}]\label{lemma:Hadamard2}
Let $R(y) := \text{\rm res}(Q(x), y - P(x))$ be the resultant with respect to $x$, so that $R(y)$ is a polynomial in $\mathbb{C}[y]$.  Then we have the inequality
\begin{align}
|R|_1 \leq (|P|_1 + 1)^m |Q|_1^n\,.
\end{align}
% where $|\,\cdot\,|_1$ on the left-hand side means the $1$-seminorm on $\mathbb{C}[y]$, and $|\,\cdot\,|_\infty$ on the right-hand side means the $\infty$-seminorm on $\mathbb{C}[x]$.
\end{lemma}
\begin{proof}
The resultant $R(y) := \text{\rm res}(Q(x), y - P(x))$ with respect to $x$, namely $R(y) = b_{m}^n\prod_{\ell = 1}^m (y - P(\mu_\ell))$, can be expressed as the determinant of the $(m + n) \times (m + n)$ matrix
\begin{align}
\label{E:bigmatrix1}
M = \begin{bmatrix}
b_m & 0 & \cdots & 0 & -a_n & 0 & \cdots & 0 \\
b_{m-1} & b_m & \cdots & 0 & -a_{n-1} & -a_n & \cdots & 0 \\
b_{m-2} & b_{m-1} & \ddots & 0 & -a_{n-2} & -a_{n-1} & \ddots & 0 \\
\vdots & \vdots & \ddots & b_m & \vdots & \vdots & \ddots & -a_n \\
b_0 & b_1 & \cdots &\vdots & y-a_0 & -a_1 & \cdots & \vdots \\
0 & b_0 & \ddots & \vdots & 0 & y-a_0 & \ddots & \vdots \\
\vdots & \vdots & \ddots & b_1 & \vdots & \vdots & \ddots & -a_1 \\
0 & 0 & \cdots & b_0 & 0 & 0 & \cdots & y - a_0
\end{bmatrix}
\end{align}
with entries values in $\mathbb{C}[y]$.  That is, $R(y) = \det(M)$.  Let $M_j$ denote the $j$th column of $M$.  Since Lemma~\ref{lemma:Hadamard1} gives $|\det(M)|_1 \leq \prod_{i=1}^{m+n} |M_i|_1$, examining~\eqref{E:bigmatrix1} we have
\begin{align}
|R|_1 = |\det(M)|_1 \leq  \prod_{j=1}^{m+n} |M_j|_1 = |y - P(x)|_1^m \,|Q|_1^n \leq (|P|_1 + 1)^m |Q|_1^n\,,
\end{align}
which gives the desired bound.
\end{proof}
\noindent Additionally, we need a classical result due to Cauchy.
\begin{lemma}[Cauchy's root bound, see e.g.~\cite{hirst1997bounding}]\label{lemma:Cauchyroot1}
If $\lambda$ is a root of $P \in \mathbb{C}[x]$, then $|\lambda| \leq \frac{|P|_\infty}{a_n} + 1$.
\end{lemma}
\begin{proof}
Since $\lambda$ is a root, we have $|a_n|\,|\lambda^n| = \left|\sum_{k=0}^{n-1} a_k \lambda^k\right|$.  If $|\lambda| \leq 1$, then
\begin{align}
\left|\sum_{k=0}^{n-1} a_k \lambda^k\right| \leq \max_k |a_k| \sum_{k=0}^{n-1} |\lambda^k| \leq \max_k |a_k| \frac{|\lambda|^n}{|\lambda| - 1}
\end{align}
and so altogether $|a_n| |\lambda^n| \leq \max_k |a_k| \frac{|\lambda|^n}{|\lambda| - 1}$ which implies $|\lambda| \leq \frac{|P|_\infty}{|a_n|} + 1$.  This inequality also holds for $|\lambda| \leq 1$ due to the $1$ in the right-hand side.
\end{proof}

The two previous Lemmas are useful to prove the following Lemma, which is central to the proof of Theorem~\ref{thm:mainrootbound1}:
\begin{lemma}[Adapted from Lemma 2 of~\cite{rump1979polynomial}]\label{lemma:rump2}
Let $\widetilde{P}$ and $\widetilde{Q}$ be standardized polynomials of degrees $n$ and $m$, respectively.  If for some $\mu$ we have $\widetilde{P}(\mu) \not = 0$ but $\widetilde{Q}(\mu) = 0$, then
\begin{align}
|\widetilde{P}(\mu)| \geq \{(|\widetilde{P}|_1 + 1)^n |\widetilde{Q}|_1^m + 1\}^{-1}.
\end{align}
\end{lemma}
\begin{proof}
As before, we let $R(y) := \text{res}(\widetilde{Q}(x), y - \widetilde{P}(x))$ be the resultant with respect to $x$.  By Lemma~\ref{lemma:Hadamard1} and the inequality between the $\infty$-seminorm and $1$-seminorm we have
\begin{align}
|R|_\infty \leq |R|_1 \leq (|\widetilde{P}|_1 + 1)^m |\widetilde{Q}_1|^n\,.
\end{align}
Writing $R(y) = \sum_{\ell = p}^m r_\ell\,y^\ell$ where $p = \argmin_\ell \{r_\ell\,:\,r_\ell \not = 0\}$, we can define the reciprocal polynomial $\bar{R}(y) := y^{n-p} R(y^{-1})$ whose roots are the inverses of the roots of $R(y)$.  Note also that $|R|_\infty = |\bar{R}|_\infty$.  For each root $\alpha = \widetilde{P}(\mu)$ of $R$, by Cauchy's root bound in Lemma~\ref{lemma:Cauchyroot1} we have
\begin{align}
|\alpha| \leq \frac{|R|_\infty}{|r_m|} + 1 \leq |R|_\infty + 1\,,
\end{align}
where in the last inequality we have used that $\widetilde{P}$ and $\widetilde{Q}$ are standardized.  So then for roots $\alpha^{-1}$ of $\bar{R}(y)$, we similarly have
\begin{align}
|\alpha^{-1}| = |\widetilde{P}(\mu)^{-1}| \leq \frac{|\bar{R}|_\infty}{|r_p|} + 1 \leq |\bar{R}|_\infty + 1 = |R|_\infty + 1\,,
\end{align}
and so using~\eqref{E:Hadamard1} we find
\begin{align}
|\widetilde{P}(\mu)^{-1}| \leq (|\widetilde{P}|_1 + 1)^n \, |\widetilde{Q}|_1^m + 1\,,
\end{align}
which implies our desired inequality.
\end{proof}
\noindent The above Lemma has the following corollary:
\begin{corollary}[Adapted from Lemma 3 of~\cite{rump1979polynomial}]\label{corr:rump3}
Let $\widetilde{P}$ be a standardized polynomial of degree $n \geq 2$.  If some $\gamma$ satisfies $P'(\gamma) = 0$ but $P(\gamma) \not = 0$, then
\begin{align}
|\widetilde{P}(\gamma)| \geq \{n^n (|\widetilde{P}|_1 + 1)^{2n-1}\}^{-1}\,.
\end{align}
\end{corollary}
\begin{proof}
We can use Lemma~\ref{lemma:rump2} with $\widetilde{Q} = \widetilde{P}'$, where we observe that if $\widetilde{P}$ is standardized then so is $\widetilde{Q} = \widetilde{P}'$.  Using $|\widetilde{P}'|_1 \leq n\cdot |\widetilde{P}|_1$, we find
\begin{align}
|\widetilde{P}(\gamma)| \geq \{n^n (|\widetilde{P}|_1^n + 1)^{n-1} \cdot |\widetilde{P}|_1^n + 1\}^{-1} \geq \{n^n (|\widetilde{P}|_1 + 1)^{2n-1}\}^{-1}\,,
\end{align}
which gives the stated bound.
\end{proof}

Finally, we can put together all of the above to prove Theorem~\ref{thm:mainrootbound1}.
\begin{proof}[Proof of Theorem~\ref{thm:mainrootbound1}, following Theorems 2 and 3 of~\cite{rump1979polynomial}]
Consider a polynomial $P$, and its standardized counterpart $\widetilde{P}$.  We will show that
\begin{align}
\label{E:mainrootbound2}
\text{rsep}_{[-1,1]}(\widetilde{P}) \geq \frac{2\sqrt{2}}{n^{n/2 + 1}(\tilde{s}(P) + 1)^n}\,.
\end{align}
Since $\text{rsep}_{[-1,1]}(P) = \text{rsep}_{[-1,1]}(\widetilde{P})$, the bound above in~\eqref{E:mainrootbound2} implies the desired bound~\eqref{E:mainrootbound1}.

Suppose that $\widetilde{P}(\alpha) = \widetilde{P}(\beta) = 0$ for $\alpha, \beta$ real, and moreover that $\text{rsep}(\widetilde{P}) = |\alpha - \beta|$.  Without loss of generality, we can take $-1 \leq \alpha < \beta \leq 1$.  By Rolle's Theorem there exists a real $\gamma$ with $-1 \leq \alpha < \gamma < \beta \leq 1$ such that $\widetilde{P}'(\gamma) = 0$.  Then we consider the expansion
\begin{align}
0 = \widetilde{P}(\beta) = \widetilde{P}(\gamma) + \frac{(\gamma - \beta)^2}{2}\,\widetilde{P}''(\omega)
\end{align}
for $\gamma < \omega < \beta$.  Using Lemma~\ref{lemma:rump2} we have
\begin{align}
(\gamma - \beta)^2\,|\widetilde{P}''(\omega)| = 2\,|\widetilde{P}(\gamma)| \geq 2\,\{n^n (|\widetilde{P}|_1 + 1)^{2n-1}\}^{-1}\,.
\end{align}
Since we have $|\omega| < 1$, it follows that
\begin{align}
|\widetilde{P}''(\omega)| \leq \left|\sum_{k=2}^n k(k-1) \,a_k\,\omega^{k-2}\right| \leq n^2\,|\widetilde{P}|_1
\end{align}
and so all together
\begin{align}
\label{E:gammabeta1}
(\gamma - \beta)^2 \geq 2\,\{n^{n+2} (|\widetilde{P}|_1 + 1)^{2n}\}^{-1}\,.
\end{align}
Repeating the same analysis using $\alpha$ in place of $\gamma$, we similarly find
\begin{align}
\label{E:alphagamma1}
(\alpha - \gamma)^2 \geq 2\,\{n^{n+2} (|\widetilde{P}|_1 + 1)^{2n}\}^{-1}\,,
\end{align}
and so combining~\eqref{E:gammabeta1} and~\eqref{E:alphagamma1} we arrive at the bound in~\eqref{E:mainrootbound2}.
\end{proof}

\bibliographystyle{alpha}
\bibliography{bibliography}

\end{document}